\documentclass[11pt]{amsart}
\usepackage{fullpage,graphicx,mathpazo,color}
\usepackage{amsmath,amscd,amssymb}
\usepackage{mathrsfs}
\usepackage{subfigure}
\usepackage{epstopdf}
\usepackage{amsmath}
\usepackage{tikz}
\usepackage{enumitem}
\allowdisplaybreaks[4]
\usepackage[colorlinks, citecolor=blue,pagebackref,hypertexnames=false]{hyperref}
\begin{document}
\newtheorem{assumption}{Assumption}[section]
\newtheorem{rhp}{Riemann--Hilbert Problem}[section]
\newtheorem{dbar}{$\bar{\partial}$-Problem}[section]
\newtheorem{dbarrhp}{$\bar{\partial}$-Riemann--Hilbert Problem}[section]
\newtheorem{theorem}{Theorem}[section]
\newtheorem{lemma}{Lemma}[section]
\newtheorem{proposition}{Proposition}[section]
\newtheorem{corollary}{Corollary}[section]
\newtheorem{definition}{Definition}[section]
\newtheorem{conjecture}{Conjecture}[section]
\newtheorem{remark}{Remark}[section]
\numberwithin{equation}{section}

\newcommand{\bfR}{{\mathbb R}}
\newcommand{\bfC}{{\mathbb C}}
\newcommand{\bfZ}{{\mathbb Z}}
\newcommand{\ii}{\text{i}}
\newcommand{\e}{\text{e}}
\newcommand{\dd}{\text{d}}
\newcommand{\nn}{\nonumber}
\newcommand\be{\begin{equation}}
\newcommand\ee{\end{equation}}
\newcommand{\bea}{\begin{eqnarray}}
\newcommand{\eea}{\end{eqnarray}}
\newcommand\berr{\begin{eqnarray*}}
\newcommand\eerr{\end{eqnarray*}}
\newcommand{\bfu}{\mathbf{u}}
\newcommand{\bfw}{\mathbf{w}}
\newcommand{\change}[1]{{\color{blue}#1}}
\newcommand{\note}[1]{{\color{red}#1}}

\title{ Soliton resolution for the coupled complex short pulse equation}

\author{Nan Liu$^*$}
\address{School of Mathematics and Statistics, Nanjing University
of Information Science and Technology, Nanjing 210044,  P.R. China}
\email{ln10475@163.com}

\author{Ran Wang}
\address{School of Mathematics and Statistics, Nanjing University
of Information Science and Technology, Nanjing 210044,  P.R. China}
\email{wran1102@163.com}


\begin{abstract}
We address the long-time asymptotics of the solution to the Cauchy problem of ccSP (coupled complex short pulse) equation on the line for decaying initial data that can support solitons. The ccSP system describes ultra-short pulse propagation in optical fibers, which is a completely integrable system and posses a $4\times4$ matrix Wadati--Konno--Ichikawa type Lax pair. Based on the $\bar{\partial}$-generalization of the Deift--Zhou steepest descent method, we obtain the long-time asymptotic approximations of the solution in two kinds of space-time regions under a new scale $(\zeta,t)$. The solution of the ccSP equation decays as a speed of $O(t^{-1})$ in the region $\zeta/t>\varepsilon$ with any $\varepsilon>0$; while in the region $\zeta/t<-\varepsilon$, the solution is depicted by the form of a multi-self-symmetric soliton/composite breather and $t^{-1/2}$ order term arises from self-symmetric soliton/composite breather-radiation interactions as well as an residual error order $O(t^{-1}\ln t)$.

\end{abstract}
\thanks{$^*$Corresponding author.}
\maketitle
\section{Introduction}
The NLS (nonlinear Schr\"odinger) equation plays a crucial role in the field of optical communications since it has been shown to be a universal model for the propagations of picosecond optical pulses in the single-mode nonlinear media. However, when one considers optical pulses whose width is of the order of the femtosecond ($10^{-15}$s) and
much smaller than the carrier frequency, the propagation of such ultra-short packets in nonlinear media is better described by the cSP (complex short pulse) equation \cite{F-PhysD,LFZ}
\be
u_{xt}=u+\frac{1}{2}(|u|^2u_x)_x,
\ee
where $u=u(x,t)$ is a complex function, which represents the electric field associated with the propagating optical pulse. To describe the propagation of optical pulses in birefringence fibers, two orthogonally polarized modes have to be considered, and in analogy to the Manakov system, a ccSP (coupled complex short pulse) equation
\begin{equation}\label{ccSPE}
\begin{aligned}
q_{1xt}=q_1+\frac{1}{2}\left[\left(|q_1|^2+|q_2|^2\right)q_{1x}\right]_{x},\\
q_{2xt}=q_2+\frac{1}{2}\left[\left(|q_1|^2+|q_2|^2\right)q_{2x}\right]_{x},
\end{aligned}
\end{equation}
was derived from Maxwell's equations in the literature \cite{F-PhysD}, where $q_j=q_j(x,t)$, $j=1,2$ are two complex functions representing the electric fields. In the original paper \cite{F-PhysD}, the author displayed the Lax pair, conservation laws and bright soliton solutions in pfaffians by virtue of Hirota's bilinear method.

In recent years, the ccSP equation has attracted considerable interest and been studied extensively due to its rich mathematical structure and remarkable properties.
Through the Hirota method, bright-dark one- and two-soliton solutions have been constructed in \cite{GW-WM}, and interactions between two bright or two dark solitons have been verified to be elastic through the asymptotic analysis. In \cite{KW-MMAS}, Lie symmetries and exact solutions of ccSP equation have been obtained. Moreover, the Darboux transformation for the Equation \eqref{ccSPE} was also constructed through the loop group method in \cite{FL-PhysD}, as a by-product, various exact solutions including bright-soliton, dark-soliton, breather and rogue wave solutions were obtained. In addition, soliton interactions and Yang-Baxter maps for the ccSP equation were also discussed in \cite{CGP}. The RH (Riemann--Hilbert) approach to the IST (inverse scattering transform) for the ccSP equation on the line with zero boundary conditions at space infinity was developed in \cite{LSL-IST} to solve the initial-value problem, and the long-time asymptotic behavior was analyzed in \cite{GLL-JDE} under the assumption that the initial conditions do not support solitons via the Deift--Zhou nonlinear steepest descent method \cite{PD}. Therefore, there is no systematic analysis for the asymptotics of ccSP equation in the presence of discrete spectrum, which motivates the present study.

The celebrated IST method developed by Gardner, Greene, Kruskal and Miura in \cite{GGKM} has been turned out to be very effective for solving the initial-value problems for a wide class of physically significant nonlinear partial differential equations. This robust approach allowed one to give a huge number of very interesting results in different areas of mathematics and physics. While the original IST was formulated in terms of the integral equations of Gel'fand--Levitan--Marchenko type, however, this method was subsequently rewritten as a RH factorization problem to study the various kinds of integrable nonlinear equations \cite{APT,YJK}. In particular, a great achievement in the further development of the IST method is the nonlinear steepest descent method for oscillatory matrix Riemann--Hilbert problems done by Deift and Zhou \cite{PD} based on earlier work \cite{Its,ZM}. This powerful method offers a systematic procedure for finding the asymptotics of integrable systems by reducing the original RH problem to a canonical model RH problem whose solution is calculated in terms of parabolic cylinder functions or Painlev\'e functions. This reduction is done through a sequence of transformations whose effects do not change the large-time behavior of the recovered solution at leading order. With this new method came the nice possibility to obtain numerous new significant asymptotic results in the theory of completely integrable nonlinear equations. Such equations include the NLS equation \cite{BLM,BV}, derivative NLS equation \cite{AL}, mKdV (modified Korteweg--de Vries) equation \cite{CL,HZ-SIAM}, Hirota equation \cite{GLW-NA}, SP/cSP equation \cite{BS-SP,XF-JDE}, CH (Camassa--Holm) equation \cite{AAD}, extended mKdV equation \cite{X-JDE} and so on, which are associated with $2\times2$ matrix spectral problems. Furthermore, this approach also has been extended to study the asymptotics for higher-order matrix Lax pair systems, such as Degasperis--Procesi equation \cite{BS-DP}, coupled NLS equation \cite{GL-JNS}, Sasa--Satsuma equation \cite{LG2021-JMP}, Spin-1 Gross--Pitaevskii equation \cite{GWC-CMP}, matrix mKdV equation \cite{liu2023}, two-component Sasa--Satsuma equation \cite{ZW-JNS}, Boussinesq equation \cite{CL-JMPA,charlier2023good}, Hirota--Satsuma equation \cite{wang2025miura}, to name a few.

In recent years, the Deift--Zhou steepest descent method was further extended to the $\bar{\partial}$ steepest descent method of McLaughlin and Miller \cite{MM2006,MM2008}, which first appeared in the calculating the asymptotic behavior of orthogonal polynomials. The $\bar{\partial}$ steepest descent method follows the general scheme of the Deift--Zhou steepest descent argument, while the rational approximation of the reflection coefficient is replaced by some non-analytic extension, which leads to a $\bar{\partial}$-problem in some sectors of the complex plane. On the other hand, this method also has displayed some advantages, such as avoiding delicate estimates involving $L^p$ estimates of
the Cauchy projection operators, lowering demands on initial conditions and improving the error estimates. This method then was adapted to obtain the long-time asymptotics for solutions to the NLS equation \cite{BJM,DM}, derivative NLS equation \cite{JLPS}, mKdV equation \cite{CL2021}, fifth-order mKdV equation \cite{LCG-SAPM}, cSP/SP equation \cite{LTYF-JDE,YF-JDE}, modified CH equation \cite{YF-ADV}, cSP positive flow \cite{GWWL-JDE}, Wadati--Konno--Ichikawa equation \cite{LTY-ADV}, Sasa--Satsuma equation \cite{XF-SS-JDE}, Novikov equation \cite{YF-2023}, coupled NLS equation \cite{HLZ}, among others, with a sharp error bound for the weighted Sobolev initial data.

The main purpose of this paper is to apply the $\bar{\partial}$-techniques to investigate the long-time asymptotic behavior of solution for the Cauchy problem of the ccSP equation \eqref{ccSPE} formulated on the whole line $x\in\bfR$ with the initial data
\be\label{IVD}
q_1(x,t=0)=q_{10}(x),\quad q_2(x,t=0)=q_{20}(x).
\ee
We will focus on the case that $q_{10}(x),q_{20}(x)$ belong to the Schwartz space $\mathcal{S}(\bfR)$ and support solitons. It is worth noting that the motivations and key highlights of this work involve the following aspects.

(I) Our attention is focused not only on the case $\zeta/t<-\varepsilon$, in which the oscillatory term $\e^{\pm2\ii t\theta}$ has two stationary points on the real axis, but also on the case $\zeta/t>\varepsilon$. However, only case $\zeta/t<-\varepsilon$ is considered in reference \cite{GLL-JDE}.

(II) We present a detailed proof that the a solution of the RH problem \ref{rhp2.1}, by means of the representation results \eqref{2.50}-\eqref{2.51}, gives rise to a solution of the ccSP equation, see Theorem \ref{th2.2}.

(III) We provide some more elaborate estimates on the rate of convergence; see the proof of Proposition \ref{prop3.3}. In particular, we bridge a gap in the literature \cite{GLL-JDE} concerning the approximation of two distinct $2\times2$ matrix-valued functions $\delta_j(k)$ using the unit matrix multiplied by their determinants. This is because the function $\delta_{2-}(k)$ is only continuous on the real axis, which makes it impossible to directly perform an analytic decomposition of the function $f(k)$ defined in \eqref{A.12}.

(IV) A class of $4\times4$ model RH problems is constructed and solved using the solutions of the parabolic cylinder equation, which may be applied in the future to address other asymptotic questions in integrable PDEs.

(V) Our long-time asymptotic expansion will result in the verification of the soliton resolution conjecture for the initial-value problem of the ccSP equation associated with a $4\times4$ matrix Lax pair.

Our main result is expressed as follows.
\begin{theorem}\label{th1.1}
Let $q_{10}(x), q_{20}(x)\in\mathcal{S}(\bfR)$ be the initial data such that Assumption \ref{assump1} is fulfilled. Let $\varepsilon$ be any small positive number. Then the behavior of the solution of the Cauchy problem for ccSP equation \eqref{ccSPE} with initial data \eqref{IVD} enjoys the following asymptotics as $t\to\infty$:

$\bullet$ In the domain $\hat{\zeta}=\zeta/t<-\varepsilon$, the solution can be written as the superposition of self-symmetric solitons/composite breathers and radiation:
\begin{align}\label{1.4}
&\begin{pmatrix}q_1(x,t)&q_2(x,t)\end{pmatrix}=
\begin{pmatrix}q_1(\zeta(x,t),t)&q_2(\zeta(x,t),t)\end{pmatrix}\\
=&\left\{
\begin{aligned}
&\sum_{n=1}^N\begin{pmatrix}\tilde{q}_{1n}^{(cb)}(\zeta,t)&\tilde{q}_{2n}^{(cb)}(\zeta,t)\end{pmatrix}
+\frac{1}{\sqrt{t}}\begin{pmatrix}q^{(cb)}_{1as}(\zeta,t)&q^{(cb)}_{2as}(\zeta,t)\end{pmatrix}
+O(t^{-1}\ln t),\,\text{if}\, \text{Re}k_n>0,\\
&\sum_{n=1}^N\begin{pmatrix}\tilde{q}_{1n}^{(sol)}(\zeta,t)&\tilde{q}_{2n}^{(sol)}(\zeta,t)\end{pmatrix}
+\frac{1}{\sqrt{t}}\begin{pmatrix}q^{(sol)}_{1as}(\zeta,t)&q^{(sol)}_{2as}(\zeta,t)\end{pmatrix}
+O(t^{-1}\ln t),\,\,\text{if}\, \text{Re}k_n=0,\nn
\end{aligned}
\right.
\end{align}
where for $\ell=cb,sol$,  $\tilde{q}_{1n}^{(\ell)}(\zeta,t)$ and $\tilde{q}_{2n}^{(\ell)}(\zeta,t)$ are given respectively by \eqref{3.181} and \eqref{3.182}, by setting $\zeta=v_nt$ with $v_n<0$ corresponding to the speed of the $n$th composite breather/self-symmetric soliton,
\begin{align}
q^{(\ell)}_{1as}(\zeta,t)=&\ii T^{-2}(0) \left(\delta_{10}^{-1}\left[[\mu^{(\ell)}_*(\zeta,t;0)]^{-1}
\mathcal{E}_1\mu^{(\ell)}_*(\zeta,t;0)\right]_{UR}\delta_{20}^{-1}\right)_{11},\\
q^{(\ell)}_{2as}(\zeta,t)=& \ii T^{-2}(0)
\left(\delta_{10}^{-1}\left[[\mu^{(\ell)}_*(\zeta,t;0)]^{-1}
\mathcal{E}_1\mu^{(\ell)}_*(\zeta,t;0)\right]_{UR}\delta_{20}^{-1}\right)_{12},
\end{align}
and
\begin{align}
x=\zeta(x,t)&+\ii\left(1+T_1+\left[[\mu^{(\ell)}_*(\zeta,t;0)]^{-1}
\mu^{(\ell)}_{*1}(\zeta,t)\right]_{11}\right.\\
&\left.+\frac{1}{\sqrt{t}}\left[[\mu^{(\ell)}_*(\zeta,t;0)]^{-1}
\mathcal{E}_1\mu^{(\ell)}_*(\zeta,t;0)\right]_{11}+O(t^{-1}\ln t)\right),\nn
\end{align}
where the constant matrices $\mu^{(\ell)}_*(\zeta,t;0)$ and $\mu^{(\ell)}_{*1}(\zeta,t)$ come from the exact solution of RH problem \ref{rh3.4} evaluated at $k=0$, $\mathcal{E}_1$ and $\delta_{j0}$ are presented in \eqref{eq3.128e1} and \eqref{3.180}.

If $\zeta/t=v$ with $v<0$ but $v\neq v_n$ for all $n=1,\cdots,N$, then we have
\begin{align}\label{1.8}
\begin{pmatrix}q_1(x,t)&q_2(x,t)\end{pmatrix}
=&\frac{T^{-2}(0)\sqrt{2\pi}\e^{-\frac{\pi\nu(k_0)}{2}}}{\sqrt{k_0t}}
\left(\frac{(\delta_{k_0}^0)^2T^2(k_0)\e^{-\frac{\pi\ii}{4}}}
{\Gamma(\ii\nu(k_0))\det[\rho(k_0)]}\begin{pmatrix}f_{11}&f_{12}\end{pmatrix}\right.\\
&+\left.\frac{(\delta_{-k_0}^0)^2T^2(-k_0)\e^{\frac{\pi\ii}{4}}}
{\Gamma(-\ii\nu(k_0))\det[\rho(-k_0)]}\begin{pmatrix}g_{11}&g_{12}\end{pmatrix}\right)
+O(t^{-1}\ln t)\nn,
\end{align}
where $k_0$, $T(k)$, $\nu(k_0)$, $\delta_{k_0}^0$ and $\delta_{-k_0}^0$ are respectively defined in \eqref{2.110}, \eqref{3.6}, \eqref{3.84}, \eqref{deltak0} and \eqref{3.117}, moreover,
\be
f=\delta_{10}^{-1}\begin{pmatrix}
\rho_{22}(k_0) & -\rho_{12}(k_0)\\ -\rho_{21}(k_0) & \rho_{11}(k_0)
\end{pmatrix}\delta_{20}^{-1},\quad g=\delta_{10}^{-1}\begin{pmatrix}
\rho_{22}(-k_0) & -\rho_{12}(-k_0)\\ -\rho_{21}(-k_0) & \rho_{11}(-k_0)
\end{pmatrix}\delta_{20}^{-1}.
\ee

$\bullet$ In the domain $\hat{\zeta}=\zeta/t>\varepsilon$, the solution tends to 0,
\begin{align}
\begin{pmatrix}q_1(x,t)&q_2(x,t)\end{pmatrix}=
\begin{pmatrix}q_1(\zeta(x,t),t)&q_2(\zeta(x,t),t)\end{pmatrix}
=O(t^{-1}),
\end{align}
and
\be
x=\zeta(x,t)+O(t^{-1}).
\ee
\end{theorem}

The rest of the paper is organized as follows. In Section \ref{sec2}, based on the Lax pair of the ccSP equation \eqref{ccSPE}, we present two kinds of eigenfunctions to control the singularity of the Lax pair as developed in \cite{LSL-IST}. One class is used to formulate the main RH problem, while the other class is used to provide the reconstruction formula for the solution of ccSP equation. In Section \ref{sec3}, we deal with the main RH problem in the region $\hat{\zeta}<-\varepsilon$ with two stationary points. With a series of transformations, the original RH problem is deformed into a mixed $\bar{\partial}$-RH problem that can be decomposed into a pure RH problem and a $\bar{\partial}$-problem, which are asymptotically analyzed in Subsection \ref{sec3.4} and Subsection \ref{sec3.5}. As a consequence, we obtain the long-time asymptotic result for the solution of the ccSP equation via the reconstruction formula. Finally, in Section \ref{sec4}, we show the long-time asymptotics in the region $\hat{\zeta}>\varepsilon$ in the similar way.

We end the introduction with some notations:

{\bf Notations.} The complex conjugate of a complex number \(a\) is denoted by \({a^*}\). For a complex-valued matrix \(A\), \(A^\dagger\) denotes the conjugate transpose. A \(4 \times 4\) matrix \(A\) is divided into three blocks:
\[
A = \begin{pmatrix}
	A_{11} & A_{12} & A_{13} & A_{14} \\
	A_{21} & A_{22} & A_{23} & A_{24} \\
	A_{31} & A_{32} & A_{33} & A_{34} \\
	A_{41} & A_{42} & A_{43} & A_{44}
\end{pmatrix} = \begin{pmatrix}A_L& A_R\end{pmatrix} = \begin{pmatrix} A_{UL} & A_{UR} \\ A_{DL} & A_{DR} \end{pmatrix},
\]
where \(A_{ij}\) represents the \((i, j)\)-entry, \(A_L\) represents the first two columns, \(A_R\) represents the last two columns, \(A_{UL}, A_{UR}, A_{DL},\) \(A_{DR}\) are \(2 \times 2\) matrices. $\mathbb{I}_{n\times n}$ indicates $n\times n$ identity matrix. For any matrix $M$ define $|M|^2=$tr$M^\dag M$, and for any matrix function $M(\cdot)$ define $\|M(\cdot)\|_{L^p}=\||M(\cdot)|\|_{L^p}$.
\section{Spectral Analysis and a RH problem}\label{sec2}
In this section, we present the RH approach to the IST for the initial-value problem of ccSP equation \eqref{ccSPE}. System \eqref{ccSPE} is the $k$-independent compatibility condition for the simultaneous linear equations of a Lax pair \cite{LSL-IST}
\be\label{lax}
\begin{aligned}
\Phi_x=&X\Phi=\begin{pmatrix}
-\ii k\mathbb{I}_{2\times2} & kQ_x \\[4pt]
-kQ^\dag_x & \ii k\mathbb{I}_{2\times2}
\end{pmatrix}\Phi,\\
\Phi_t=&T\Phi=\begin{pmatrix}
\frac{\ii}{4k}\mathbb{I}_{2\times2}-\frac{\ii}{2}kQQ^\dag & -\frac{\ii}{2}Q+\frac{1}{2}kQQ^\dag Q_x \\[4pt]
-\frac{\ii}{2}Q^\dag-\frac{1}{2}kQ^\dag QQ^\dag_x & -\frac{\ii}{4k}\mathbb{I}_{2\times2}+\frac{\ii}{2} kQ^\dag Q
\end{pmatrix}\Phi,
\end{aligned}
\ee
governing an auxiliary matrix $\Phi$ that depends on $(x,t)\in\bfR^2$ and the spectral parameter $k\in\bfC$, where $\mathbb{I}_{2\times2}$ is a $2\times2$ identity matrix and $Q$ is defined as
\be
Q=\begin{pmatrix}
-\ii q_1 & -\ii q_2\\
\ii q_2^* & -\ii q_1^*
\end{pmatrix}.
\ee
\subsection{Eigenfunctions appropriate at $k=\infty$}
In order to control the large $k$ behavior of solutions of \eqref{lax}, we will transform this Lax pair to a appropriate form. Specifically, define a $4\times4$ matrix-valued function $P(x,t)$ as
\be\label{2.1}
P=\sqrt{\frac{1+q}{2q}}\begin{pmatrix}
\mathbb{I}_{2\times2} & \frac{\ii Q_x}{1+q} \\[4pt]
\frac{\ii Q_x^\dag}{1+q} & \mathbb{I}_{2\times2} \\
\end{pmatrix},\ \text{with}\
q=\sqrt{1+|q_{1x}|^2+|q_{2x}|^2},
\ee
such that
\be
PXP^{-1}=-\ii kq\Sigma_3,\quad \Sigma_3=\begin{pmatrix}
\mathbb{I}_{2\times2} & \textbf{0}_{2\times2} \\[4pt]
\textbf{0}_{2\times2} & -\mathbb{I}_{2\times2}
\end{pmatrix}.
\ee
Then the gauge transformation
\be
\hat{\Phi}(x,t;k)=P(x,t)\Phi(x,t;k)
\ee
reduces the Lax pair \eqref{lax} to the following form:
\be\label{2.4}
\hat{\Phi}_x+G_x\hat{\Phi}=\hat{X}\hat{\Phi},\quad
\hat{\Phi}_t+G_t\hat{\Phi}=\hat{T}\hat{\Phi},
\ee
where
\begin{align}
G_x=&\ii kq\Sigma_3,\quad G_t=\ii k\left[\frac{1}{2}\left(|q_1|^2+|q_2|^2\right)q-\frac{1}{4 k^2}\right]\Sigma_3,\label{2.5}\\
\hat{X}=&\frac{1}{2q(1+q)}\begin{pmatrix}
qq_x\mathbb{I}_{2\times2}-Q_xQ^\dag_{xx}& \ii(1+q)Q_{xx}-\ii q_xQ_x\\[4pt]
\ii(1+q)Q^\dag_{xx}-\ii q_xQ^\dag_x & -qq_x\mathbb{I}_{2\times2}+Q^\dag_{xx}Q_x
\end{pmatrix},\label{2.6}\\
\hat{T}=&\frac{1+q}{4q}\begin{pmatrix}
R_tR^\dag-RR_t^\dag-\ii(QR^\dag+RQ^\dag) & 2R_t-\ii Q+\ii RQ^\dag R\\[4pt]
-2R^\dag_t-\ii Q^\dag+\ii R^\dag QR^\dag & R_t^\dag R-R^\dag R_t+\ii(R^\dag Q+Q^\dag R)
\end{pmatrix}\label{2.7}\\
&+\frac{\ii}{4 kq}\begin{pmatrix}
(1-q)\mathbb{I}_{2\times2} & -\ii Q_x\\[4pt]
\ii Q_x^\dag & -(1-q)\mathbb{I}_{2\times2}\\
\end{pmatrix}, \ R=\frac{\ii Q_x}{1+q}.\nn
\end{align}
Through the conservation law of the equation \eqref{ccSPE}
\be\label{2.8}
q_t=\frac{1}{2}\left[\left(|q_1|^2+|q_2|^2\right)q\right]_x,
\ee
we can introduce the function
\be\label{2.9}
G(x,t;k)=\ii t\hat{\theta}(x,t;k)\Sigma_3,\quad \hat{\theta}(x,t;k)=\frac{\zeta(x,t)}{t}k-\frac{1}{4k},
\ee
where
\be\label{2.10}
\zeta(x,t)\doteq x-\int_x^{+\infty}\left(q(s,t)-1\right)\dd s.
\ee
We seek simultaneous solutions $\hat{\Phi}^\pm$ of the Lax pair \eqref{2.4} such that
\be
\hat{\Phi}^\pm(x,t;k)\to\e^{-\ii t\hat{\theta}(x,t;k)\Sigma_3},\quad x\to\pm\infty.
\ee
It is convenient to introduce the matrix-valued functions $M^\pm=\hat{\Phi}^\pm\e^{\ii t\hat{\theta}\Sigma_3}$ satisfying
\be\label{2.11}
M_x+[G_x,M]=\hat{X}M,\quad
M_t+[G_t,M]=\hat{T}M,
\ee
as the unique solutions of the Volterra integral equations
\begin{align}
M^\pm(x,t;k)=\mathbb{I}_{4\times4}+\int_{\pm\infty}^x\e^{\ii k\int_x^{x'}q(s,t)\dd s\Sigma_3}\hat{X}(x',t)M^\pm(x',t;k)\e^{-\ii k\int^{x'}_xq(s,t)\dd s\Sigma_3}\dd x'.\label{2.12}
\end{align}
The existence, analyticity, symmetry and asymptotics of $M^\pm(x,t;k)$ can be proven directly. Here we list their properties.
\begin{proposition}
Assume that for each $t\in\bfR$, $q_{jx},q_{jxx}\in L^1(\bfR)$, $j=1,2$, then we have\\
$\bullet$ Analyticity: $M^-_L(x,t;k)$ and $M^+_R(x,t;k)$ are analytic for $k\in\bfC^+$ and continuous for $k\in\bfC^+\cup\bfR$; $M^+_L(x,t;k)$ and $M^-_R(x,t;k)$ are analytic for $k\in\bfC^-$ and continuous for $k\in\bfC^-\cup\bfR$.\\
$\bullet$ Symmetries: $M^\pm(x,t;k)$ obeys the symmetries
\be\label{2.17}
[M^\pm(x,t;k^*)]^\dag=[M^\pm(x,t;k)]^{-1},\quad [M^\pm(x,t;-k^*)]^*=\mathcal{A}M^\pm(x,t;k)\mathcal{A},
\ee
where
\be
\mathcal{A}=\begin{pmatrix}\ii\sigma_2 & \mathbf{0}_{2\times2}\\
\mathbf{0}_{2\times2} & \ii\sigma_2 \end{pmatrix},\quad
\sigma_2=\begin{pmatrix}0 & -\ii\\\ii & 0\end{pmatrix}.
\ee
$\bullet$ Asymptotics: The large-$k$ asymptotic behavior of $M^\pm(x,t;k)$ is given by
\be\label{2.19}
M^\pm(x,t;k)=\mathbb{I}_{4\times4}+\int_{\pm\infty}^x\hat{X}_D(s,t)\dd s+O(k^{-1}),
\ee
where $\hat{X}_D$ is the block-diagonal part of the matrix $\hat{X}$.\\
$\bullet$ Unimodularity: $\det[M^\pm(x,t;k)]=1$.
\end{proposition}
\subsection{The scattering data}
Since $\hat{\Phi}^+$ and $\hat{\Phi}^-$ are two fundamental solutions of the Lax pair for any $k\in\bfR$, one can define a $4\times4$ scattering matrix $S(k)$ such that
\be\label{2.18}
\hat{\Phi}^-(x,t;k)=\hat{\Phi}^+(x,t;k)S(k),\quad
S(k)=\begin{pmatrix} a(k) & \bar{b}(k) \\ b(k) & \bar{a}(k)\end{pmatrix},\quad k\in\bfR,
\ee
with $2\times2$ blocks $a(k)$ and $b(k)$ being the scattering coefficients. Note that $S$ is unimodular as
\be
\det[S(k)]=1,\quad k\in \bfR.
\ee
We define the reflection coefficients
\be
\rho(k)=b(k)a^{-1}(k),\quad \bar{\rho}(k)=\bar{b}(k)\bar{a}^{-1}(k),\quad k\in \bfR.
\ee
It is shown that the scattering coefficients and the reflection coefficient have the following properties.
\begin{proposition}\label{prop2.2}
We have\\
$\bullet$ $a(k)$ (respectively, $\bar{a}(k)$) is analytic in $\bfC^+$ (respectively, in $\bfC^-$) and continuous in $\bfR$, whereas $b(k)$, $\bar{b}(k)$ are in general only defined for $k\in\bfR$.\\
$\bullet$ Symmetries:
\begin{align}
&\bar{\rho}(k)=-\rho^\dag(k),\quad k\in\bfR,\quad \det[\bar{a}(k)]=\det[a^\dag(k^*)],\quad k\in\bfC^-,\\
&a^*(-k^*)=\sigma_2a(k)\sigma_2,\quad k\in\bfC^+, \quad \rho^*(-k)=\sigma_2\rho(k)\sigma_2,\quad k\in\bfR.
\end{align}
$\bullet$ Asymptotics: as $k\to\infty$,
\begin{align}
a(k)\to&\mathbb{I}_{2\times2}+\int_{-\infty}^x\hat{X}_{1D}(s,t)\dd s+\int_{+\infty}^x\hat{X}_{1D}^\dag(s,t)\dd s\label{2.25'}\\
&+\left(\int_{-\infty}^x\hat{X}_{1D}(s,t)\dd s\right)\left(\int_{+\infty}^x\hat{X}_{1D}^\dag(s,t)\dd s\right),\nn\\
b(k)\to&\mathbf{0}_{2\times2},
\end{align}
where $\hat{X}_{1D}$ denotes the diagonal block of the matrix $\hat{X}$.
\end{proposition}
\subsection{Discrete spectrum}
The values of $k$ where $\det[a(k)]$ becomes zero provide the discrete spectrum of the scattering problem. Let us assume that $\det[a(k)]$ has a finite number $N$ of zeros in $\bfC^+$. The two symmetries in Proposition \ref{prop2.2} combined give that discrete eigenvalues appear in the set $Z\cup Z^*$ with $Z=\{k_n,-k_n^*\}_{n=1}^N$, where, for each $n$, $k_n$, $-k_n^*$ are the zeros of $\det[a(k)]$ in $\bfC^+$, and $k_n^*$, $-k_n$ are the zeros of $\det[\bar{a}(k)]$ in $\bfC^-$. Furthermore, the following results hold.
\begin{proposition}
If rank $[a(k_n)]=$ rank $[a(-k_n^*)]=1$ and rank $[\bar{a}(-k_n)]=$ rank $[\bar{a}(k_n^*)]=1$, then the zeros of $\det[a(k)]$ in $\bfC^+$ and the zeros of $\det[\bar{a}(k)]$ in $\bfC^-$ are simple. If $a(k_n)=a(-k_n^*)=\mathbf{0}_{2\times2}$ and $\bar{a}(-k_n)=\bar{a}(k_n^*)=\mathbf{0}_{2\times2}$, then the zeros of $\det[a(k)]$ in $\bfC^+$ and the zeros of $\det[\bar{a}(k)]$ in $\bfC^-$ are double.
\end{proposition}
However, it is also shown in \cite{LSL-IST} that the points $k_n$, $-k_n^*$ (as well as $k_n^*$, $-k_n$) are simple poles of the function $M^-_L(x,t;k)a^{-1}(k)$ in $\bfC^+$ (and $M^-_R(x,t;k)\bar{a}^{-1}(k)$ in $\bfC^-$). Moreover, the corresponding residues are calculated as
\begin{align}
\underset{k=k_n}{\rm Res\ }[M^-_L(x,t;k)a^{-1}(k)]=&\e^{2\ii t\hat{\theta}(x,t;k_n)}M^+_R(x,t;k_n)C_n,\\
\underset{k=-k_n^*}{\rm Res\ }[M^-_L(x,t;k)a^{-1}(k)]=&-\e^{2\ii t\hat{\theta}(x,t;-k^*_n)}M^+_R(x,t;-k^*_n)\sigma_2C^*_n\sigma_2,\\
\underset{k=k_n^*}{\rm Res\ }[M^-_R(x,t;k)\bar{a}^{-1}(k)]=&-\e^{-2\ii t\hat{\theta}(x,t;k_n^*)}M^+_L(x,t;k^*_n)C_n^\dag,\\
\underset{k=-k_n}{\rm Res\ }[M^-_R(x,t;k)\bar{a}^{-1}(k)]=&\e^{-2\ii t\hat{\theta}(x,t;-k_n)}M^+_L(x,t;-k_n)\sigma_2C^{\texttt{T}}_n\sigma_2,
\end{align}
where $C_n$ is the $2\times2$ norming constant matrix associated to the discrete eigenvalue $k_n$. Additionally, rank $[C_n]=1$ if rank $[a(k_n)]=$ rank $a[(-k_n^*)]=$ rank $[\bar{a}(-k_n)]=$ rank $[\bar{a}(k_n^*)]=1$, and $C_n$ can be either full-rank or rank-1 matrix if $a(k_n)=a(-k_n^*)=\bar{a}(-k_n)=\bar{a}(k_n^*)=\mathbf{0}_{2\times2}$.

It should be noted that if $k_n$ is purely imaginary, namely, $k_n=-k_n^*$, then $C_n=-\sigma_2C^*_n\sigma_2$. Thus, it is easy to see that if $C_n$ is a rank-1 matrix, it must be a zero matrix. Hence, in this case, the only nontrivial solutions are associated with a full rank norming constant matrix $C_n=\begin{pmatrix}\alpha_n & \beta_n^*\\ \beta_n & -\alpha^*_n\end{pmatrix}$ with $\alpha_n,\beta_n\in\bfC$. When $\rho(k)\equiv\mathbf{0}_{2\times2}$, this case above corresponds to a so-called self-symmetric soliton \cite{CGP}, while if $k_n\neq-k_n^*$ and $C_n$ is a $2\times2$ full rank matrix, then the solution is referred to as a composite breather \cite{CGP,LSL-IST}.

On the other hand, the zeros of $\det[a(k)]$ on $\bfR$ are known to occur at specific values of $k$, and these correspond to spectral singularities. To exclude this phenomena and facilitate the following asymptotic analysis, we let the initial data satisfy the hypothesis.
\begin{assumption}\label{assump1}
The initial data $q_{10}(x),q_{20}(x)\in\mathcal{S}(\bfR)$ generate scattering data which satisfy that

$\bullet$ For $k\in\bfR$, no spectral singularities exist, that is, $\det[a(k)]\neq0$;

$\bullet$ The norming constant matrices $C_n$ for $n=1,\cdots,N$ are all of full rank;

$\bullet$ All the discrete eigenvalues $\{k_n\}_{n=1}^N$ satisfy
\be
-\infty<-\frac{1}{4|k_1|^2}<-\frac{1}{4|k_2|^2}<\cdots<-\frac{1}{4|k_N|^2}<\infty.
\ee
\end{assumption}
\begin{remark}
It should be pointed that the second item in Assumption \ref{assump1} is introduced to facilitate the calculation of residue conditions in the first transformation of the asymptotic analysis, see \eqref{3.15}-\eqref{3.18}. The third item aims to avoid the unstable structure in which the self-symmetric solitons and composite breathers corresponding to the zeros $k_n$ of $\det[a(k)]$ in the same velocity. Moreover, by choosing a pair of purely imaginary spectral points $k_1=\ii\nu_1$ and $k_2=\ii\nu_2$ with $\nu_2\neq\nu_1$, this will introduce the self-symmetric two-soliton solution of Equation \eqref{ccSPE}, then a self-symmetric two-soliton initial condition will satisfy all items of Asumption \ref{assump1}.
\end{remark}
\subsection{A RH problem constructed from dedicated eigenfunctions}
The analyticity properties of $M^\pm(x,t;k)$ and $a(k)$, $\bar{a}(k)$ allow us to define
\be\label{2.25}
\mu(x,t;k)=\left\{
\begin{aligned}
&\begin{pmatrix}M^-_L(x,t;k)a^{-1}(k) & M^+_R(x,t;k) \end{pmatrix},\quad k\in\bfC^+,\\
&\begin{pmatrix}M^+_L(x,t;k) & M^-_R(x,t;k)\bar{a}^{-1}(k)\end{pmatrix},\quad k\in\bfC^-.
\end{aligned}
\right.
\ee
Then the limiting values $\mu_\pm(x,t,k)$, $k\in\bfR$ of $\mu$ as $k$ is approached from the domains $\pm$Im $k>0$ are related as follows:
\be
\mu_+(x,t;k)=\mu_-(x,t;k)J(x,t;k),\quad k\in\bfR,
\ee
where
\be
J(x,t;k)=\begin{pmatrix} \mathbb{I}_{2\times2}+\rho^\dag(k)\rho(k) & \rho^\dag(k)\e^{-2\ii t\hat{\theta}(x,t;k)}\\[4pt]
\rho(k)\e^{2\ii t\hat{\theta}(x,t;k)} &\mathbb{I}_{2\times2}\end{pmatrix}.
\ee
Moreover, as $k\to\infty$,
\be\label{2.35}
\mu(x,t;k)=\mu_\infty(x,t)+O(k^{-1}),
\ee
where $\mu_\infty(x,t)$ is invertible and the explicit expression is omitted for brevity.

In order to formulate a RH problem which has an explicit dependence on parameters and equip with the normalization condition, we introduce
\be\label{2.32}
\breve{\mu}(\zeta,t;k)\doteq\mu^{-1}_\infty(x,t)\mu(x(\zeta,t),t;k).
\ee
Then we can obtain the Riemann--Hilbert problem for $\breve{\mu}(\zeta,t;k)$ as follows:
\begin{rhp}\label{rhp2.1} Find a $4\times4$ matrix-valued function $\breve{\mu}(\zeta,t;k)$ which satisfies:
\begin{itemize}
\item  Analyticity: $\breve{\mu}(\zeta,t;k)$ is analytic in $\bfC\setminus(\bfR\cup Z\cup Z^*)$, and continuous up to the boundary $k\in\bfR$.
\item Jump condition: The jump condition of $\breve{\mu}(\zeta,t;k)$  takes the form
\be
\breve{\mu}_+(\zeta,t;k)=\breve{\mu}_-(\zeta,t;k)\breve{J}(\zeta,t;k),\quad k\in\bfR,
\ee
where
\be\label{2.34}
\breve{J}(\zeta,t;k)=\begin{pmatrix} \mathbb{I}_{2\times2}+\rho^\dag(k)\rho(k) & \rho^\dag(k)\e^{-2\ii t\breve{\theta}(\zeta,t;k)}\\[4pt]
\rho(k)\e^{2\ii t\breve{\theta}(\zeta,t;k)} &\mathbb{I}_{2\times2}\end{pmatrix},\quad
\breve{\theta}(\zeta,t;k)=\frac{\zeta}{t}k-\frac{1}{4k}.
\ee
\item Normalization:
\be
\breve{\mu}(\zeta,t;k)\to\mathbb{I}_{4\times4},\quad k\to\infty.
\ee
\item Residue conditions: $\breve{\mu}(\zeta,t;k)$ has simple poles at each point in $Z\cup Z^*$ with
    \begin{align}
    \underset{k=k_n}{\rm Res\ }\breve{\mu}(\zeta,t;k)&=\lim_{k\to k_n}\breve{\mu}(\zeta,t;k)\begin{pmatrix}
    \mathbf{0}_{2\times2} &  \mathbf{0}_{2\times2}\\
    C_n\e^{2\ii t\breve{\theta}(\zeta,t;k)}& \mathbf{0}_{2\times2}
    \end{pmatrix},\label{2.38}\\
    \underset{k=-k_n^*}{\rm Res\ }\breve{\mu}(\zeta,t;k)&=\lim_{k\to -k^*_n}\breve{\mu}(\zeta,t;k)\begin{pmatrix}
    \mathbf{0}_{2\times2} &  \mathbf{0}_{2\times2}\\
    -\sigma_2C^*_n\sigma_2\e^{2\ii t\breve{\theta}(\zeta,t;k)}& \mathbf{0}_{2\times2}
    \end{pmatrix},\label{2.39}\\
    \underset{k=k_n^*}{\rm Res\ }\breve{\mu}(\zeta,t;k)&=\lim_{k\to k_n^*}\breve{\mu}(\zeta,t;k)\begin{pmatrix}
    \mathbf{0}_{2\times2} & -C^\dag_n\e^{-2\ii t\breve{\theta}(\zeta,t;k)} \\
    \mathbf{0}_{2\times2} & \mathbf{0}_{2\times2}
    \end{pmatrix},\label{2.40}\\
    \underset{k=-k_n}{\rm Res\ }\breve{\mu}(\zeta,t;k)&=\lim_{k\to -k_n}\breve{\mu}(\zeta,t;k)\begin{pmatrix}
    \mathbf{0}_{2\times2} & \sigma_2C^{\texttt{T}}_n\sigma_2\e^{-2\ii t\breve{\theta}(\zeta,t;k)} \\
    \mathbf{0}_{2\times2} & \mathbf{0}_{2\times2}
    \end{pmatrix}.\label{2.41}
    \end{align}
\end{itemize}
\end{rhp}
\subsection{Eigenfunctions appropriate at $k=0$}
In order to have better control of the behavior of solution $\breve{\mu}(\zeta,t;k)$ as $k\to0$, and reconstruct the solution of \eqref{ccSPE}, it is convenient to rewrite the Lax pair \eqref{lax} in the form
\be\label{lax2}
\Phi_x+\ii k\Sigma_3\Phi=X_0\Phi,\quad \Phi_t-\frac{\ii}{4k}\Sigma_3\Phi=T_0\Phi,
\ee
where
\be
X_0=k\begin{pmatrix}\mathbf{0}_{2\times2} & Q_x\\
-Q^\dag_x & \mathbf{0}_{2\times2}\end{pmatrix},\quad
T_0=\begin{pmatrix}-\frac{\ii}{2}kQQ^\dag & -\frac{\ii}{2}Q+\frac{1}{2}kQQ^\dag Q_x \\
-\frac{\ii}{2}Q^\dag-\frac{1}{2}kQ^\dag QQ^\dag_x & \frac{\ii}{2} kQ^\dag Q\end{pmatrix}.
\ee
Introduce
\be
G_0(x,t;k)=\ii t\theta_0(x,t;k)\Sigma_3,\quad \theta_0(x,t;k)=\frac{x}{t}k-\frac{1}{4k},
\ee
and
\be
M_0(x,t;k)=\Phi(x,t;k)\e^{\ii t\theta_0(x,t;k)\Sigma_3}.
\ee
Then the Lax pair \eqref{lax2} can be rewritten as
\be
M_{0,x}+[G_{0,x},M_0]=X_0M_0,\quad M_{0,t}+[G_{0,t},M_0]=T_0M_0.
\ee
The Jost solutions $M_0^\pm(x,t;k)$ are determined, similarly to above, as the solutions of associated Volterra integral equations:
\be
M_0^\pm(x,t;k)=\mathbb{I}_{4\times4}+\int_{\pm\infty}^x\e^{\ii k(y-x)\Sigma_3}X_0(y,t;k)
M_0^\pm(y,t;k)\e^{-\ii k(y-x)\Sigma_3}\dd y.
\ee
\begin{proposition}
As $k\to0$, we have
\be\label{2.42}
M_0^\pm(x,t;k)=\mathbb{I}_{4\times4}+k\begin{pmatrix}\mathbf{0}_{2\times2}&Q\\
-Q^\dag & \mathbf{0}_{2\times2}\end{pmatrix}+O(k^2).
\ee
\end{proposition}

\subsection{The solution of the coupled complex short pulse equation}
We notice that $M^\pm$ and $M_0^\pm$, being related to the same system of
\eqref{lax}, must be related as
\be
M^\pm(x,t;k)=P(x,t)M_0^\pm(x,t;k)
\e^{-\ii t\theta_0(x,t;k)\Sigma_3}\Gamma^\pm(k)\e^{\ii t\hat{\theta}(x,t;k)\Sigma_3},
\ee
where
\be
\Gamma^+(k)=\mathbb{I}_{4\times4},\quad \Gamma^-(k)=\e^{\ii k\varsigma\Sigma_3},\quad
\varsigma=\int_{-\infty}^{+\infty}(q(s,t)-1)\dd s.
\ee
Combining with \eqref{2.42}, we can obtain the asymptotics of $M^\pm(x,t;k)$ as $k\to0$
\be\label{2.45}
\begin{aligned}
M^+(x,t;k)=&P(x,t)\left(\mathbb{I}_{4\times4}+k\begin{pmatrix}\mathbf{0}_{2\times2}&Q\\
-Q^\dag & \mathbf{0}_{2\times2}\end{pmatrix}-\ii k\int_{x}^{+\infty}(q(s,t)-1)\dd s\Sigma_3+O(k^2)\right),\\
M^-(x,t;k)=&P(x,t)\left(\mathbb{I}_{4\times4}+k\begin{pmatrix}\mathbf{0}_{2\times2}&Q\\
-Q^\dag & \mathbf{0}_{2\times2}\end{pmatrix}+\ii k\int^{x}_{-\infty}(q(s,t)-1)\dd s\Sigma_3+O(k^2)\right).
\end{aligned}
\ee
Recalling the definition of $M^\pm(x,t;k)$, and from \eqref{2.18}, we can rewrite the scattering matrix as
\be
S(k)=\e^{\ii t\hat{\theta}(x,t;k)\Sigma_3}[M^+(x,t;k)]^{-1}M^-(x,t;k)\e^{-\ii t\hat{\theta}(x,t;k)\Sigma_3},
\ee
and expand $S(k)$ as $k=0$, then we obtain as $k\to0$
\begin{align}
a(k)=&(1+\ii k\varsigma)\mathbb{I}_{2\times2}+O(k^2),\quad
\bar{a}(k)=(1-\ii k\varsigma)\mathbb{I}_{2\times2}+O(k^2),\label{2.47}\\
b(k)=&\bar{b}(k)=O(k^2).
\end{align}
\begin{proposition}\label{pro2.5}
As $k\to0$, $\rho(k)=O(k^2)$, that is,
\be
\lim_{k\to0}\rho(k)=\mathbf{0}_{2\times2}.
\ee
\end{proposition}
Finally, substituting \eqref{2.45} and \eqref{2.47} into \eqref{2.25} gives
\be\label{2.53}
\mu(x,t;k)=P(x,t)\left(\mathbb{I}_{4\times4}-\ii k\begin{pmatrix}
(x-\zeta)\mathbb{I}_{2\times2} & \ii Q\\-\ii Q^\dag & (\zeta-x)\mathbb{I}_{2\times2}
\end{pmatrix}+O(k^2)\right),\quad k\to0.
\ee
Note that $\mu(x,t;0)=P(x,t)$, it follows from \eqref{2.32} and \eqref{2.53} that
\be\label{2.59}
\begin{aligned}
&\lim_{k\to0}\frac{\ii}{k}[\breve{\mu}^{-1}(\zeta,t;0)
\breve{\mu}(\zeta,t;k)-\mathbb{I}_{4\times4}]\\
=&\lim_{k\to0}\frac{\ii}{k}\left[\left(\mu_\infty^{-1}(x,t)
\mu(x,t;0)\right)^{-1}\left(\mu_\infty^{-1}(x,t)
\mu(x,t;k)\right)-\mathbb{I}_{4\times4}\right]\\
=&\lim_{k\to0}\frac{\ii}{k}\left[\mu^{-1}(x,t;0)\mu(x,t;k)-\mathbb{I}_{4\times4}\right]\\
=&\begin{pmatrix}
(x-\zeta)\mathbb{I}_{2\times2} & \ii Q\\-\ii Q^\dag & (\zeta-x)\mathbb{I}_{2\times2}
\end{pmatrix}.
\end{aligned}
\ee
This important relation gives the following result.
\begin{theorem}\label{th2.1}
The associated $4\times4$ matrix RH problem \ref{rhp2.1} for $\breve{\mu}(\zeta,t;k)$ has a unique solution. Evaluating $\breve{\mu}(\zeta,t;k)$ as $k\to0$, we get a parametric representation for the solution of the initial-value problem \eqref{ccSPE}-\eqref{IVD} in terms of the solution of this RH problem:
\be\label{2.60}
\begin{pmatrix} q_1(x,t) & q_2(x,t)\end{pmatrix}=\begin{pmatrix} q_1(\zeta(x,t),t) & q_2(\zeta(x,t),t)\end{pmatrix},
\ee
where
\begin{align}
\begin{pmatrix} q_1(\zeta,t) & q_2(\zeta,t)\end{pmatrix}=&
\lim_{k\to0}\frac{\ii}{k}\begin{pmatrix}\left(\breve{\mu}^{-1}(\zeta,t;0)
\breve{\mu}(\zeta,t;k)\right)_{13}
& \left(\breve{\mu}^{-1}(\zeta,t;0)
\breve{\mu}(\zeta,t;k)\right)_{14}\end{pmatrix},\label{2.50}\\
x-\zeta(x,t)=&\lim_{k\to0}\frac{\ii}{k}\left[\left(\breve{\mu}^{-1}(\zeta,t;0)
\breve{\mu}(\zeta,t;k)\right)_{11}-1\right].\label{2.51}
\end{align}
\end{theorem}
\begin{proof}
It is easy to show that the jump conditions and the jump matrices in RH problem \ref{rhp2.1} satisfy the hypotheses of Zhou's vanishing lemma \cite{ZX1989} by replacing the residue conditions \eqref{2.38}-\eqref{2.41} with Schwarz invariant jump condition across a set of complete contours centered at $\pm k_n$ and $\pm k^*_n$. Thus, the solution $\breve{\mu}(\zeta,t;k)$ of the RH problem \ref{rhp2.1} is existent and unique. In view of \eqref{2.59}, the representation formulas hold immediately.
\end{proof}
In the following, we will check that the representation results \eqref{2.60} and \eqref{2.50} in Theorem \ref{th2.1} actually satisfies the ccSP equation and the initial condition.

In order to streamline the forthcoming analysis, we now present the following definitions:
\begin{gather}
	\label{eq1}\bfu(x,t)= (q_1(x,t)\quad q_2(x,t)), \quad \hat{\bfu}(\zeta,t)= (q_1(\zeta(x,t),t)\quad q_2(\zeta(x,t),t)), \\
	 \label{eq2.62}\mu(\zeta,t;0)= P(\zeta,t)= \begin{pmatrix}
		A(\zeta,t) & B(\zeta,t) \\
		-B^{\dagger}(\zeta,t) & A(\zeta,t)
	\end{pmatrix},\quad A(\zeta,t)= m(\zeta,t)\mathbb{I}_{2\times 2},
\end{gather}
where $m(\zeta,t)$ is a real-valued function and $B(\zeta,t)$ is a $2\times 2$ complex-valued matrix. In terms of \(P^{-1}(\zeta,t)= P^{\dagger}(\zeta,t)\), we can derive that
\begin{equation*}
	A A+B B^{\dagger}= \mathbb{I}_{2\times 2}, \quad B^{\dagger}B + A A= \mathbb{I}_{2\times 2}.
\end{equation*}
By examining \eqref{eq2.62}, it is evident to observe that
\begin{equation}
	BB^{\dagger}= B^{\dagger}B = (1-m^2)\mathbb{I}_{2\times 2}.
\end{equation}
Therefore, letting \(n^{2}= 1-m^2\) leads to
\begin{equation}\label{eq3}
	BB^{\dagger}= B^{\dagger}B = n^2\mathbb{I}_{2\times 2}.
\end{equation}
For clarity, we express the \(2\times 2\) matrix \(B\) as
\begin{equation}
\label{eq2.65}	B= \begin{pmatrix}
		B_{11} & B_{12}\\
		B_{21} & B_{22}\\
	\end{pmatrix}.
\end{equation}
Denote \(\gamma = (B_{11},B_{12})\). Then, from \eqref{eq3},  it is straightforward to verify that
\begin{equation}
	\gamma\gamma^{\dagger}= n^2.
\end{equation}
On the other hand, according to \eqref{2.19} and \eqref{2.25'}, it is clear that $\mu_{\infty}(x,t)$ is a block diagonal matrix, when viewed as a $2\times 2$ block matrix. Therefore, we can express $\mu_{\infty}(x,t)$ in the following form:
\begin{equation}\label{2.69}
	\mu_{\infty}^{-1}(x,t) = \begin{pmatrix}
		C(x,t) & \mathbf{0}_{2\times2} \\
		\mathbf{0}_{2\times2}  & D(x,t),
	\end{pmatrix}
\end{equation}
where $C(x,t)$ and $D(x,t)$ are matrix-valued functions.
Based on the symmetry condition satisfied by $M(x,t;k)$ in \eqref{2.17}, it follows that $\mu_{\infty}(x,t)$ also satisfies the symmetry
\begin{equation}\label{2.70}
	[\mu_{\infty}(x,t)]^{-1} = [\mu_{\infty}(x,t)]^{\dagger}.
\end{equation}
Hence, from \eqref{2.69} and \eqref{2.70}, we have
\begin{equation}\label{2.71}
	CC^{\dagger} = C^{\dagger}C = \mathbb{I}_{2\times 2},\quad 	DD^{\dagger} = D^{\dagger}D = \mathbb{I}_{2\times 2}.
\end{equation}
In view of \eqref{2.32} and \eqref{2.53}, we obtain
\begin{equation}
	\breve{\mu}(x,t;k) = \mu_{\infty}^{-1}(x,t)P(x,t)\left(\mathbb{I}_{4\times4}-\ii k\begin{pmatrix}
		(x-\zeta)\mathbb{I}_{2\times2} & \ii Q\\-\ii Q^\dag & (\zeta-x)\mathbb{I}_{2\times2}
	\end{pmatrix}+O(k^2)\right),\quad k\to0,\label{eq2.72}
\end{equation}
and hence
\begin{equation}\label{2.73}
	\breve{\mu}(\zeta,t;0)= \mu_{\infty}^{-1}(\zeta,t)P(\zeta,t)= \begin{pmatrix}
		CA & CB \\
		-DB^{\dagger} &DA
	\end{pmatrix}(\zeta,t).
\end{equation}
\begin{theorem}\label{th2.2}
We define the newly introduced functions as follows:
\begin{equation}\label{eq6}
	\hat{q}(\zeta,t)\doteq \frac{1}{m^2(\zeta,t)-n^2(\zeta,t)},\quad \hat{\bfw}(\zeta,t)\doteq \frac{2m(\zeta,t)\gamma(\zeta,t)}{m^2(\zeta,t)- n^2(\zeta,t)}.
\end{equation}
Then the following equations (between functions of \((\zeta, t)\)) hold:
\begin{align}
\label{eqi} x_{\zeta} =& \frac{1}{\hat{q}},\\
\label{eqii} \hat{\bfu}_{\zeta} =& \frac{\hat{\bfw}}{\hat{q}},\\
\label{eqiii} \hat{q}_t =&\frac{1}{2}\hat{q}(\hat{\bfw}\hat{\bfu}+\hat{\bfu}\hat{\bfw}).
\end{align}
Moreover, we introduce $\bfu(x,t)\doteq\hat{\bfu}(\zeta(x,t),t),\ q(x,t)\doteq \hat{q}(\zeta(x,t),t).$
Thus, Equations \eqref{eqi}-\eqref{eqiii} become
\begin{align}
\label{eq26a} q_t &= \frac{1}{2}(\bfu\bfu^{\dagger}q)_x,\\
\label{eq26b} q &=\sqrt{1+\bfu_x\bfu^{\dagger}_x},
\end{align}
which is the ccSP equation \eqref{ccSPE} in the conservation law form.
\end{theorem}
\begin{proof}
The proofs of the above equations are carried out through detailed computations of \(\Psi_{\zeta} \Psi^{-1} \) and \(\Psi_{t} \Psi^{-1} \), where
\begin{equation}
	\Psi(\zeta,t;k)\doteq \breve{\mu}(\zeta,t;k)\text{e}^{(-\text{i}k\zeta- \frac{t}{4\ii k})\Sigma_3}.
\end{equation}
	We begin our analysis of \(\Psi_{\zeta} \Psi^{-1} \). By considering the expansion
	\begin{equation}
		\breve{\mu}(\zeta,t;k)= \mathbb{I}_{4\times 4} + \frac{\mu_1}{\ii k} + O(k^2), \quad k\to\infty,
	\end{equation}
	and letting \(W\doteq -[\mu_1,\Sigma_3]\), we obtain
	\begin{equation}
		\Psi_{\zeta} \Psi^{-1}(\zeta,t;k)= -\ii k \Sigma_3 + W(\zeta,t) + O(k^{-1}),\quad k\to\infty.
	\end{equation}
	Furthermore, \((\Psi_{\zeta} \Psi^{-1})(\zeta,t;k)+\ii k \Sigma_3 \) is analytic in \(\mathbb{C}\), with no jump discontinuities or singularities, and remains bounded as \(k\to\infty\). Hence, according to Liouville's theorem, we conclude that
	\begin{equation}\label{eq10}
		\Psi_{\zeta} \Psi^{-1}(\zeta,t;k)= -\ii k \Sigma_3 + W(\zeta,t).
	\end{equation}
	On the other hand, we proceed from the expansion
	\begin{equation}
		\Psi(\zeta,t;k)= G_0(\zeta,t)\left(\mathbb{I}_{4\times 4}-\ii kG_1(\zeta,t)+O(k^2)\right)\e^{(-\ii k\zeta-\frac{t}{4\ii k})\Sigma_3},\quad k\to 0,
	\end{equation}
	where, by \eqref{eq2.72} and \eqref{2.73}
	\begin{equation}
		G_0=\begin{pmatrix}
			CA & CB \\
			-DB^{\dagger} & DA
		\end{pmatrix},\quad
		G_1=\begin{pmatrix}
		(x-\zeta)\mathbb{I}_{2\times2} & \ii Q\\-\ii Q^\dag & (\zeta-x)\mathbb{I}_{2\times2}
	\end{pmatrix}\doteq\begin{pmatrix}
			f_1 & f_2 \\
			f_2^{\dagger} & -f_1
		\end{pmatrix},
	\end{equation}
	we then can get
	\begin{equation}
		\Psi_{\zeta} \Psi^{-1}= G_{0\zeta}G_0^{-1}- \ii kG_0(G_{1\zeta}+\Sigma_3)G_0^{-1} + O(k^2),\quad k\to 0.
	\end{equation}
	A comparison with \eqref{eq10} immediately yields
	\begin{equation}
		G_{1\zeta}= -\Sigma_3 + G_0^{-1}\Sigma_3G_0 = \begin{pmatrix}
			(m^2-n^2-1)\mathbb{I}_{2\times 2} & 2mB \\
			2mB^{\dagger} & -(m^2-n^2-1)\mathbb{I}_{2\times 2}
		\end{pmatrix},
	\end{equation}
	In view of \eqref{2.50}, \eqref{2.51} and \eqref{eq6}, we arrive at
	\begin{equation}
		((f_1)_{11})_\zeta= \frac{1}{\hat{q}}-1, \quad ((f_2)_{11},(f_2)_{12})_{\zeta} =\frac{\hat{\bfw}}{\hat{q}}.
	\end{equation}
	Therefore, both \eqref{eqi} and \eqref{eqii} are satisfied. In fact, by \eqref{2.50} and \eqref{2.51}, we have \(x_\zeta = 1+ ((f_1)_{11})_\zeta\) and \(\hat{\bfu}_{\zeta}= ((f_2)_{11},(f_2)_{12})_{\zeta}\).

	Now	we begin our analysis of \(\Psi_{t} \Psi^{-1}\). Since
	\begin{equation}
		\Psi_{t} \Psi^{-1}= O(k^{-1}),\quad k\to\infty,
	\end{equation}
	and
	\begin{equation}
		\Psi_{t} \Psi^{-1}= -\frac{1}{4\ii k}G_0\Sigma_3G_0^{-1}+\left\{G_{0t}+ \frac{1}{4}G_0[G_1,\Sigma_3]\right\}G_0^{-1},\quad k\to 0.
	\end{equation}
	Then, by Liouville's theorem, we derive that
	\begin{equation}\label{eq18}
		G_{0t}= -\frac{1}{4}G_0[G_1,\Sigma_3]= -\frac{1}{2}\begin{pmatrix}
			CB f_2^{\dagger} & -mCf_2 \\
			mDf_2^{\dagger} & DB^{\dagger}f_2
		\end{pmatrix}.
	\end{equation}
	It is easy to see that
	\begin{equation}\label{eq19}
		(CB)_{t} = \frac{1}{2}mCf_2.
	\end{equation}
	Based on the fact that
	\begin{equation}
		CB(CB)^{\dagger} = CBB^{\dagger}C^{\dagger} = n^2\mathbb{I}_{2\times 2}.
	\end{equation}
	Thus, we can deduce that
	\begin{equation}\label{2.89}
		\left(CB(CB)^{\dagger}\right)_t = (n^2)_t\mathbb{I}_{2\times 2}.
	\end{equation}
	On the other hand, by examining \eqref{eq19}, we have
	\begin{align}
		\left(CB(CB)^{\dagger}\right)_t &=  (CB)_t(CB)^{\dagger}+CB(CB)^{\dagger}_t= \frac{1}{2}mCf_2B^{\dagger}C^{\dagger}+\frac{1}{2}mCBf_2^{\dagger}C^{\dagger} \\
		&=\frac{1}{2}mC(f_2B^{\dagger}+Bf_2^{\dagger})C^{\dagger}.\nn
	\end{align}
	In view of \eqref{2.71}, the fact that $f_2=\ii Q$, $B=1/\sqrt{2q(1+q)}\ii Q_x$, and by comparison with \eqref{2.89}, we arrive at
	\begin{equation}\label{eq22}
		(n^2)_t=\frac{1}{2}m\hat{\bfu}\gamma^{\dagger}+\frac{1}{2}m\gamma \hat{\bfu}^{\dagger}.
	\end{equation}
	Recalling that $n^{2}= 1- m^2$, we can thus derive the expression for the derivative of $m$ with respect to $t$
	\begin{equation}
		(m^2)_t= (1-n^2)_t= -(n^2)_t.
	\end{equation}
	From the definition of \(\hat{q}\) in \eqref{eq6},  we now consider \(\hat{q}_t\)
	\begin{align} \label{eq24}
		\hat{q}_t = -\frac{(m^2)_t- (n^2)_t}{(m^2-n^2)^2} = \frac{2(n^2)_t}{(m^2-n^2)^2}
		=\frac{m\hat{\bfu}\gamma^{\dagger}+m\gamma \hat{\bfu}^{\dagger}}{(m^2-n^2)^2}.
	\end{align}
	By substituting \eqref{eq6} into \eqref{eq24}, we conclude that
	\begin{equation}
		\hat{q}_t= \frac{1}{2}\hat{q}(\hat{\bfw}\hat{\bfu}^{\dagger}+\hat{\bfu}\hat{\bfw}^{\dagger}).
	\end{equation}
	Therefore, \eqref{eqiii} is satisfied.

We next present the proof of \eqref{eq26a} and \eqref{eq26b}.
	To begin with, the equation \eqref{eqi} yields \(\zeta_x(x,t)=q(x,t)\), while equation \eqref{eqii} gives \(\hat{\bfu}_{\zeta}(\zeta(x,t),t) = \frac{\bfw}{q}(x,t)\), where \(\bfw(x,t)\doteq\hat{\bfw}(\zeta(x,t),t)\). Accordingly, the identity \(\bfu_x(x,t) = \hat{\bfu}_\zeta(\zeta(x,t),t)\zeta_x(x,t)\) implies
	\begin{equation}\label{eq27}
		\bfw=\bfu_x.
	\end{equation}
	Hence, \eqref{eq26b} takes the form $q= \sqrt{1+\bfw\bfw^{\dagger}}$, or equivalently $\hat{q}= \sqrt{1+\hat{\bfw}\hat{\bfw}^{\dagger}}$, which follows from the expressions for \(\hat{q}\) and \(\hat{\bfw}\) in \eqref{eq6}.
	To derive equation \eqref{eq26a}, we observe that \eqref{eqiii} can be expressed in the form of a conservation law
	\begin{equation}\label{eq28}
		\left(\frac{1}{\hat{q}}\right)_t = -\frac{1}{2}\left(\hat{\bfu}\hat{\bfu}^{\dagger}\right)_{\zeta}.
	\end{equation}
	This follows from the relationship between \(\hat{q}\) and \(\hat{q}_t\), which leads to
	\begin{equation}
		\left(\frac{1}{\hat{q}}\right)_t= -\frac{\hat{q}_t}{\hat{q}^2}= -\frac{\frac{1}{2}(\hat{\bfw}\hat{\bfu}^{\dagger}+\hat{\bfu}\hat{\bfw}^{\dagger})}{\hat{q}} = -\frac{1}{2}(\hat{\bfu}_\zeta\hat{\bfu}^{\dagger}+\hat{\bfu}\hat{\bfu}^{\dagger}_\zeta) = -\frac{1}{2}(\hat{\bfu}\hat{\bfu}^{\dagger})_\zeta.
	\end{equation}
We proceed by calculating \(x_t(\zeta,t)\) based on (i), and subsequently apply \eqref{eq28}:
	\begin{equation}
		x_t(\zeta,t)= -\frac{\partial}{\partial t}\left(\int_{\zeta}^{+\infty}\frac{\text{d}s}{\hat{q}(s,t)}\right) = \frac{1}{2}\int_{\zeta}^{+\infty}\left(\hat{\bfu}\hat{\bfu}^{\dagger}\right)_s(s,t)\text{d}s = -\frac{1}{2}\hat{\bfu}\hat{\bfu}^{\dagger}(\zeta,t).
	\end{equation}
	Substituting this into the identity \(\hat{q}_t = q_xx_t+q_t\) where the functions depend on \((\zeta,t)\), and applying \eqref{eqiii} gives
\be
q_t= \frac{1}{2}\hat{q}(\hat{\bfw}\hat{\bfu}^{\dagger}+\hat{\bfu}\hat{\bfw}^{\dagger}) + \frac{1}{2}q_x\hat{\bfu}\hat{\bfu}^{\dagger},
\ee
which, when written in terms of functions \((x,t)\), becomes
\be
q_t=\frac{1}{2}q(\bfw\bfu^{\dagger}+\bfu\bfw^{\dagger})+\frac{1}{2}q_x\bfu\bfu^{\dagger}.
\ee
From here, applying \eqref{eq27} leads to \eqref{eq26a}:
	\begin{equation}
		q_t = \frac{1}{2}q(\bfu_x\bfu^{\dagger}+\bfu\bfu_x^{\dagger})+\frac{1}{2}q_x\bfu\bfu^{\dagger}= \frac{1}{2}(\bfu\bfu^{\dagger}q)_x.
	\end{equation}
In order to verify the initial conditions, one observes that for $t=0$, the RH problem reduces to that associated with $q_{1,0}(x)$ and $q_{2,0}(x)$, which yields $q_1(x,t=0)=q_{1,0}(x), q_2(x,t=0)=q_{2,0}(x)$, owing to the uniqueness of the solution of the RH problem.
Therefore, we have completed the whole proof.
\end{proof}
\subsection{Classification of asymptotic regions}
The existence of a representation \eqref{2.50} of the solution in terms of the solution of the associated RH problem makes it possible to study the long-time behavior of the former problem via the long-time analysis of the latter, applying the $\bar{\partial}$-generalization of the Deift--Zhou nonlinear steepest descent method. A key feature of this method is the deformation of the original RH problem \ref{rhp2.1} according to the ``signature table" for the phase function $\breve{\theta}(\zeta,t;k)$ in the jump matrix and residue conditions. Introduce
\be
\theta(\hat{\zeta};k)\doteq\breve{\theta}(\zeta,t;k)=\hat{\zeta}k-\frac{1}{4k}, \quad \hat{\zeta}\doteq\frac{\zeta}{t}.
\ee
Then the signature table is the distribution of signs of Im$\theta(\hat{\zeta};k)$ in the $k$-plane
\be
\text{Im}\theta(\hat{\zeta};k)=\text{Im}k\cdot\left(\hat{\zeta}+\frac{1}{4|k|^2}\right).
\ee
Thus, under the condition $\hat{\zeta}>\varepsilon$ for any $\varepsilon>0$, the set $\{k|\text{Im}\theta(\hat{\zeta};k)=0\}$ coincides with the real axis Im$k=0$ and $\pm$Im$\theta>0$ for $\pm$Im$k>0$. While in the case of $\hat{\zeta}<-\varepsilon$, we have
\be
\{k|\text{Im}\theta(\hat{\zeta};k)=0\}=\{k|\text{Im}k=0\}\cup\{k||k|=(-4\hat{\zeta})^{-1/2}\},
\ee
and the sign picture of Im$\theta$ for this case is shown in Figure \ref{fig1}, where
\be\label{2.110}
k_0=\frac{1}{2\sqrt{-\hat{\zeta}}}.
\ee
\begin{figure}[htbp]
\centering
\includegraphics[width=3in]{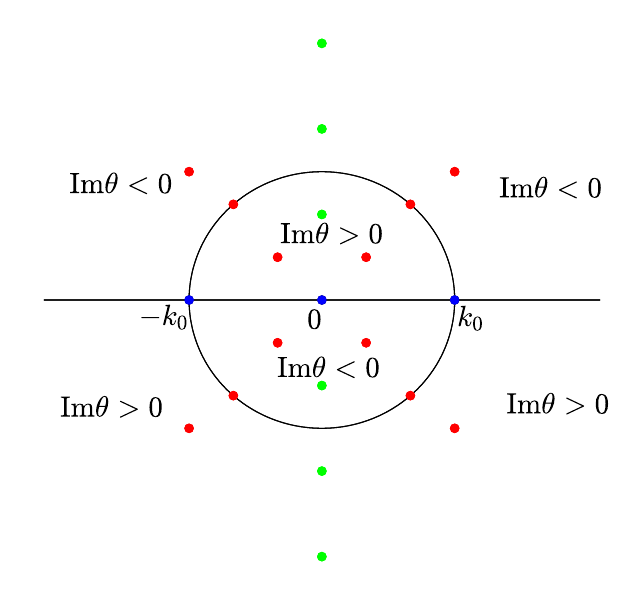}
\caption{Sign distribution of Im$\theta$ in the $k$-plane in the case $\hat{\zeta}<-\varepsilon$.}\label{fig1}
\end{figure}

These analysis suggests us to divide half-plane $\{(\zeta,t)|-\infty<\zeta<\infty, t>0\}$ in two space-time regions: Case I: $\hat{\zeta}>\varepsilon$, Case II: $\hat{\zeta}<-\varepsilon$, where $\varepsilon$ is any small positive number. For the case I, there is no stationary phase point on the real axis, while for the case II, there exist two stationary phase points on the real axis, denoted as $\pm k_0$.

\section{Asymptotics in range $\hat{\zeta}<-\varepsilon$}\label{sec3}

In a domain of the form $\hat{\zeta}<-\varepsilon$ for any $\varepsilon>0$, along a characteristic line $\zeta=v_nt$ for $v_n=-1/(4|k_n|^2)$, the associated signature table in Figure \ref{fig1} dictates the use of two factorizations of the jump matrix $\breve{J}$:
\be\label{3.1}
\breve{J}=\left\{
\begin{aligned}
&\begin{pmatrix} \mathbb{I}_{2\times2} & \rho^\dag(k)\e^{-2\ii t\theta}\\[4pt]
\mathbf{0}_{2\times2} &\mathbb{I}_{2\times2}\end{pmatrix}\begin{pmatrix} \mathbb{I}_{2\times2} & \mathbf{0}_{2\times2}\\[4pt]
\rho(k)\e^{2\ii t\theta} &\mathbb{I}_{2\times2}\end{pmatrix},\quad k\in(-k_0,k_0),\\
&\begin{pmatrix}
\mathbb{I}_{2\times2} & \textbf{0}_{2\times2}\\
\rho(k)\left(\mathbb{I}_{2\times2}+\rho^\dag(k)\rho(k)\right)^{-1}\e^{2\ii t\theta} & \mathbb{I}_{2\times2}
\end{pmatrix}
\begin{pmatrix}
\mathbb{I}_{2\times2}+\rho^\dag(k)\rho(k) & \textbf{0}_{2\times2}\\
\textbf{0}_{2\times2} & \left(\mathbb{I}_{2\times2}+\rho(k)\rho^\dag(k)\right)^{-1}
\end{pmatrix}\\
&\times
\begin{pmatrix}
\mathbb{I}_{2\times2} & \left(\mathbb{I}_{2\times2}+\rho^\dag(k)\rho(k)\right)^{-1}\rho^\dag(k)\e^{-2\ii t\theta}\\
\textbf{0}_{2\times2} & \mathbb{I}_{2\times2}
\end{pmatrix},\quad k\in(-\infty,-k_0)\cup(k_0,+\infty).
\end{aligned}
\right.
\ee
Moreover, we also need the partition $\Delta_{k_0}^\pm$ of $N$ defined by
\begin{align}
\Delta_{k_0}^+=\{n\in\{1,2,\cdots,N\}| |k_n|\leq k_0\},\ \Delta_{k_0}^-=\{n\in\{1,2,\cdots,N\}| |k_n|>k_0\}.
\end{align}
\subsection{The first transformation}
In order to decompose the two factorizations in \eqref{3.5} into the lower-upper triangular factorization, we introduce two $2\times2$ matrix-valued functions $\delta_1(k)$ and $\delta_2(k)$, and they respectively satisfy the following RH problems
\be\label{3.3}
\left\{
\begin{aligned}
\delta_{1+}(k)&=\delta_{1-}(k)
\left(\mathbb{I}_{2\times2}+\rho^\dag(k)\rho(k)\right),\quad |k|>k_0,\\
&=\delta_{1-}(k),\qquad\qquad\qquad\qquad\qquad |k|<k_0,\\
\delta_{1}(k)&\to\mathbb{I}_{2\times2},\qquad\qquad\qquad\qquad\qquad\quad k\to\infty,
\end{aligned}
\right.
\ee
and
\be\label{3.4}
\left\{
\begin{aligned}
\delta_{2+}(k)&=\left(\mathbb{I}_{2\times2}+\rho(k)\rho^\dag(k)\right)\delta_{2-}(k),\quad \ |k|>k_0,\\
&=\delta_{2-}(k),\qquad\qquad\qquad\qquad\qquad |k|<k_0,\\
\delta_{2}(k)&\to\mathbb{I}_{2\times2},\qquad\qquad\qquad\qquad\qquad\quad k\to\infty.
\end{aligned}
\right.
\ee
The unique solvability of above two RH problems is a consequence of the ``vanishing lemma" of Zhou \cite{ZX1989} since $\mathbb{I}_{2\times2}+\rho^\dag(k)\rho(k)$ and $\mathbb{I}_{2\times2}+\rho(k)\rho^\dag(k)$ are positive definite. Moreover, for $j=1,2$, $\delta_j(k)$ is bounded and satisfy the symmetry relations
\be\label{3.5}
\delta_j^{-1}(k)=\delta_j^\dag(k^*),\quad \delta_j^*(-k^*)=\sigma_2\delta_j(k)\sigma_2.
\ee
We also define a function $T(k)$ which will be used to modify the residue conditions, ensuring that they behave well as $t\to\infty$,
\be\label{3.6}
T(k)=
\prod_{\substack{\text{Re}k_n\neq0,\text{Im}k_n>0\\n\in\Delta_{k_0}^-}}
\frac{k-k_n^*}{k-k_n}\frac{k+k_n}{k+k_n^*}
\prod_{\substack{\text{Re}k_n=0,\text{Im}k_n>0\\n\in\Delta_{k_0}^-}}\frac{k-k_n^*}{k-k_n}.
\ee
The first transformation is as follows:
\be\label{3.7}
\mu^{(1)}(\zeta,t;k)=\breve{\mu}(\zeta,t;k)\tilde{\Delta}(k)[T(k)]^{-\Sigma_3},
\ee
where
\be\label{3.8}
\tilde{\Delta}(k)=\begin{pmatrix}
\delta_1^{-1}(k) & \textbf{0}_{2\times2}\\
\textbf{0}_{2\times2} & \delta_2(k)
\end{pmatrix}.
\ee
Then we get the following $4\times4$ matrix RH problem for $\mu^{(1)}$:
\begin{rhp}\label{rhp3.1} Find a $4\times4$ matrix-valued function $\mu^{(1)}(\zeta,t;k)$ which satisfies:
\begin{itemize}
\item Analyticity: $\mu^{(1)}(\zeta,t;k)$ is analytic in $\bfC\setminus(\bfR\cup Z\cup Z^*)$ and has simple poles;\\
\item Jump condition: For $k\in\bfR$,
\be
\mu^{(1)}_+(\zeta,t;k)=\mu^{(1)}_-(\zeta,t;k)J^{(1)}(\zeta,t;k),
\ee
where
\be\label{3.10}
J^{(1)}=\left\{
\begin{aligned}
&\begin{pmatrix}
\mathbb{I}_{2\times2} & \delta_{1}\rho^\dag\delta_{2}T^2\e^{-2\ii t\theta}\\
\textbf{0}_{2\times2} & \mathbb{I}_{2\times2}
\end{pmatrix}\begin{pmatrix}
\mathbb{I}_{2\times2} & \textbf{0}_{2\times2} \\
\delta_{2}^{-1}\rho\delta_{1}^{-1}T^{-2}\e^{2\ii t\theta} & \mathbb{I}_{2\times2}
\end{pmatrix},\,|k|<k_0,\\
&\begin{pmatrix}
\mathbb{I}_{2\times2} & \textbf{0}_{2\times2}\\
\delta_{2-}^{-1}\rho\left(\mathbb{I}_{2\times2}+\rho^\dag \rho\right)^{-1}\delta_{1-}^{-1}T^{-2}\e^{2\ii t\theta} & \mathbb{I}_{2\times2}
\end{pmatrix}\\
&\times\begin{pmatrix}
\mathbb{I}_{2\times2} & \delta_{1+}\left(\mathbb{I}_{2\times2}+\rho^\dag \rho\right)^{-1}\rho^\dag\delta_{2+}T^2\e^{-2\ii t\theta} \\
\textbf{0}_{2\times2} & \mathbb{I}_{2\times2}
\end{pmatrix},\, |k|>k_0.
\end{aligned}
\right.
\ee
\item Normalization: $\mu^{(1)}(\zeta,t;k)\to\mathbb{I}_{4\times4}$, as $k\to\infty$.
\item Residue conditions: $\mu^{(1)}(\zeta,t;k)$ has simple poles at each point in $Z\cup Z^*$ with:

    For $n\in\Delta_{k_0}^+$,
    \begin{align}
    \underset{k=k_n}{\rm Res\ }\mu^{(1)}(\zeta,t;k)&=\lim_{k\to k_n}\mu^{(1)}(\zeta,t;k)\begin{pmatrix}
    \mathbf{0}_{2\times2} &  \mathbf{0}_{2\times2}\\
    \delta_{2}^{-1}(k_n)C_n\delta_{1}^{-1}(k_n)T^{-2}(k_n)\e^{2\ii t\theta}& \mathbf{0}_{2\times2}
    \end{pmatrix},\label{3.11}\\
    \underset{k=-k_n^*}{\rm Res\ }\mu^{(1)}(\zeta,t;k)&=\lim_{k\to -k^*_n}\mu^{(1)}(\zeta,t;k)\begin{pmatrix}
    \mathbf{0}_{2\times2} &  \mathbf{0}_{2\times2}\\
    -\sigma_2[\delta_{2}^{-1}(k_n)C_n\delta_{1}^{-1}(k_n)]^*\sigma_2T^{-2}(-k^*_n)\e^{2\ii t\theta}& \mathbf{0}_{2\times2}
    \end{pmatrix},\label{3.12}\\
    \underset{k=k_n^*}{\rm Res\ }\mu^{(1)}(\zeta,t;k)&=\lim_{k\to k_n^*}\mu^{(1)}(\zeta,t;k)\begin{pmatrix}
    \mathbf{0}_{2\times2} & -[\delta_{2}^{-1}(k_n)C_n\delta_{1}^{-1}(k_n)]^\dag T^{2}(k_n^*)\e^{-2\ii t\theta} \\
    \mathbf{0}_{2\times2} & \mathbf{0}_{2\times2}
    \end{pmatrix},\label{3.13}\\
    \underset{k=-k_n}{\rm Res\ }\mu^{(1)}(\zeta,t;k)&=\lim_{k\to -k_n}\mu^{(1)}(\zeta,t;k)\begin{pmatrix}
    \mathbf{0}_{2\times2} & \sigma_2[\delta_{2}^{-1}(k_n)C_n\delta_{1}^{-1}(k_n)]^{\texttt{T}}\sigma_2T^{2}(-k_n)\e^{-2\ii t\theta} \\
    \mathbf{0}_{2\times2} & \mathbf{0}_{2\times2}
    \end{pmatrix}.\label{3.14}
    \end{align}

    For $n\in\Delta_{k_0}^-$,
    \begin{align}
    \underset{k=k_n}{\rm Res\ }\mu^{(1)}(\zeta,t;k)&=\lim_{k\to k_n}\mu^{(1)}(\zeta,t;k)\begin{pmatrix}
    \mathbf{0}_{2\times2} &  \delta_1(k_n)C_n^{-1}\delta_2(k_n)
(1/T)'(k_n)^{-2}\e^{-2\ii t\theta}\\
    \mathbf{0}_{2\times2}& \mathbf{0}_{2\times2}
    \end{pmatrix},\label{3.15}\\
    \underset{k=-k_n^*}{\rm Res\ }\mu^{(1)}(\zeta,t;k)&=\lim_{k\to -k^*_n}\mu^{(1)}(\zeta,t;k)\begin{pmatrix}
    \mathbf{0}_{2\times2} & -\sigma_2[\delta_1(k_n)C_n^{-1}\delta_2(k_n)]^*\sigma_2(1/T)'(-k_n^*)^{-2}\e^{-2\ii t\theta}\\
     \mathbf{0}_{2\times2}& \mathbf{0}_{2\times2}
    \end{pmatrix},\label{3.16}\\
    \underset{k=k_n^*}{\rm Res\ }\mu^{(1)}(\zeta,t;k)&=\lim_{k\to k_n^*}\mu^{(1)}(\zeta,t;k)\begin{pmatrix}
    \mathbf{0}_{2\times2} & \mathbf{0}_{2\times2} \\
    -[\delta_1(k_n)C_n^{-1}\delta_2(k_n)]^\dag
T'(k_n^*)^{-2}\e^{2\ii t\theta} & \mathbf{0}_{2\times2}
    \end{pmatrix},\label{3.17}\\
    \underset{k=-k_n}{\rm Res\ }\mu^{(1)}(\zeta,t;k)&=\lim_{k\to -k_n}\mu^{(1)}(\zeta,t;k)\begin{pmatrix}
    \mathbf{0}_{2\times2} & \mathbf{0}_{2\times2}\\
    \sigma_2[\delta_1(k_n)C_n^{-1}\delta_2(k_n)]^{\texttt{T}}\sigma_2T'(-k_n)^{-2}\e^{2\ii t\theta}  & \mathbf{0}_{2\times2}
    \end{pmatrix}.\label{3.18}
    \end{align}
\end{itemize}
\end{rhp}
\begin{proof}
The analyticity and asymptotic behavior of $\mu^{(1)}(\zeta,t;k)$ is directly form its definition \eqref{3.7} and the properties of $\breve{\mu}(\zeta,t;k)$. The jump matrix \eqref{3.10} follows from $J^{(1)}=T^{\Sigma_3}\tilde{\Delta}^{-1}_-\breve{J}\tilde{\Delta}_+T^{-\Sigma_3}$ and the factorizations in \eqref{3.1}. As for residues, since $T(k)$ is analytic at each $\pm k_n,\pm k_n^*$ for $n\in\Delta_{k_0}^+$, the residue conditions \eqref{3.11}-\eqref{3.14} at these poles are a result of \eqref{3.7} and the symmetries of $\delta_j(k)$ stated in \eqref{3.5}. For $n\in\Delta_{k_0}^-$, $T(k)$ has simple zeros at $k_n^*$, $-k_n$ and poles at $k_n$, $-k_n^*$ as $\text{Re}k_n\neq0$, also has a simple zero at $k_n^*$ and a pole at $k_n$ while $\text{Re}k_n=0$. Thus, by \eqref{3.7}, that is,
\berr
\mu^{(1)}_L(\zeta,t;k)=\breve{\mu}_L(\zeta,t;k)\delta_1^{-1}(k)T^{-1}(k),\
\mu^{(1)}_R(\zeta,t;k)=\breve{\mu}_R(\zeta,t;k)\delta_2(k)T(k),
\eerr
we learn that for $\text{Re}k_n\neq0$, $k_n$, $-k_n^*$ are no longer the poles of $\mu^{(1)}_L(\zeta,t;k)$ with $k_n^*$, $-k_n$ becoming the poles of it, while as $\text{Re}k_n\neq0$, $\mu^{(1)}_L(\zeta,t;k)$ has a removable singularity at $k_n$, but acquires a pole at $k_n^*$. And $\mu^{(1)}_R(\zeta,t;k)$ has opposite situation. At $k_n^*$, we then find by \eqref{2.40} that
\be
\begin{aligned}
\mu^{(1)}_R(\zeta,t;k_n^*)=&\lim_{k\to k_n^*}[\breve{\mu}_R(\zeta,t;k)\delta_2(k)T(k)]=
[\underset{k=k_n^*}{\rm Res\ }\breve{\mu}_R(\zeta,t;k)]\delta_2(k_n^*)T'(k_n^*)\\
=&-\breve{\mu}_L(\zeta,t;k_n^*)]C_n^\dag\delta_2(k_n^*)T'(k_n^*)\e^{-2\ii t\theta(k_n^*)}.
\end{aligned}
\ee
Therefore, we have
\be
\begin{aligned}
\underset{k=k_n^*}{\rm Res\ }\mu^{(1)}_L(\zeta,t;k)=&\breve{\mu}_L(\zeta,t;k_n^*)\delta_1^{-1}(k_n^*)
\underset{k=k_n^*}{\rm Res\ }T^{-1}(k)=\breve{\mu}_L(\zeta,t;k_n^*)\delta_1^{-1}(k_n^*)T'(k_n^*)^{-1}\\
=&-\mu^{(1)}_R(\zeta,t;k_n^*)\delta_2^{-1}(k_n^*)(C_n^\dag)^{-1}\delta_1^{-1}(k_n^*)
T'(k_n^*)^{-2}\e^{2\ii t\theta(k_n^*)}\\
=&-\mu^{(1)}_R(\zeta,t;k_n^*)[\delta_1(k_n)C_n^{-1}\delta_2(k_n)]^\dag
T'(k_n^*)^{-2}\e^{2\ii t\theta(k_n^*)},
\end{aligned}
\ee
from which \eqref{3.17} follows. At $k_n$, one has from \eqref{2.38}
\be
\begin{aligned}
\mu^{(1)}_L(\zeta,t;k_n)=&\lim_{k\to k_n}[\breve{\mu}_L(\zeta,t;k)\delta_2(k)T(k)]=
[\underset{k=k_n}{\rm Res\ }\breve{\mu}_L(\zeta,t;k)]\delta_1^{-1}(k_n)(1/T)'(k_n)\\
=&\breve{\mu}_R(\zeta,t;k_n)]C_n\delta_1^{-1}(k_n)(1/T)'(k_n)\e^{2\ii t\theta(k_n)}.
\end{aligned}
\ee
It follows that
\be
\begin{aligned}
\underset{k=k_n}{\rm Res\ }\mu^{(1)}_R(\zeta,t;k)=&\breve{\mu}_R(\zeta,t;k_n)\delta_2(k_n)
\underset{k=k_n}{\rm Res\ }T(k)=\breve{\mu}_R(\zeta,t;k_n)\delta_2(k_n)(1/T)'(k_n)^{-1}\\
=&\mu^{(1)}_L(\zeta,t;k_n)\delta_1(k_n)C_n^{-1}\delta_2(k_n)
(1/T)'(k_n)^{-2}\e^{-2\ii t\theta(k_n)},
\end{aligned}
\ee
from which \eqref{3.15} holds. The calculation of others is similar.
\end{proof}
\subsection{Contour deformation}
The next step is to make continuous extension for the jump matrix to remove the jump from the real axis. Here we just require the extension to be continuous but not necessarily analytic. The price we pay for this non-analytic transformation is that the new unknown matrix function has nonzero $\bar{\partial}$-derivatives inside the regions, and hence a  mixed $\bar{\partial}$-RH problem will be introduced.

To start with, let us define the contour as follows:
\begin{align}
L\doteq&\left\{k\in\bfC|k=k_0+k_0\alpha\e^{\frac{\pi\ii}{4}},
-\frac{\sqrt{2}}{2}\leq\alpha<\infty\right\}
\cup\left\{k\in\bfC|k=k_0\alpha\e^{-\frac{\pi\ii}{4}},
-\frac{\sqrt{2}}{2}<\alpha<\frac{\sqrt{2}}{2}\right\}\nn\\
&\cup\left\{k\in\bfC|k=-k_0+k_0\alpha\e^{\frac{\pi\ii}{4}},
-\infty<\alpha\leq\frac{\sqrt{2}}{2}\right\}.
\end{align}
Then, the complex plane $\bfC$ is split by $L$ and $L^*$ into ten regions $\{\Omega_l\}_{l=1}^{10}$, and for convenience, we write $L\cup L^*=\cup_{l=1}^{12}\Gamma_l\doteq\Gamma$, see Figure \ref{fig2}.
\begin{figure}[htbp]
 \centering
 \includegraphics[width=3in]{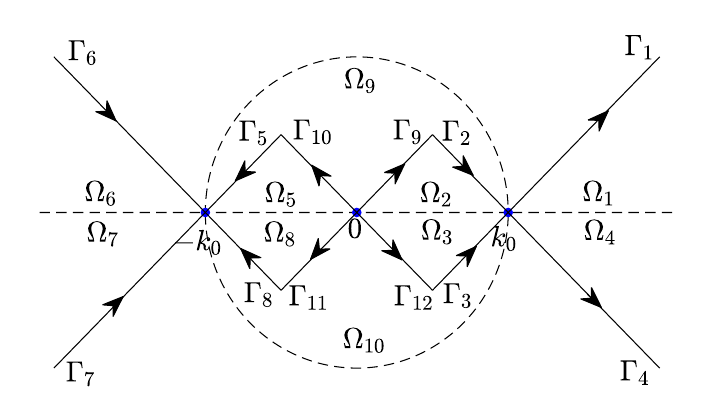}
  \caption{The open sets $\{\Omega_l\}_{l=1}^{10}$ and the contours $\{\Gamma_l\}_{l=1}^{12}$ in the complex $k$-plane.}\label{fig2}
\end{figure}

In addition, we define
\be
\varrho=\min\left\{\frac{1}{2}\min_{k\neq k'\in Z\cup Z^*}|k-k'|,~\text{dist}(Z,\bfR)\right\},
\ee
and the following smooth cut-off function
\be
\chi(k)=\left\{
\begin{aligned}
&1,\quad\text{dist}(k,Z\cup Z^*)<\varrho/4,\\
&0,\quad\text{dist}(k,Z\cup Z^*)>3\varrho/4.
\end{aligned}
\right.
\ee
Now, we introduce the $\bar{\partial}$ extensions, which aims to separate the phases and performs the contour deformation.
\begin{lemma}\label{lem3.1}
It is possible to define functions $R_j:\bar{\Omega}_j\mapsto\bfC$, $j=1,2,\cdots,8$ with boundary values satisfying
\begin{align}
R_1(k)&=\left\{
\begin{aligned}
& -\delta_{1}(k)\left(\mathbb{I}_{2\times2}+\rho^\dag(k) \rho(k)\right)^{-1}\rho^\dag(k)\delta_{2}(k)T^2(k),\qquad\qquad\qquad\quad\ k>k_0,\\
&-\delta_{1}(k)\left(\mathbb{I}_{2\times2}+\rho^\dag(k_0) \rho(k_0)\right)^{-1}\rho^\dag(k_0)\delta_{2}(k)T^2(k_0)(1-\chi(k)),\quad k\in\Gamma_1,
\end{aligned}
\right.\label{3.26}\\
R_2(k)&=\left\{
\begin{aligned}
&-\delta_{2}^{-1}(k)\rho(k)\delta_{1}^{-1}(k)T^{-2}(k),\qquad\qquad\quad\quad\,\ 0<k<k_0,\\
&-\delta_{2}^{-1}(k)\rho(k_0)\delta_{1}^{-1}(k)T^{-2}(k_0)(1-\chi(k)),\quad k\in\Gamma_2,
\end{aligned}
\right.\label{3.27}\\
R_3(k)&=\left\{
\begin{aligned}
&\delta_{1}(k)\rho^\dag(k)\delta_{2}(k)T^2(k),\qquad\qquad\qquad\ 0<k<k_0,\\
&\delta_{1}(k)\rho^\dag(k_0)\delta_{2}(k)T^2(k_0)(1-\chi(k)),\quad k\in\Gamma_3,
\end{aligned}
\right.\\
R_4(k)&=\left\{
\begin{aligned}
&\delta_{2}^{-1}(k)\rho(k)\left(\mathbb{I}_{2\times2}+\rho^\dag(k) \rho(k)\right)^{-1}\delta_{1}^{-1}(k)T^{-2}(k),\qquad\qquad\qquad\quad\ k>k_0,\\
&\delta_{2}^{-1}(k)\rho(k_0)\left(\mathbb{I}_{2\times2}+\rho^\dag(k_0) \rho(k_0)\right)^{-1}\delta_{1}^{-1}(k)T^{-2}(k_0)(1-\chi(k)), \quad k\in\Gamma_4,
\end{aligned}
\right.\\
R_5(k)&=\left\{
\begin{aligned}
&-\delta_{2}^{-1}(k)\rho(k)\delta_{1}^{-1}(k)T^{-2}(k),\qquad\qquad\qquad\quad\quad -k_0<k<0,\\
&-\delta_{2}^{-1}(k)\rho(-k_0)\delta_{1}^{-1}(k)T^{-2}(-k_0)(1-\chi(k)),\quad k\in\Gamma_5,
\end{aligned}
\right.\\
R_6(k)&=\left\{
\begin{aligned}
&-\delta_{1}(k)\left(\mathbb{I}_{2\times2}+\rho^\dag(k) \rho(k)\right)^{-1}\rho^\dag(k)\delta_{2}(k)T^2(k),\qquad\qquad\qquad\qquad\qquad k<-k_0,\\
&-\delta_{1}(k)\left(\mathbb{I}_{2\times2}+\rho^\dag(-k_0) \rho(-k_0)\right)^{-1}\rho^\dag(-k_0)\delta_{2}(k)T^2(-k_0)(1-\chi(k)),\,\ k\in\Gamma_6,
\end{aligned}
\right.\\
R_7(k)&=\left\{
\begin{aligned}
&\delta_{2}^{-1}(k)\rho(k)\left(\mathbb{I}_{2\times2}+\rho^\dag(k) \rho(k)\right)^{-1}\delta_{1}^{-1}(k)T^{-2}(k),\qquad\qquad\qquad\qquad\quad\,\ k<-k_0,\\
&\delta_{2}^{-1}(k)\rho(-k_0)\left(\mathbb{I}_{2\times2}+\rho^\dag(-k_0) \rho(-k_0)\right)^{-1}\delta_{1}^{-1}(k)T^{-2}(-k_0)(1-\chi(k)),\,\ k\in\Gamma_7,
\end{aligned}
\right.\\
R_8(k)&=\left\{
\begin{aligned}
&\delta_{1}(k)\rho^\dag(k)\delta_{2}(k)T^2(k),\qquad\qquad\qquad\qquad -k_0<k<0,\\
&\delta_{1}(k)\rho^\dag(-k_0)\delta_{2}(k)T^2(-k_0)(1-\chi(k)),\quad k\in\Gamma_8.
\end{aligned}
\right.
\end{align}
Moreover, $R_j$ admit estimates
\be\label{3.34}
|\bar{\partial}R_j(k)|\leq c_1|\rho'(\text{Re}k)|+c_2|k\mp k_0|^{-1/2}+c_3\bar{\partial}\chi(k),
\ee
for positive constants $c_1,~c_2$ and $c_3$ depended on $\|\rho\|_{H^1(\bfR)}$.
\end{lemma}
\begin{proof}
Without loss of generality, we only provide the detailed proof for $R_2$, as other cases can be proved in a similar way. Define the function
\be
f(k)=\rho(k_0)T^{2}(k)T^{-2}(k_0).
\ee
Then, we can define the extension for $k\in\bar{\Omega}_2$ as follows:
\begin{align}
R_2(k)=-\delta_{2}^{-1}(k)\left[f(k)+\left(\rho(\text{Re}k)-f(k)\right)\cos(2\phi)\right]
\delta_{1}^{-1}(k)T^{-2}(k)(1-\chi(k)),\ k\in\bar{\Omega}_2.
\end{align}
where $\phi=\arg(k-k_0)$. Clearly, $R_2(k)$ satisfies the boundary values \eqref{3.27} as $\cos(2\phi)$ vanishes on $\Gamma_2$ and $\chi(k)$ is zero on the real axis. Let $k-k_0=s\e^{\ii\phi}$. It follows from
\begin{align*}
&\bar{\partial}=\frac{1}{2}\left(\frac{\partial}{\partial k_1}+\ii\frac{\partial}{\partial k_2}\right)
=\frac{1}{2}\e^{\ii\phi}\left(\frac{\partial}{\partial s}+\frac{\ii}{s}\frac{\partial}{\partial\phi}\right),\ \prod_{\substack{\text{Re}k_n=0,\text{Im}k_n>0\\n\in\Delta_{k_0}^-}}
\frac{k-k_n^*}{k-k_n}\frac{k_0-k_n^*}{k_0-k_n}=1+O(k-k_0),
\end{align*}
and
\berr
\left|\rho(\text{Re}k)-\rho(k_0)\right|\leq\left|\int_{k_0}^{\text{Re}k}\rho'(s)\dd s\right|\leq\|\rho\|_{H^1(\bfR)}|k-k_0|^{1/2}
\eerr
that
\begin{align}
|\bar{\partial}R_2(k)|=&\left|\delta_{2}^{-1}(k)\left[\frac{1}{2}\rho'(\text{Re}k)\cos(2\phi)
-\ii\e^{\ii\phi}\frac{\rho(\text{Re}k)-f(k)}{|k-k_0|}
\sin(2\phi)\right]\delta^{-1}_1(k)T^{-2}(k)(1-\chi(k))\right.\nn\\
&\left.+\delta_{2}^{-1}(k)\left[f(k)+\left(\rho(\text{Re}k)-f(k)\right)\cos(2\phi)\right]
\delta_{1}^{-1}(k)T^{-2}(k)\bar{\partial}\chi(k)\right|\\
\leq& c_1|\rho'(\text{Re}k)|+c_2|k-k_0|^{-1/2}+c_3\bar{\partial}\chi(k).\nn
\end{align}
\end{proof}
Now, we construct a matrix function $\mathcal{R}^{(2)}(k)$ by
\be\label{3.38}
\mathcal{R}^{(2)}(k)=\left\{
\begin{aligned}
&\begin{pmatrix}
\mathbb{I}_{2\times2} & R_j(k)\e^{-2\ii t\theta(k)}\\
\mathbf{0}_{2\times2} & \mathbb{I}_{2\times2}
\end{pmatrix},\  k\in\Omega_j,\,j=1,3,6,8,\\
&\begin{pmatrix}
\mathbb{I}_{2\times2} & \mathbf{0}_{2\times2}\\
R_j(k)\e^{2\ii t\theta(k)} & \mathbb{I}_{2\times2}
\end{pmatrix},\,\,\,\  k\in\Omega_j,\,j=2,4,5,7,\\
&\mathbb{I}_{4\times4},\qquad\qquad\qquad\quad\quad\,\,\ k\in\Omega_9\cup\Omega_{10}.
\end{aligned}
\right.
\ee
We then can define a new unknown matrix-valued function $\mu^{(2)}$ by
\be
\mu^{(2)}(\zeta,t;k)=\mu^{(1)}(\zeta,t;k)\mathcal{R}^{(2)}(k),
\ee
and it follows that $\mu^{(2)}$ satisfies a mixed $\bar{\partial}$-Riemann--Hilbert problem.

\begin{dbarrhp}\label{dbarrh3.1}
Find a $4\times4$ matrix-valued function $\mu^{(2)}(\zeta,t;k)$ on  $\bfC\setminus(\Gamma\cup Z\cup Z^*)$ with the following properties:
\begin{itemize}
\item Analyticity: $\mu^{(2)}(\zeta,t;k)$ is continuous with sectionally continuous first partial derivatives in $\bfC\setminus(\Gamma\cup Z\cup Z^*)$.

\item Jump condition: For $k\in\Gamma$, the continuous boundary values $\mu^{(2)}_\pm$ satisfy the jump relation
\be
\mu^{(2)}_+(\zeta,t;k)=\mu^{(2)}_-(\zeta,t;k)J^{(2)}(\zeta,t;k),
\ee
where the jump matrix $J^{(2)}(\zeta,t;k)$ is given by
\bea\label{3.41}
J^{(2)}=\left\{
\begin{aligned}
&\begin{pmatrix}
\mathbb{I}_{2\times2} & -R_1(k)\e^{-2\ii t\theta(k)}\\
\mathbf{0}_{2\times2} & \mathbb{I}_{2\times2}
\end{pmatrix},\quad k\in\Gamma_1,\\
&\begin{pmatrix}
\mathbb{I}_{2\times2} & \mathbf{0}_{2\times2}\\
-R_2(k)\e^{2\ii t\theta(k)} & \mathbb{I}_{2\times2}
\end{pmatrix},\quad\, \, \  k\in\Gamma_2\cup\Gamma_9,\\
&\begin{pmatrix}
\mathbb{I}_{2\times2} & R_3(k)\e^{-2\ii t\theta(k)}\\
\mathbf{0}_{2\times2} & \mathbb{I}_{2\times2}
\end{pmatrix},\qquad k\in\Gamma_3\cup\Gamma_{12},\\
&\begin{pmatrix}
\mathbb{I}_{2\times2} & \mathbf{0}_{2\times2}\\
R_4(k)\e^{2\ii t\theta(k)} & \mathbb{I}_{2\times2}
\end{pmatrix}, \,\ \qquad\ k\in\Gamma_4,\\
&\begin{pmatrix}
\mathbb{I}_{2\times2} &\mathbf{0}_{2\times2}\\
R_5(k)\e^{2\ii t\theta(k)} & \mathbb{I}_{2\times2}
\end{pmatrix},\qquad \,\,\,\ k\in\Gamma_5\cup\Gamma_{10},\\
&\begin{pmatrix}
\mathbb{I}_{2\times2} & -R_6(k)\e^{-2\ii t\theta(k)}\\
\mathbf{0}_{2\times2} & \mathbb{I}_{2\times2}
\end{pmatrix},\quad\,\ k\in\Gamma_6,\\
&\begin{pmatrix}
\mathbb{I}_{2\times2} &\mathbf{0}_{2\times2}\\
R_7(k)\e^{2\ii t\theta(k)} & \mathbb{I}_{2\times2}
\end{pmatrix},\quad\qquad k\in\Gamma_7,\\
&\begin{pmatrix}
\mathbb{I}_{2\times2} & -R_8(k)\e^{-2\ii t\theta(k)}\\
\mathbf{0}_{2\times2} & \mathbb{I}_{2\times2}
\end{pmatrix},\quad\,\,\,\,  k\in\Gamma_8\cup\Gamma_{11}.
\end{aligned}
\right.
\eea
\item Normalization: $\mu^{(2)}(\zeta,t;k)\to\mathbb{I}_{4\times4}$, as $k\to\infty$.
\item Residue conditions: $\mu^{(2)}(\zeta,t;k)$ has simple poles at each point in $Z\cup Z^*$ with:

    For $n\in\Delta_{k_0}^+$,
    \begin{align}
    \underset{k=k_n}{\rm Res\ }\mu^{(2)}(\zeta,t;k)&=\lim_{k\to k_n}\mu^{(2)}(\zeta,t;k)\begin{pmatrix}
    \mathbf{0}_{2\times2} &  \mathbf{0}_{2\times2}\\
    \delta_{2}^{-1}(k_n)C_n\delta_{1}^{-1}(k_n)T^{-2}(k_n)\e^{2\ii t\theta}& \mathbf{0}_{2\times2}
    \end{pmatrix},\label{3.42}\\
    \underset{k=-k_n^*}{\rm Res\ }\mu^{(2)}(\zeta,t;k)&=\lim_{k\to -k^*_n}\mu^{(2)}(\zeta,t;k)\begin{pmatrix}
    \mathbf{0}_{2\times2} &  \mathbf{0}_{2\times2}\\
    -\sigma_2[\delta_{2}^{-1}(k_n)C_n\delta_{1}^{-1}(k_n)]^*\sigma_2T^{-2}(-k^*_n)\e^{2\ii t\theta}& \mathbf{0}_{2\times2}
    \end{pmatrix},\label{3.43}\\
    \underset{k=k_n^*}{\rm Res\ }\mu^{(2)}(\zeta,t;k)&=\lim_{k\to k_n^*}\mu^{(2)}(\zeta,t;k)\begin{pmatrix}
    \mathbf{0}_{2\times2} & -[\delta_{2}^{-1}(k_n)C_n\delta_{1}^{-1}(k_n)]^\dag T^{2}(k_n^*)\e^{-2\ii t\theta} \\
    \mathbf{0}_{2\times2} & \mathbf{0}_{2\times2}
    \end{pmatrix},\label{3.44}\\
    \underset{k=-k_n}{\rm Res\ }\mu^{(2)}(\zeta,t;k)&=\lim_{k\to -k_n}\mu^{(2)}(\zeta,t;k)\begin{pmatrix}
    \mathbf{0}_{2\times2} & \sigma_2[\delta_{2}^{-1}(k_n)C_n\delta_{1}^{-1}(k_n)]^{\texttt{T}}\sigma_2T^{2}(-k_n)\e^{-2\ii t\theta} \\
    \mathbf{0}_{2\times2} & \mathbf{0}_{2\times2}
    \end{pmatrix}.\label{3.45}
    \end{align}

    For $n\in\Delta_{k_0}^-$,
    \begin{align}
    \underset{k=k_n}{\rm Res\ }\mu^{(2)}(\zeta,t;k)&=\lim_{k\to k_n}\mu^{(2)}(\zeta,t;k)\begin{pmatrix}
    \mathbf{0}_{2\times2} &  \delta_1(k_n)C_n^{-1}\delta_2(k_n)
(1/T)'(k_n)^{-2}\e^{-2\ii t\theta}\\
    \mathbf{0}_{2\times2}& \mathbf{0}_{2\times2}
    \end{pmatrix},\label{3.46}\\
    \underset{k=-k_n^*}{\rm Res\ }\mu^{(2)}(\zeta,t;k)&=\lim_{k\to -k^*_n}\mu^{(2)}(\zeta,t;k)\begin{pmatrix}
    \mathbf{0}_{2\times2} & -\sigma_2[\delta_1(k_n)C_n^{-1}\delta_2(k_n)]^*\sigma_2(1/T)'(-k_n^*)^{-2}\e^{-2\ii t\theta}\\
     \mathbf{0}_{2\times2}& \mathbf{0}_{2\times2}
    \end{pmatrix},\label{3.47}\\
    \underset{k=k_n^*}{\rm Res\ }\mu^{(2)}(\zeta,t;k)&=\lim_{k\to k_n^*}\mu^{(2)}(\zeta,t;k)\begin{pmatrix}
    \mathbf{0}_{2\times2} & \mathbf{0}_{2\times2} \\
    -[\delta_1(k_n)C_n^{-1}\delta_2(k_n)]^\dag
T'(k_n^*)^{-2}\e^{2\ii t\theta} & \mathbf{0}_{2\times2}
    \end{pmatrix},\label{3.48}\\
    \underset{k=-k_n}{\rm Res\ }\mu^{(2)}(\zeta,t;k)&=\lim_{k\to -k_n}\mu^{(2)}(\zeta,t;k)\begin{pmatrix}
    \mathbf{0}_{2\times2} & \mathbf{0}_{2\times2}\\
    \sigma_2[\delta_1(k_n)C_n^{-1}\delta_2(k_n)]^{\texttt{T}}\sigma_2T'(-k_n)^{-2}\e^{2\ii t\theta}  & \mathbf{0}_{2\times2}
    \end{pmatrix}.\label{3.49}
    \end{align}

\item $\bar{\partial}$-Derivative: For $k\in\bfC\setminus(\Gamma\cup Z\cup Z^*)$, we have
\be\label{3.50}
\bar{\partial}\mu^{(2)}(\zeta,t;k)=\mu^{(2)}(\zeta,t;k)\bar{\partial}\mathcal{R}^{(2)}(k),
\ee
where
\be\label{3.51}
\bar{\partial}\mathcal{R}^{(2)}(k)=\left\{
\begin{aligned}
&\begin{pmatrix}
\mathbb{I}_{2\times2} & \bar{\partial}R_j(k)\e^{-2\ii t\theta(k)}\\
\mathbf{0}_{2\times2} & \mathbb{I}_{2\times2}
\end{pmatrix},\  k\in\Omega_j,\,j=1,3,6,8,\\
&\begin{pmatrix}
\mathbb{I}_{2\times2} & \mathbf{0}_{2\times2}\\
\bar{\partial}R_j(k)\e^{2\ii t\theta(k)} & \mathbb{I}_{2\times2}
\end{pmatrix},\,\,\  k\in\Omega_j,\,j=2,4,5,7,\\
&\mathbf{0}_{4\times4},\qquad\qquad\qquad\qquad\quad k\in\Omega_9\cup\Omega_{10}.
\end{aligned}
\right.
\ee
\end{itemize}
\end{dbarrhp}
\subsection{Decomposition of the mixed $\bar{\partial}$-RH problem}
To solve the $\bar{\partial}$-RH problem \ref{dbarrh3.1}, we decompose it to a localized RH problem with $\bar{\partial}\mathcal{R}^{(2)}(k)=0$ and a pure $\bar{\partial}$-problem with $\bar{\partial}\mathcal{R}^{(2)}(k)\neq0$.
Now, we establish a pure RH problem part $\mu^{(RHP)}(\zeta,t;k)$, which corresponds to the following RH problem.
\begin{rhp}\label{rh3.2}
Find a $4\times4$ matrix-valued function $\mu^{(RHP)}(\zeta,t;k)$ on  $\bfC\setminus(\Gamma\cup Z\cup Z^*)$ with the following properties:
\begin{itemize}
\item Analyticity: $\mu^{(RHP)}(\zeta,t;k)$ is analytic in $\bfC\setminus(\Gamma\cup Z\cup Z^*)$.

\item Jump condition: $\mu^{(RHP)}(\zeta,t;k)$ has the continuous boundary values $\mu^{(RHP)}_\pm$ on $\Gamma$, and
\be
\mu^{(RHP)}_+(\zeta,t;k)=\mu^{(RHP)}_-(\zeta,t;k)J^{(2)}(\zeta,t;k).
\ee

\item Normalization: $\mu^{(RHP)}(\zeta,t;k)\to\mathbb{I}_{4\times4}$, as $k\to\infty$.
\item Residue conditions: $\mu^{(RHP)}(\zeta,t;k)$ has simple poles at each point in $Z\cup Z^*$ with the same residue conditions as those in the $\bar{\partial}$-RH problem \ref{dbarrh3.1}, specified in \eqref{3.42}-\eqref{3.49}.
\end{itemize}
\end{rhp}
The existence and asymptotics of $\mu^{(RHP)}(\zeta,t;k)$ will be discussed in next subsection \ref{sec3.4}. Now, we suppose that the solution $\mu^{(RHP)}(\zeta,t;k)$ exists, and perform the factorization
\be\label{3.55}
\mu^{(2)}(\zeta,t;k)=\mu^{(3)}(\zeta,t;k)\mu^{(RHP)}(\zeta,t;k),
\ee
which results in $\mu^{(3)}(\zeta,t;k)$ corresponding to the solution of a pure $\bar{\partial}$-problem without jumps and poles.
\begin{dbar}\label{dbar1}
Find a $4\times4$ matrix-valued function $\mu^{(3)}(\zeta,t;k)$ with the following properties:
\begin{itemize}
\item Analyticity: $\mu^{(3)}(\zeta,t;k)$ is continuous with sectionally continuous first partial derivatives in $\bfC\setminus\Sigma$.

\item $\bar{\partial}$-Derivative: For $k\in\bfC\setminus\Sigma$, we have
\be
\bar{\partial}\mu^{(3)}(\zeta,t;k)=\mu^{(3)}(\zeta,t;k)w^{(3)}(\zeta,t;k),
\ee
where
\be
w^{(3)}(\zeta,t;k)=\mu^{(RHP)}(\zeta,t;k)\bar{\partial}
\mathcal{R}^{(2)}(k)[\mu^{(RHP)}(\zeta,t;k)]^{-1}
\ee
\item Normalization: As $k\rightarrow\infty$, $\mu^{(3)}(\zeta,t;k)\rightarrow\mathbb{I}_{4\times4}.$
\end{itemize}
\end{dbar}
In the following, we will respectively study the solution $\mu^{(RHP)}$ of RH problem \ref{rh3.2} and solution $\mu^{(3)}$ of the pure $\bar{\partial}$-problem \ref{dbar1}. For RH problem \ref{rh3.2}, we will establish the existence and asymptotic expansion of the solution. For the pure $\bar{\partial}$-problem, efforts are put to show that it has a solution, the solution will decay rapidly as $t\to\infty$ and only contribute to an error term with higher-order decay rate.
\subsection{Analysis on the pure RH problem}\label{sec3.4}
The current subsection focuses on finding $\mu^{(RHP)}$. For a start, let $U_{\pm k_0}$ be the neighborhoods of $\pm k_0$, respectively
\be\label{3.56}
U_{\pm k_0}=\left\{k\in\bfC||k\mp k_0|\leq\text{min}\{k_0/2,\varrho/4\}\doteq\epsilon\right\}.
\ee
Then, we find that the jump matrix for the RH problem \ref{rh3.2} admits the following estimates.
\begin{proposition}\label{prop3.1}
The jump matrix $J^{(2)}(\zeta,t;k)$ satisfies the following estimates:
\begin{align}
&\|J^{(2)}-\mathbb{I}_{4\times4}\|_{L^\infty(\Gamma_{\pm k_0}\setminus U_{\pm k_0})}\leq c\e^{-\frac{\sqrt{2}\epsilon t}{16k_0}},\label{3.57}\\
&\|J^{(2)}-\mathbb{I}_{4\times4}\|_{L^\infty(\Gamma_{0}\setminus U_0)}\leq c\e^{-\frac{\sqrt{2}\epsilon t}{4k_0}},\label{3.58}
\end{align}
where $\Gamma_{k_0}=\cup_{l=1}^4\Gamma_l$, $\Gamma_{-k_0}=\cup_{l=5}^8\Gamma_l$, $\Gamma_0=\cup_{l=9}^{12}\Gamma_l$  and $U_0=\{k\in\bfC||k|\leq\epsilon\}$.
\end{proposition}
\begin{proof}
Without loss of generality, we prove \eqref{3.57} only for the case $k\in\Gamma_1$ and \eqref{3.58} for $k\in\Gamma_9$. For $k=k_0+k_0\alpha\e^{\pi\ii/4}$ with $\epsilon<\alpha<\infty$, we find
\begin{align}
\text{Re}(-2\ii t\theta(k))=2t\text{Im}k\left(\frac{1}{4|k|^2}-\frac{1}{4k_0^2}\right)
=-\frac{\sqrt{2}t}{4}|k-k_0|(k_0^{-2}-|k|^{-2})
\leq-\frac{\sqrt{2}\epsilon t}{16k_0}.
\end{align}
Thus, by using the definition of $J^{(2)}$ and the boundedness of $R_1$, we get
\be
\|J^{(2)}-\mathbb{I}_{4\times4}\|_{L^\infty(\Gamma_1\setminus U_{k_0})}\leq\|R_1(k)\e^{-2\ii t\theta(k)}\|_{L^\infty(\Gamma_1\setminus U_{k_0})}\leq c\e^{-\frac{t}{4k_0}\frac{\epsilon^2}{\epsilon^2+1}}.
\ee
On the other hand, when $k\in\Gamma_9\setminus U_0$, that is, $k=k_0\alpha\e^{\pi\ii/4}$ with $\epsilon<\alpha<\sqrt{2}/2$, we have
\berr
\text{Re}(2\ii t\theta(k))=-\sqrt{2}k_0t\alpha\left(-\frac{1}{4k_0^2}+\frac{1}{4k_0^2\alpha^2}\right)\leq
-\frac{\sqrt{2}t\epsilon}{4k_0},
\eerr
which gives the corresponding estimate in \eqref{3.58}.
\end{proof}
This proposition implies that the jump matrix $J^{(2)}$ uniformly tends to $\mathbb{I}_{4\times4}$ on both $\Gamma_{\pm k_0}\setminus U_{\pm k_0}$ and $\Gamma_{0}\setminus U_0$ as $t\to\infty$. Thus, if we omit the jump conditions outside the $U_{\pm k_0}$ and $U_0$ of $\mu^{(RHP)}$, there only exists exponentially small error with respect to $t$. Moreover, by Proposition \ref{pro2.5}, we note that in the neighborhood $U_0$ of $k=0$, $J^{(2)}\to\mathbb{I}_{4\times4}$ as $k\to0$, and hence, the study of the neighborhood of $k=0$ alone is not necessary.

We therefore construct the solution $\mu^{(RHP)}$ of the RH problem \ref{rh3.2} in the following form
\be\label{3.60}
\mu^{(RHP)}(\zeta,t;k)=\left\{\begin{aligned}
&E(k)\mu^{(out)}(\zeta,t;k),\qquad\qquad\qquad\,\,\, k\in\bfC\setminus U_{\pm k_0},\\
&E(k)\mu^{(out)}(\zeta,t;k)\mu^{(k_0)}(\zeta,t;k),\quad\,\,\ k\in U_{k_0},\\
&E(k)\mu^{(out)}(\zeta,t;k)\mu^{(-k_0)}(\zeta,t;k),\quad k\in U_{-k_0},
\end{aligned}
\right.
\ee
where $\mu^{(out)}$ satisfies a model RH problem obtained by ignoring the jump conditions of RH problem \ref{rh3.2}, which is defined in $\bfC$ and has only discrete spectrum with no jumps. While $\mu^{(\pm k_0)}$ are the model RH problems which exactly match the jumps of $\mu^{(RHP)}$ in $U_{\pm k_0}$, respectively. The remainder $E(k)$ is a error function, which is a solution of a small-norm RH problem.
\begin{remark}
The motivation for constructing $\mu^{(RHP)}$ of the form \eqref{3.60} is that $\mu^{(RHP)}$ has no poles in $U_{\pm k_0}$, since dist$(Z,\bfR)>\varrho$ and \eqref{3.56}. Thus, we separate the jumps and poles into two parts.
\end{remark}
\begin{figure}[htbp]
  \centering
  \subfigure[]{\includegraphics[width=3.0in]{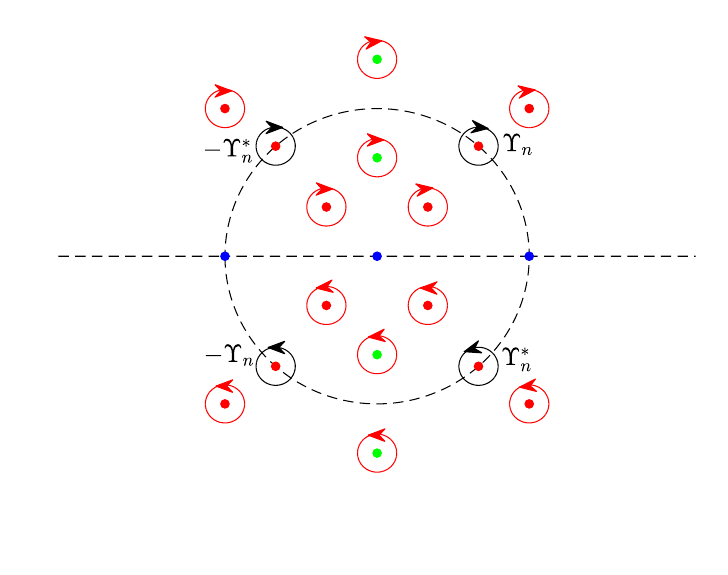}}
  \hfil
\subfigure[]{\includegraphics[width=3.0in]{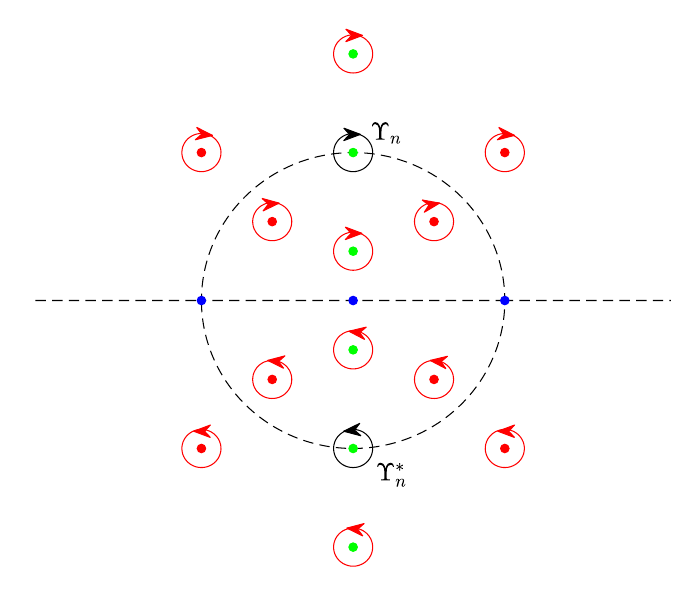}}
\caption{The contour $\Upsilon$, (a) Re$k_n>0$, (b) Re$k_n=0$. $J^{(out)}$ decays exponentially on red contours.}\label{fig3}
\end{figure}
To facilitate the analysis, it is more convenient to transform the residue conditions at the poles into jump conditions. We begin with some notations.
Let $\Upsilon_n$ be a small circle centered at $k_n$ for each $n$ with sufficiently small radius such that it lies inside the upper half plane and is disjoint from all other circles. We assume the small circles around $k_n$ and $-k_n^*$ are oriented clockwise and around $k^*_n$ and $-k_n$ are oriented counterclockwise, see Figure \ref{fig3}. Denote
\be
\Upsilon\doteq\left(\bigcup_{n=1}^N\pm\Upsilon_n\right)
\cup\left(\bigcup_{n=1}^N\pm\Upsilon_n^*\right).
\ee
By doing so, we will replace the residue conditions \eqref{3.42}-\eqref{3.49} of the RH problem with Schwarz invariant jump conditions across closed contours.
\subsubsection{The outer RH model }
We now establish the outer model RH problem for $\mu^{(out)}$.
\begin{rhp}\label{rh3.3}
Find a $4\times4$ matrix-valued function $\mu^{(out)}(\zeta,t;k)$ on $\bfC\setminus\Upsilon$ with the following properties:
\begin{itemize}
\item Analyticity: $\mu^{(out)}(\zeta,t;k)$ is analytic in $\bfC\setminus\Upsilon$.
\item Jump condition: The jump relation of the continuous boundary values $\mu^{out}_\pm$ on $\Upsilon$ is
    \be
    \mu^{(out)}_+(\zeta,t;k)=\mu^{(out)}_-(\zeta,t;k)J^{(out)}(\zeta,t;k),
    \ee
    where
\begin{align}
J^{(out)}=
\left\{
\begin{aligned}
&\begin{pmatrix}
\mathbb{I}_{2\times2} &  \mathbf{0}_{2\times2}\\
\delta_{2}^{-1}(k_n)C_n\delta_{1}^{-1}(k_n)\frac{T^{-2}(k_n)}{k-k_n}\e^{2\ii t\theta(k_n)} & \mathbb{I}_{2\times2}
\end{pmatrix},\  k\in\Upsilon_n, \ n\in\Delta_{k_0}^+,\\
&\begin{pmatrix}
\mathbb{I}_{2\times2} &  \mathbf{0}_{2\times2}\\
-\sigma_2[\delta_{2}^{-1}(k_n)C_n\delta_{1}^{-1}(k_n)]^*\sigma_2
\frac{T^{-2}(-k^*_n)}{k+k_n^*}\e^{2\ii t\theta(-k^*_n)} & \mathbb{I}_{2\times2}
\end{pmatrix},\  k\in-\Upsilon^*_n, \ n\in\Delta_{k_0}^+,\\
&\begin{pmatrix}
\mathbb{I}_{2\times2} & [\delta_{2}^{-1}(k_n)C_n\delta_{1}^{-1}(k_n)]^\dag \frac{T^{2}(k_n^*)}{k-k_n^*}\e^{-2\ii t\theta(k_n^*)} \\
\mathbf{0}_{2\times2} & \mathbb{I}_{2\times2}
    \end{pmatrix},\  k\in\Upsilon^*_n, \ n\in\Delta_{k_0}^+,\\
&\begin{pmatrix}
    \mathbb{I}_{2\times2} & \sigma_2[\delta_{2}^{-1}(k_n)C_n\delta_{1}^{-1}(k_n)]^{\texttt{T}}\sigma_2
    \frac{T^{2}(-k_n)}{k+k_n}\e^{-2\ii t\theta(-k_n)} \\
    \mathbf{0}_{2\times2} & \mathbb{I}_{2\times2}
    \end{pmatrix},\  k\in-\Upsilon_n, \ n\in\Delta_{k_0}^+, \\
&\begin{pmatrix}
    \mathbb{I}_{2\times2} &  \delta_1(k_n)C_n^{-1}\delta_2(k_n)
\frac{(1/T)'(k_n)^{-2}}{k-k_n}\e^{-2\ii t\theta(k_n)}\\
    \mathbf{0}_{2\times2}& \mathbb{I}_{2\times2}
    \end{pmatrix},\  k\in\Upsilon_n,\ n\in\Delta_{k_0}^-,\\
& \begin{pmatrix}
    \mathbb{I}_{2\times2} & -\sigma_2[\delta_1(k_n)C_n^{-1}\delta_2(k_n)]^*\sigma_2
    \frac{(1/T)'(-k_n^*)^{-2}}{k+k_n^*}\e^{-2\ii t\theta(-k_n^*)}\\
     \mathbf{0}_{2\times2}& \mathbb{I}_{2\times2}
    \end{pmatrix}, \ k\in-\Upsilon^*_n,\ n\in\Delta_{k_0}^-,\\
&\begin{pmatrix}
    \mathbb{I}_{2\times2} & \mathbf{0}_{2\times2} \\
    [\delta_1(k_n)C_n^{-1}\delta_2(k_n)]^\dag
\frac{T'(k_n^*)^{-2}}{k-k_n^*}\e^{2\ii t\theta(k_n^*)} & \mathbb{I}_{2\times2}
    \end{pmatrix}, \  k\in\Upsilon_n^*,\ n\in\Delta_{k_0}^-,\\
&\begin{pmatrix}
    \mathbb{I}_{2\times2} & \mathbf{0}_{2\times2}\\
    -\sigma_2[\delta_1(k_n)C_n^{-1}\delta_2(k_n)]^{\texttt{T}}\sigma_2
    \frac{T'(-k_n)^{-2}}{k+k_n}\e^{2\ii t\theta(-k_n)}  & \mathbb{I}_{2\times2}
    \end{pmatrix},\  k\in-\Upsilon_n,\ n\in\Delta_{k_0}^-.
\end{aligned}
\right.
\end{align}
\item Normalization: $\mu^{(out)}(\zeta,t;k)\to\mathbb{I}_{4\times4}$, as $k\to\infty$.
\end{itemize}
\end{rhp}

From the signature table Figure \ref{fig1}, we observe that along the characteristic line $\zeta=vt$ where $v=-1/(4|k_n|^2)$, by choosing the radius of each element of $\Upsilon$ small enough, we have for $k\in\Upsilon\setminus(\pm\Upsilon_n\cup\pm\Upsilon_n^*)$
\be\label{3.64}
|J^{(out)}-\mathbb{I}_{4\times4}|\leq c\e^{-ct},\quad t\to\infty.
\ee
\begin{proposition}\label{prop3.2}
There exists a $4\times4$ matrix-valued function $E_*(\zeta,t;k)$ with
\be
E_*(\zeta,t;k)=\mathbb{I}_{4\times4}+O(\e^{-ct})
\ee
such that
\be
\mu^{(out)}(\zeta,t;k)=E_*(\zeta,t;k)\mu^{(out)}_*(\zeta,t;k),
\ee
where $\mu^{(out)}_*(\zeta,t;k)$ solves the RH problem:
\begin{rhp}\label{rh3.4}
Find a $4\times4$ matrix-valued function $\mu_*^{(out)}(\zeta,t;k)$ on $\bfC\setminus(\pm\Upsilon_n\cup\pm\Upsilon_n^*)$ with the following properties:
\begin{itemize}
\item Analyticity: $\mu^{(out)}(\zeta,t;k)$ is analytic for $k\in\bfC\setminus(\pm\Upsilon_n\cup\pm\Upsilon_n^*)$ with continuous boundary values $\mu_{*\pm}^{(out)}$.
\item Jump condition: On $\pm\Upsilon_n\cup\pm\Upsilon_n^*$, we have jump relation
    \be
    \mu_{*+}^{(out)}(\zeta,t;k)=\mu_{*-}^{(out)}(\zeta,t;k)J_*^{(out)}(\zeta,t;k),
    \ee
    where
\begin{align}
J^{(out)}_*=
\left\{
\begin{aligned}
&\begin{pmatrix}
\mathbb{I}_{2\times2} &  \mathbf{0}_{2\times2}\\
\delta_{2}^{-1}(k_n)C_n\delta_{1}^{-1}(k_n)\frac{T^{-2}(k_n)}{k-k_n}\e^{2\ii t\theta(k_n)} & \mathbb{I}_{2\times2}
\end{pmatrix},\  k\in\Upsilon_n,\\
&\begin{pmatrix}
\mathbb{I}_{2\times2} &  \mathbf{0}_{2\times2}\\
-\sigma_2[\delta_{2}^{-1}(k_n)C_n\delta_{1}^{-1}(k_n)]^*\sigma_2
\frac{T^{-2}(-k^*_n)}{k+k_n^*}\e^{2\ii t\theta(-k^*_n)} & \mathbb{I}_{2\times2}
\end{pmatrix},\  k\in-\Upsilon^*_n,\\
&\begin{pmatrix}
\mathbb{I}_{2\times2} & [\delta_{2}^{-1}(k_n)C_n\delta_{1}^{-1}(k_n)]^\dag \frac{T^{2}(k_n^*)}{k-k_n^*}\e^{-2\ii t\theta(k_n^*)} \\
\mathbf{0}_{2\times2} & \mathbb{I}_{2\times2}
    \end{pmatrix},\  k\in\Upsilon^*_n,\\
&\begin{pmatrix}
    \mathbb{I}_{2\times2} & \sigma_2[\delta_{2}^{-1}(k_n)C_n\delta_{1}^{-1}(k_n)]^{\texttt{T}}\sigma_2
    \frac{T^{2}(-k_n)}{k+k_n}\e^{-2\ii t\theta(-k_n)} \\
    \mathbf{0}_{2\times2} & \mathbb{I}_{2\times2}
    \end{pmatrix},\  k\in-\Upsilon_n.
\end{aligned}
\right.
\end{align}
\item Normalization: $\mu_*^{(out)}(\zeta,t;k)\to\mathbb{I}_{4\times4}$, as $k\to\infty$.
\end{itemize}
\end{rhp}
\end{proposition}
\begin{proof}
The solvability of the RH problem \ref{rh3.4} follows from the Schwarz invariant condition of the jump matrices \cite{ZX1989}. Moreover, it is easy to see that on $\Upsilon\setminus(\pm\Upsilon_n\cup\pm\Upsilon_n^*)$, $E_*$ satisfies the following jump condition:
\be
E_{*+}=E_{*-}\left(\mu_*^{(out)}J^{(out)}[\mu_*^{(out)}]^{-1}\right),
\ee
Using \eqref{3.64}, the conclusion follows from solving a small norm RH problem (see the solution to RH problem \ref{rhp3.8} for detail).
\end{proof}
We now study the solution to RH problem \ref{rh3.4}. Using the Plemelj formula and Cauchy Residue theorem, if Re$k_n>0$, the solution of this RH problem is given by
\begin{align}
\mu_{*L}^{(cb)}(k)=&\begin{pmatrix}\mathbb{I}_{2\times2}\\ \mathbf{0}_{2\times2}\end{pmatrix}
+\mu_{*R}^{(cb)}(k_n)\delta_{2}^{-1}(k_n)C_n\delta_{1}^{-1}(k_n)
\frac{T^{-2}(k_n)}{k-k_n}\e^{2\ii t\theta(k_n)}\\
&-\mu_{*R}^{(cb)}(-k_n^*)
\sigma_2[\delta_{2}^{-1}(k_n)C_n\delta_{1}^{-1}(k_n)]^*\sigma_2
\frac{T^{-2}(-k^*_n)}{k+k^*_n}\e^{2\ii t\theta(-k_n^*)},\nn\\
\mu_{*R}^{(cb)}(k)=&\begin{pmatrix}\mathbf{0}_{2\times2}\\ \mathbb{I}_{2\times2} \end{pmatrix}-\mu_{*L}^{(cb)}(k_n^*)[\delta_{2}^{-1}(k_n)C_n\delta_{1}^{-1}(k_n)]^\dag
\frac{T^{2}(k^*_n)}{k-k^*_n}\e^{-2\ii t\theta(k^*_n)}\label{3.71}\\
&+\mu_{*L}^{(cb)}(-k_n)\sigma_2[\delta_{2}^{-1}(k_n)
C_n\delta_{1}^{-1}(k_n)]^{\texttt{T}}\sigma_2
\frac{T^{2}(-k_n)}{k+k_n}\e^{-2\ii t\theta(-k_n)}.\nn
\end{align}
Then, we have a closed system:
\begin{align}
\mu_{*UL}^{(cb)}(k_n^*)=&\mathbb{I}_{2\times2}+\mu_{*UR}^{(cb)}(k_n)
\delta_{2}^{-1}(k_n)C_n\delta_{1}^{-1}(k_n)\frac{T^{-2}(k_n)}{k_n^*-k_n}\e^{2\ii t\theta(k_n)}\label{3.72}\\
&-\mu_{*UR}^{(cb)}(-k_n^*)\sigma_2[\delta_{2}^{-1}(k_n)C_n\delta_{1}^{-1}(k_n)]^*\sigma_2
\frac{T^{-2}(-k^*_n)}{k_n^*+k^*_n}\e^{2\ii t\theta(-k^*_n)},\nn\\
\mu_{*UL}^{(cb)}(-k_n)=&\mathbb{I}_{2\times2}+\mu_{*UR}^{(cb)}(k_n)
\delta_{2}^{-1}(k_n)C_n\delta_{1}^{-1}(k_n)\frac{T^{-2}(k_n)}{-k_n-k_n}\e^{2\ii t\theta(k_n)}\label{3.73}\\
&-\mu_{*UR}^{(cb)}(-k_n^*)\sigma_2[\delta_{2}^{-1}(k_n)C_n\delta_{1}^{-1}(k_n)]^*\sigma_2
\frac{T^{-2}(-k^*_n)}{-k_n+k^*_n}\e^{2\ii t\theta(-k^*_n)},\nn\\
\mu_{*UR}^{(cb)}(k_n)=&-\mu_{*UL}^{(cb)}(k_n^*)
[\delta_{2}^{-1}(k_n)C_n\delta_{1}^{-1}(k_n)]^\dag\frac{T^{2}(k^*_n)}{k_n-k^*_n}\e^{-2\ii t\theta(k_n^*)}\label{3.74}\\
&+\mu_{*UL}^{(cb)}(-k_n)\sigma_2[\delta_{2}^{-1}(k_n)C_n\delta_{1}^{-1}(k_n)]^{\texttt{T}}\sigma_2
\frac{T^{2}(-k_n)}{k_n+k_n}\e^{-2\ii t\theta(-k_n)},\nn\\
\mu_{*UR}^{(cb)}(-k_n^*)=&-\mu_{*UL}^{(out)}(k_n^*)
[\delta_{2}^{-1}(k_n)C_n\delta_{1}^{-1}(k_n)]^\dag\frac{T^{2}(k^*_n)}{-k^*_n-k^*_n}\e^{-2\ii t\theta(k_n^*)}\label{3.75}\\
&+\mu_{*UL}^{(cb)}(-k_n)\sigma_2[\delta_{2}^{-1}(k_n)C_n\delta_{1}^{-1}(k_n)]^{\texttt{T}}\sigma_2
\frac{T^{2}(-k_n)}{-k_n^*+k_n}\e^{-2\ii t\theta(-k_n)}.\nn
\end{align}
Substituting Equations \eqref{3.74}, \eqref{3.75} into \eqref{3.72} and \eqref{3.73}, we can obtain $\mu_{*UL}^{(cb)}(k_n^*)$ and $\mu_{*UL}^{(cb)}(-k_n)$. Then, by reconstruction formulae \eqref{2.50} and \eqref{2.51}, as well as \eqref{3.71}, we can find the composite breather solution
\begin{align}
\begin{pmatrix} q_{1n}^{(cb)}(\zeta,t) & q_{2n}^{(cb)}(\zeta,t)\end{pmatrix}=&\ii\begin{pmatrix}
\left([\mu_*^{(cb)}(0)]^{-1}
\mu_{*1}^{(cb)}\right)_{13}
& \left([\mu_*^{(cb)}(0)]^{-1}
\mu_{*1}^{(cb)}\right)_{14}\end{pmatrix},\label{eq3.77cb}\\
x-\zeta(x,t)=& \ii \left[\left([\mu_*^{(cb)}(0)]^{-1}
\mu_{*1}^{(cb)}\right)_{11}-1\right],
\end{align}
where we expand
\be
\mu_{*}^{(cb)}(k)=\mu_*^{(cb)}(0)+\mu_{*1}^{(cb)}k+O(k^2),\ \text{as} \ k\to0.
\ee
\begin{remark}
We do not give explicit expressions for the composite breather solution because they could not be simplified enough to be instructive. In \cite{LSL-IST}, for a specific choice of the discrete eigenvalue and the norming constant matrix, the magnitudes of  $q_1^{(cb)}(\zeta,t)$ and $q_2^{(cb)}(\zeta,t)$ are shown graphically. We also note that the explicit expression for a composite breather associated with the discrete eigenvalue $k_1$ is given by Equation (173) in \cite{CGP} through Darboux transformations. Moreover, we know that the speed of this composite breather is $-1/(4|k_1|^2)$.
\end{remark}
If Re$k_n=0$, that is, $k_n=-k_n^*=\ii\nu_n$, the solution of RH problem \ref{rh3.4} in this case can be expressed by
\begin{align}
\mu_{*L}^{(sol)}(k)=&\begin{pmatrix}\mathbb{I}_{2\times2}\\ \mathbf{0}_{2\times2}\end{pmatrix}
+\mu_{*R}^{(sol)}(k_n)\delta_{2}^{-1}(k_n)C_n\delta_{1}^{-1}(k_n)
\frac{T^{-2}(k_n)}{k-k_n}\e^{2\ii t\theta(k_n)},\\
\mu_{*R}^{(sol)}(k)=&\begin{pmatrix}\mathbf{0}_{2\times2}\\ \mathbb{I}_{2\times2} \end{pmatrix}-\mu_{*L}^{(sol)}(k_n^*)[\delta_{2}^{-1}(k_n)C_n\delta_{1}^{-1}(k_n)]^\dag
\frac{T^{2}(k^*_n)}{k-k^*_n}\e^{-2\ii t\theta(k^*_n)}.
\end{align}
Thus, we have
\begin{align}
\mu_{*UL}^{(sol)}(k_n^*)=&\mathbb{I}_{2\times2}+\mu_{*UR}^{(sol)}(k_n)
\delta_{2}^{-1}(k_n)C_n\delta_{1}^{-1}(k_n)\frac{T^{-2}(k_n)}{k_n^*-k_n}\e^{2\ii t\theta(k_n)},\label{3.76}\\
\mu_{*UR}^{(sol)}(k_n)=&-\mu_{*UL}^{(sol)}(k_n^*)
[\delta_{2}^{-1}(k_n)C_n\delta_{1}^{-1}(k_n)]^\dag\frac{T^{2}(k^*_n)}{k_n-k^*_n}\e^{-2\ii t\theta(k_n^*)}.\label{3.77}
\end{align}
By substituting \eqref{3.77} into \eqref{3.76} to find $\mu_{*UL}^{(sol)}(k_n^*)$,  as $k\to0$, expanding $\mu_{*}^{(sol)}(k)=\mu_*^{(sol)}(0)+\mu_{*1}^{(sol)}k+O(k^2)$,
and using \eqref{2.50} and \eqref{2.51}, we obtain the so-called self-symmetric soliton solution:
\begin{align}\label{3.78}
&\begin{pmatrix} q_{1n}^{(sol)}(\zeta,t) & q_{2n}^{(sol)}(\zeta,t)\end{pmatrix}=\ii\begin{pmatrix}
\left([\mu_*^{(sol)}(0)]^{-1}
\mu_{*1}^{(sol)}\right)_{13}
& \left([\mu_*^{(sol)}(0)]^{-1}
\mu_{*1}^{(sol)}\right)_{14}\end{pmatrix}\\
&\qquad\qquad\qquad\qquad\quad \quad\,\,\,\
=-\frac{1}{k_n}\text{sech}\left[\zeta_1(\zeta,t)
-x_0\right]\frac{\mathbf{g}^*}{|\mathbf{g}|},\nn\\
&x=\zeta+\frac{2}{\nu_n^2}\frac{1}{1+\e^{2(\zeta_1-x_0)}},\quad \zeta_1(\zeta,t)=2\nu_n\left(\zeta+\frac{t}{4\nu_n^2}\right),\quad
x_0=\log\frac{|\mathbf{g}|}{2\nu_n},\nn
\end{align}
with
\be
\mathbf{g}=\begin{pmatrix}\alpha_n\\ \beta_n\end{pmatrix},\quad\begin{pmatrix}\alpha_n&\beta_n^*\\ \beta_n&-\alpha_n^*\end{pmatrix}\doteq\delta_{2}^{-1}(k_n)C_n\delta_{1}^{-1}(k_n)T^{-2}(k_n).
\ee
\subsubsection{Local RH models near phase points}
From the Proposition \ref{prop3.1}, $J^{(2)}-\mathbb{I}_{4\times4}$ in the neighborhood $U_{\pm k_0}$ of $\pm k_0$ does not have a uniformly decay as $t\to\infty$. Moreover, $1-\chi(k)=1$ for $k\in U_{\pm k_0}$. Thus, we can establish the model RH problems for $\mu^{(\pm k_0)}$ in $U_{\pm k_0}$, which are the local models near the critical points $\pm k_0$. Denote $\Gamma^{\epsilon1}\doteq\cup_{l=1}^4\Gamma_l\cap U_{ k_0}=\cup_{l=1}^4\Gamma_l^\epsilon,$ $\Gamma^{\epsilon2}\doteq\cup_{l=5}^8\Gamma_l\cap U_{-k_0}=\cup_{l=5}^8\Gamma_l^\epsilon,$ see Figure \ref{fig4}.
\begin{figure}[htbp]
 \centering
 \includegraphics[width=3in]{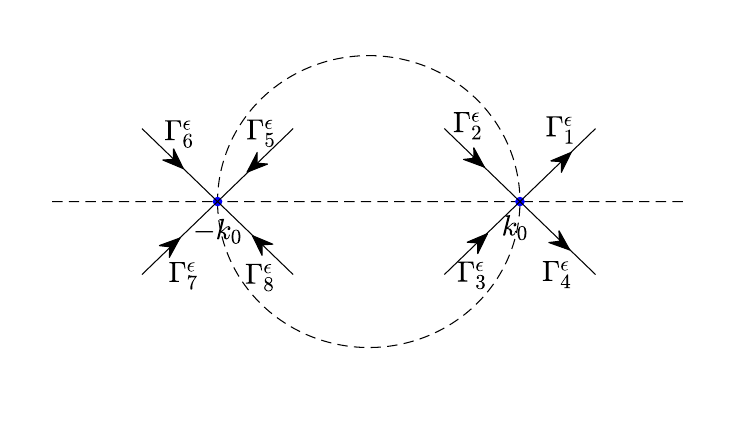}
  \caption{The contours $\Gamma^{\epsilon1}\cup\Gamma^{\epsilon2}$ in the complex $k$-plane.}\label{fig4}
\end{figure}
\begin{rhp}\label{rhp3.5}
Find a $4\times4$ matrix-valued function $\mu^{(k_0)}(\zeta,t;k)$ on  $\bfC\setminus\Gamma^{\epsilon1}$ with the following properties:
\begin{itemize}
\item Analyticity: $\mu^{(k_0)}(\zeta,t;k)$ is analytical in $\bfC\setminus\Gamma^{\epsilon1}$.

\item Jump condition: For $k\in\Gamma^{\epsilon1}$, the continuous boundary values $\mu^{(ccSP)}_\pm$ satisfy the jump relation
\be
\mu^{(k_0)}_+(\zeta,t;k)=\mu^{(k_0)}_-(\zeta,t;k)J^{(k_0)}(\zeta,t;k),
\ee
where the jump matrix $J^{(k_0)}(\zeta,t;k)$ is expressed by
\begin{align*}
J^{(k_0)}=
\left\{
\begin{aligned}
&\begin{pmatrix}
\mathbb{I}_{2\times2} & \delta_{1}(k)\left(\mathbb{I}_{2\times2}+\rho^\dag(k_0) \rho(k_0)\right)^{-1}\rho^\dag(k_0)\delta_{2}(k)T^2(k_0)\e^{-2\ii t\theta(k)}\\
\mathbf{0}_{2\times2} & \mathbb{I}_{2\times2}
\end{pmatrix},\quad k\in\Gamma_1^\epsilon,\\
&\begin{pmatrix}
\mathbb{I}_{2\times2} & \mathbf{0}_{2\times2}\\
\delta_{2}^{-1}(k)\rho(k_0)\delta_{1}^{-1}(k)T^{-2}(k_0)\e^{2\ii t\theta(k)} & \mathbb{I}_{2\times2}
\end{pmatrix},\quad\,    k\in\Gamma_2^\epsilon,\\
&\begin{pmatrix}
\mathbb{I}_{2\times2} & \delta_{1}(k)\rho^\dag(k_0)\delta_{2}(k)T^2(k_0)\e^{-2\ii t\theta(k)}\\
\mathbf{0}_{2\times2} & \mathbb{I}_{2\times2}
\end{pmatrix},\qquad k\in\Gamma_3^\epsilon,\\
&\begin{pmatrix}
\mathbb{I}_{2\times2} & \mathbf{0}_{2\times2}\\
\delta_{2}^{-1}(k)\rho(k_0)\left(\mathbb{I}_{2\times2}+\rho^\dag(k_0) \rho(k_0)\right)^{-1}\delta_{1}^{-1}(k)T^{-2}(k_0)\e^{2\ii t\theta(k)} & \mathbb{I}_{2\times2}
\end{pmatrix}, \,\,\ k\in\Gamma_4^\epsilon.
\end{aligned}
\right.
\end{align*}
\item Normalization: $\mu^{(k_0)}(\zeta,t;k)\to\mathbb{I}_{4\times4}$, as $k\to\infty$.
\end{itemize}
\end{rhp}
\begin{rhp}\label{rhp3.6}
Find a $4\times4$ matrix-valued function $\mu^{(-k_0)}(\zeta,t;k)$ on  $\bfC\setminus\Gamma^{\epsilon2}$ with the following properties:
\begin{itemize}
\item Analyticity: $\mu^{(-k_0)}(\zeta,t;k)$ is analytical in $\bfC\setminus\Gamma^{\epsilon2}$.

\item Jump condition: For $k\in\Gamma^{\epsilon2}$, the continuous boundary values $\mu^{(-k_0)}_\pm$ satisfy the jump relation
\be
\mu^{(-k_0)}_+(\zeta,t;k)=\mu^{(-k_0)}_-(\zeta,t;k)J^{(-k_0)}(\zeta,t;k),
\ee
where the jump matrix $J^{(-k_0)}(\zeta,t;k)$ is expressed by
\begin{align*}
J^{(-k_0)}=
\left\{
\begin{aligned}
&\begin{pmatrix}
\mathbb{I}_{2\times2} &\mathbf{0}_{2\times2}\\
-\delta_{2}^{-1}(k)\rho(-k_0)\delta_{1}^{-1}(k)T^{-2}(-k_0)\e^{2\ii t\theta(k)} & \mathbb{I}_{2\times2}
\end{pmatrix},\quad k\in\Gamma_5^\epsilon,\\
&\begin{pmatrix}
\mathbb{I}_{2\times2} & \delta_{1}(k)\left(\mathbb{I}_{2\times2}+\rho^\dag(-k_0) \rho(-k_0)\right)^{-1}\rho^\dag(-k_0)\delta_{2}(k)T^2(-k_0)\e^{-2\ii t\theta(k)}\\
\mathbf{0}_{2\times2} & \mathbb{I}_{2\times2}
\end{pmatrix},\,\,\ k\in\Gamma_6^\epsilon,\\
&\begin{pmatrix}
\mathbb{I}_{2\times2} &\mathbf{0}_{2\times2}\\
\delta_{2}^{-1}(k)\rho(-k_0)\left(\mathbb{I}_{2\times2}+\rho^\dag(-k_0) \rho(-k_0)\right)^{-1}\delta_{1}^{-1}(k)T^{-2}(-k_0)\e^{2\ii t\theta(k)} & \mathbb{I}_{2\times2}
\end{pmatrix},\, k\in\Gamma_7^\epsilon,\\
&\begin{pmatrix}
\mathbb{I}_{2\times2} & -\delta_{1}(k)\rho^\dag(-k_0)\delta_{2}(k)T^2(-k_0)\e^{-2\ii t\theta(k)}\\
\mathbf{0}_{2\times2} & \mathbb{I}_{2\times2}
\end{pmatrix},\quad\,\,\,\,  k\in\Gamma_8^\epsilon.
\end{aligned}
\right.
\end{align*}
\item Normalization: $\mu^{(-k_0)}(\zeta,t;k)\to\mathbb{I}_{4\times4}$, as $k\to\infty$.
\end{itemize}
\end{rhp}
We now study the solution of RH problems \ref{rhp3.5} and \ref{rhp3.6}. We will take the RH problem \ref{rhp3.5} as an example and others can be handled in the same way.

Based on the Beals--Coifman theory \cite{BC}, we decompose
$J^{(k_0)}=(\mathbb{I}_{4\times4}-w_-^{(k_0)})^{-1}(\mathbb{I}_{4\times4}+w_+^{(k_0)}),$ and let $w^{(k_0)}=w_+^{(k_0)}+w_-^{(k_0)}$. Our first step is to extend the contour $\Gamma^{\epsilon1}$ to the contour
\begin{align}
\hat{\Gamma}^{\epsilon1}\doteq\{k\in\bfC|k=k_0+k_0\alpha\e^{\pm\frac{\ii\pi}{4}},-\infty<\alpha<\infty\},\nn
\end{align}
and define $\hat{w}^{(k_0)}$ on $\hat{\Gamma}^{\epsilon1}$ through
\berr
\hat{w}^{(k_0)}=\left\{
\begin{aligned}
&w^{(k_0)},\quad k\in\Gamma^{\epsilon1}\subset\hat{\Gamma}^{\epsilon1},\\
&\mathbf{0}_{4\times4},\quad\ k\in\hat{\Gamma}^{\epsilon1}\setminus\Gamma^{\epsilon1}.
\end{aligned}
\right.
\eerr

It is first noted that the functions $\delta_1(k)$ and $\delta_2(k)$ respectively satisfy the $2\times2$ matrix RH problems \eqref{3.3} and \eqref{3.4}, and hence they can not be solved in explicit form. However, by taking the determinants, they are transformed into a same scalar RH problem as follows. Set $\delta(k)=\det[\delta_j(k)]$, we then get
\be\label{3.81}
\left\{
\begin{aligned}
\delta_{+}(k)&=\left(1+\text{tr}[\rho(k)\rho^\dag(k)]
+\det[\rho(k)\rho^\dag(k)]\right)\delta_{-}(k),\quad  |k|>k_0,\\
&=\delta_{-}(k),\qquad\qquad\qquad\qquad\qquad\qquad\qquad\quad\qquad |k|<k_0,\\
\delta(k)&\to1,\qquad\qquad\qquad\qquad\qquad\qquad\qquad\quad\qquad\qquad k\to\infty.
\end{aligned}
\right.
\ee
By the Plemelj formula, we find
\begin{align}\label{3.82}
\delta(k)=&\exp\left\{\frac{1}{2\pi\ii}\left(\int_{-\infty}^{-k_0}+\int_{k_0}^{\infty}\right)
\frac{\ln\left(1+\text{tr}[\rho(s)\rho^\dag(s)]+\det[\rho(s)\rho^\dag(s)]\right)}{s-k}\dd s\right\}\\
=&\left(\frac{k-k_0}{k+k_0}\right)^{-\ii\nu(k_0)}\e^{\chi_{k_0}(k)},
\end{align}
where
\begin{align}
\nu(k_0)&=-\frac{1}{2\pi}\ln\left(1+\text{tr}[\rho(k_0)\rho^\dag(k_0)]
+\det[\rho(k_0)\rho^\dag(k_0)]\right),\label{3.84}\\
\chi_{k_0}(k)&=\frac{1}{2\pi\ii}\left(\int_{-\infty}^{-k_0}+\int_{k_0}^{\infty}\right)
\ln\left(\frac{1+\text{tr}[\rho(s)\rho^\dag(s)]+\det[\rho(s)\rho^\dag(s)]}{1+\text{tr}[\rho(k_0)\rho^\dag(k_0)]
+\det[\rho(k_0)\rho^\dag(k_0)]}\right)\dd s,\label{3.85}
\end{align}
and hence $\delta(k)$, $\delta^{-1}(k)$ are uniformly bounded. Therefore, we can write
\begin{align}\label{3.86}
&\delta_{1}(k)\rho^\dag(k_0)\delta_{2}(k)\e^{-2\ii t\theta(k)}=\delta_{1}(k)\rho^\dag(k_0)\left(\delta_{2}(k)-\delta(k) \mathbb{I}_{2\times2}\right)\e^{-2\ii t\theta(k)}\\
&+(\delta_1(k)-\delta(k) \mathbb{I}_{2\times2})\rho^\dag(k_0)\delta(k)\e^{-2\ii t\theta(k)}+\rho^\dag(k_0)\delta^2(k)\e^{-2\ii t\theta(k)}.\nn
\end{align}
Expanding $\ii t\theta(k)$, we obtain
\be
\ii t\theta(k)=\frac{t}{2\ii k_0}+\frac{t}{4\ii k_0^3}(k-k_0)^2+\frac{\ii t}{4s^4}(k-k_0)^3,\quad s\ \text{lies\ between}\ k_0\ \text{and}\ k.
\ee
We thus define the following scaling transformation
\be
\mathcal{N}: f(k)\to(\mathcal{N}f)(k)=f\left(\frac{z}{\sqrt{k_0^{-3}t}}+k_0\right),
\ee
which acts on $\delta(k)\e^{-\ii t\theta(k)}$ and gives
\be
(\mathcal{N}\delta\e^{-\ii t\theta})(k)=\delta_{k_0}^0\delta_{k_0}^1(z),
\ee
where
\begin{align}
\delta_{k_0}^0=&\e^{\chi_{k_0}(k_0)-\frac{t}{2\ii k_0}}(4k_0^{-1}t)^{\frac{\ii\nu(k_0)}{2}},\label{deltak0}\\
\delta_{k_0}^1(z)=&z^{-\ii\nu(k_0)}\exp\left\{\frac{\ii z^2}{4}\left(1-\frac{ z}{s^4k_0^{-9/2}t^{1/2}}\right)\right\}\\
&\times\left(\frac{2k_0}{z/\sqrt{k_0^{-3}t}+2k_0}\right)^{-\ii\nu(k_0)}
\e^{\chi_{k_0}([z/\sqrt{k_0^{-3}t}]+k_0)-\chi_{k_0}(k_0)}.\nn
\end{align}
Set $\tilde{w}^{(k_0)}=\mathcal{N}\hat{w}^{(k_0)}$. \begin{figure}[htbp]
  \centering
  \includegraphics[width=3in]{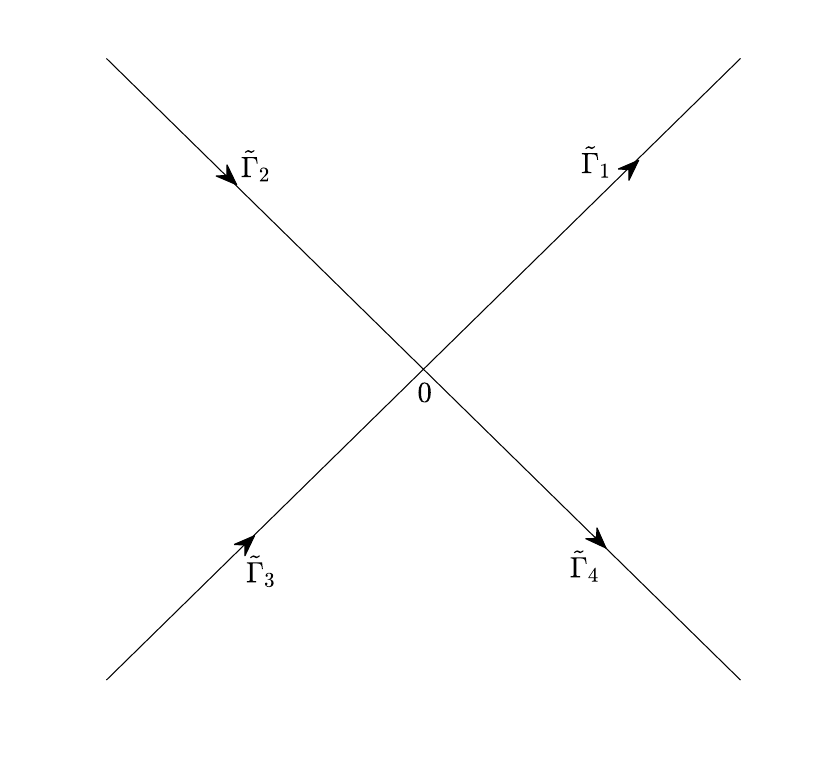}
  \caption{The oriented contour $\tilde{\Gamma}^{(k_0)}$.}\label{fig5}
\end{figure}
Let $\tilde{\Gamma}^{(k_0)}$ denote the contour $\{z\in\bfC|z=\alpha\e^{\pm\frac{\ii\pi}{4}},-\infty<\alpha<\infty\}$ centered at original point and oriented to the origin as shown in Figure \ref{fig5}.
We next let $J^{(X)}(\rho(k_0);z)=(\mathbb{I}_{4\times4}-w_-^{(X)})^{-1}(\mathbb{I}_{4\times4}+w_+^{(X)})$, where
\begin{align}
w^{(X)}=&w^{(X)}_+=\left\{
\begin{aligned}
&\begin{pmatrix}
\mathbf{0}_{2\times2} & \left(\mathbb{I}_{2\times2}+\rho^\dag(k_0) \rho(k_0)\right)^{-1}\rho^\dag(k_0)(\delta_{k_0}^0)^2T^2(k_0)z^{-2\ii\nu(k_0)}\e^{\frac{\ii z^2}{2}}\\
\mathbf{0}_{2\times2}& \mathbf{0}_{2\times2}
\end{pmatrix},\ z\in \tilde{\Gamma}_1,\\
&\begin{pmatrix}
\mathbf{0}_{2\times2} & \rho^\dag(k_0)(\delta_{k_0}^0)^2T^2(k_0)z^{-2\ii\nu(k_0)}\e^{\frac{\ii z^2}{2}}\\
 \mathbf{0}_{2\times2}& \mathbf{0}_{2\times2}
\end{pmatrix},\ z\in\tilde{\Gamma}_3,
\end{aligned}
\right.\label{WX+}\\
w^{(X)}=&w^{(X)}_-=\left\{
\begin{aligned}
&\begin{pmatrix}
\mathbf{0}_{2\times2} & \mathbf{0}_{2\times2}\\
\rho(k_0)(\delta_{k_0}^0)^{-2}T^{-2}(k_0)z^{2\ii\nu(k_0)}\e^{-\frac{\ii z^2}{2}} & \mathbf{0}_{2\times2}
\end{pmatrix},\, z\in\tilde{\Gamma}_2,\\
&\begin{pmatrix}
\mathbf{0}_{2\times2} & \mathbf{0}_{2\times2}\\
\rho(k_0)\left(\mathbb{I}_{2\times2}+\rho^\dag(k_0) \rho(k_0)\right)^{-1}(\delta_{k_0}^0)^{-2}T^{-2}(k_0)z^{2\ii\nu(k_0)}\e^{-\frac{\ii z^2}{2}} & \mathbf{0}_{2\times2}
\end{pmatrix},\ z\in\tilde{\Gamma}_4.
\end{aligned}
\right.\label{WX-}
\end{align}
Then, we have the following estimates on the rate of convergence.
\begin{proposition}\label{prop3.3}
For $z\in\{z\in\bfC|z=\sqrt{k^{-1}_0t}\alpha\e^{\frac{\pi\ii}{4}},
-\epsilon\leq\alpha\leq\epsilon\}$, as $t\to\infty$, we have the following estimations
\begin{align}
\left|[\delta_{k_0}^1(z)]^2-z^{-2\ii\nu(k_0)}\e^{\frac{\ii z^2}{2}}\right|&\leq c\frac{\ln t}{\sqrt{t}},\label{3.92}\\
\left|\left(\mathcal{N}\left(\delta_{2}-\delta \mathbb{I}_{2\times2}\right)\e^{-2\ii t\theta}\right)(k)\right|&\leq ct^{-1},\label{3.93}\\
\left|\left(\mathcal{N}(\delta_1-\delta\mathbb{I}_{2\times2})\e^{-2\ii t\theta}\right)(k)\right|&\leq ct^{-1}.\label{3.94}
\end{align}
\end{proposition}
\begin{proof}
See the Appendix \ref{secA}.
\end{proof}
Introduce the Cauchy operator $\mathcal{C}_\pm$ on $\Gamma$ as follows
\be
(\mathcal{C}_\pm f)(k)=\int_{\Gamma}\frac{f(s)}{s-k_\pm}\frac{\dd s}{2\pi\ii},\quad k\in\Gamma,~f\in L^2(\Gamma).
\ee
Define the operator $\mathcal{C}_{w}f\doteq\mathcal{C}_+(fw_-)+\mathcal{C}_-(fw_+)$.
\begin{lemma}\label{lem2.6}
As $t\rightarrow\infty$, $(1-\mathcal{C}_{w^{(X)}})^{-1}:L^2(\tilde{\Gamma}^{(k_0)})\rightarrow L^2(\tilde{\Gamma}^{(k_0)})$ exists and is uniformly bounded:
\berr
\left\|(1-\mathcal{C}_{w^{X}})^{-1}\right\|_{L^2(\tilde{\Gamma}^{(k_0)})}\leq c,
\eerr
and hence,
\berr
\left\|(1-\mathcal{C}_{\tilde{w}^{(k_0)}})^{-1}\right\|_{L^2(\tilde{\Gamma}^{(k_0)})}\leq c.
\eerr
\end{lemma}
\begin{proof}
See \cite{CL2021,PD} and references therein.
\end{proof}
A simple change of variables argument shows that
$\mathcal{C}_{\hat{w}^{(k_0)}}=\mathcal{N}^{-1}\mathcal{C}_{\tilde{w}^{(k_0)}}\mathcal{N}.$ Then, we can get
\begin{align}\label{3.99}
&\int_{\Gamma^{\epsilon1}}\left(((1-\mathcal{C}_{w^{(k_0)}})^{-1}\mathbb{I}_{4\times4})
w^{(k_0)}\right)(s)\dd s=\int_{\hat{\Gamma}^{\epsilon1}}\left(((1-\mathcal{C}_{\hat{w}^{(k_0)}})^{-1}
\mathbb{I}_{4\times4})\hat{w}^{(k_0)}\right)(s)\dd s\\
&=\int_{\hat{\Gamma}^{\epsilon1}}(\mathcal{N}^{-1}(1-\mathcal{C}_{\tilde{w}^{(k_0)}})^{-1}
\mathcal{N}\mathbb{I}_{4\times4})(s)
\hat{w}^{(k_0)}(s)\dd s\nn\\
&=\int_{\hat{\Gamma}^{\epsilon1}}((1-\mathcal{C}_{\tilde{w}^{(k_0)}})^{-1}
\mathbb{I}_{4\times4})\left((s-k_0)\sqrt{k_0^{-3}t}\right)
(\mathcal{N}\hat{w}^{(k_0)})\left((s-k_0)\sqrt{k^{-3}_0t}\right)\dd s\nn\\
&=\frac{1}{\sqrt{k_0^{-3}t}}
\int_{\tilde{\Gamma}^{(k_0)}}\left(((1-\mathcal{C}_{\tilde{w}^{(k_0)}})^{-1}
\mathbb{I}_{4\times4})\tilde{w}^{(k_0)}\right)(s)\dd s\nn\\
&=\frac{1}{\sqrt{k_0^{-3}t}}\int_{\tilde{\Gamma}^{(k_0)}}
\left(((1-\mathcal{C}_{w^{(X)}})^{-1}\mathbb{I}_{4\times4})w^{(X)}\right)(s)\dd s+O\left(\frac{\ln t}{t}\right).\nn
\end{align}
For $z\in\bfC\setminus\tilde{\Gamma}^{(k_0)}$, let
\be\label{3.100}
\mu^{(X)}(\rho(k_0);z)=\mathbb{I}_{4\times4}
+\frac{1}{2\pi\ii}\int_{\tilde{\Gamma}^{(k_0)}}
\frac{\left(((1-\mathcal{C}_{w^{(X)}})^{-1}\mathbb{I}_{4\times4})w^{(X)}\right)(s)}{s-z}
\dd s,
\ee
then $\mu^{(X)}(\rho(k_0);z)$ solves the following RH problem:
\begin{rhp}\label{rhp3.7}
Find a $4\times4$ matrix-valued function $\mu^{(X)}(\rho(k_0);z)$ with the following properties:
\begin{itemize}
\item Analyticity: $\mu^{(X)}(\rho(k_0);z)$ is analytic for $z\in\bfC\setminus\tilde{\Gamma}^{(k_0)}$ and continuous on $\tilde{\Gamma}^{(k_0)}$.

\item Jump condition: The continuous boundary values $\mu^{(X)}_\pm(\rho(k_0);z)$ satisfy the following jump relation
\bea
\mu^{(X)}_+(\rho(k_0);z)=\mu^{(X)}_-(\rho(k_0);z)J^{(X)}(\rho(k_0);z),\quad z\in\tilde{\Gamma}^{(k_0)}.
\eea

\item Normalization: $\mu^{(X)}(\rho(k_0);z)\rightarrow\mathbb{I}_{4\times4},$ as $z\rightarrow\infty.$
\end{itemize}
\end{rhp}
The solution of the RH problem for $\mu^{(X)}$ can be expressed based on the PC model, see Appendix \ref{secB}, that is,
\begin{align}
\mu^{(X)}(\rho(k_0);z)=\left(\delta_{k_0}^0T(k_0)\right)^{\Sigma_3}
\mu^{(PC)}\left(\rho(k_0);z\right)\left(\delta_{k_0}^0T(k_0)\right)^{-\Sigma_3}.
\end{align}
Then, in the large $z$ expansion
\be
\mu^{(X)}(\rho(k_0);z)=\mathbb{I}_{4\times4}+\frac{\mu_1^{(X)}}{z}+O(z^{-2}),\quad z\rightarrow\infty,
\ee
then it follows from \eqref{3.100} and \eqref{B.4} that
\begin{align}
\mu_1^{(X)}=&-\frac{1}{2\pi\ii}\int_{\tilde{\Gamma}^{(k_0)}}
\left(((1-\mathcal{C}_{w^{(X)}})^{-1}\mathbb{I}_{4\times4})w^{(X)}\right)(s)\dd s\\
=&\left(\delta_{k_0}^0T(k_0)\right)^{\Sigma_3}
\mu_1^{(PC)}(\rho(k_0))\left(\delta_{k_0}^0T(k_0)\right)^{-\Sigma_3}.\nn
\end{align}
Thus, with $z=\sqrt{k_0^{-3}t}(k-k_0)$, by \eqref{3.99}, the solution of $\mu^{(k_0)}(\zeta,t;k)$ to RH problem \ref{rhp3.5} admits the following expansion:
\be\label{3.105}
\mu^{(k_0)}(\zeta,t;k)=\mathbb{I}_{4\times4}
+\frac{\mu^{(k_0)}_1(\zeta,t)}{\sqrt{k_0^{-3}t}(k-k_0)}+O\left(\frac{\ln t}{t}\right),\quad t\rightarrow\infty,
\ee
with
\begin{align}
\mu^{(k_0)}_1(\zeta,t)=&\begin{pmatrix}\mathbf{0}_{2\times2}&
\ii(\delta_{k_0}^0)^2T^2(k_0)\beta^{(k_0)}\\ -\ii(\delta_{k_0}^0)^{-2}T^{-2}(k_0)\left(\beta^{(k_0)}\right)^\dag
&\mathbf{0}_{2\times2}\end{pmatrix},\label{muk0}\\
\beta^{(k_0)}=&\frac{
\sqrt{2\pi}\e^{\frac{3\ii\pi}{4}-\frac{\pi\nu(k_0)}{2}}}
{\Gamma(\ii\nu(k_0))\det[\rho(k_0)]}\begin{pmatrix}
\rho_{22}(k_0) & -\rho_{12}(k_0)\\ -\rho_{21}(k_0) & \rho_{11}(k_0)
\end{pmatrix}.\label{betak0}
\end{align}

For the RH problem \ref{rhp3.6}, proceeding the calculation in the same way, we conclude that the solution $\mu^{(-k_0)}(\zeta,t;k)$ satisfies the following asymptotic behavior
\be\label{3.108}
\mu^{(-k_0)}(\zeta,t;k)
=\mathbb{I}_{4\times4}+\frac{\mu_{1}^{(-k_0)}(\zeta,t)}{\sqrt{k_0^{-3}t}(k+k_0)}+O\left(\frac{\ln t}{t}\right),
\ee
where
\begin{align}
\mu^{(-k_0)}_1(\zeta,t)=&\begin{pmatrix}\mathbf{0}_{2\times2}&
-\ii(\delta_{-k_0}^0)^2T^2(-k_0)\beta^{(-k_0)}\\ \ii(\delta_{-k_0}^0)^{-2}T^{-2}(-k_0)\left(\beta^{(-k_0)}\right)^\dag
&\mathbf{0}_{2\times2}\end{pmatrix},\label{mu-k0}\\
\beta^{(-k_0)}&=\frac{
\sqrt{2\pi}\e^{\frac{\ii\pi}{4}-\frac{\pi\nu(k_0)}{2}}}
{\Gamma(-\ii\nu(k_0))\det[\rho(-k_0)]}\begin{pmatrix}
\rho_{22}(-k_0) & -\rho_{12}(-k_0)\\ -\rho_{21}(-k_0) & \rho_{11}(-k_0)
\end{pmatrix},\label{beta-k0}\\
\delta_{-k_0}^0&=\e^{\chi_{k_0}(-k_0)+\frac{t}{2\ii k_0}}(4k_0^{-1}t)^{-\frac{\ii\nu(k_0)}{2}}.\label{3.117}
\end{align}
\subsubsection{Small norm RH problem for error function}\label{sec3.4.3}
Now, we consider the error function $E(k)$. Assume the boundaries of $U_{\pm k_0}$ are oriented counterclockwise. Denote $$\Gamma^{(E)}=\partial U_{\pm k_0}\cup(\Gamma\setminus U_{\pm k_0}).$$ From the definition \eqref{3.60}, we know that $E(k)$ satisfies the following $4\times4$ matrix RH problem.

\begin{rhp}\label{rhp3.8}
Find a $4\times4$ matrix-valued function $E(k)$ with the following properties:
\begin{itemize}
\item Analyticity: $E(k)$ is analytic in $\bfC\setminus\Gamma^{(E)}$.

\item Jump condition: The continuous boundary values $E_\pm(k)$ on $\Gamma^{(E)}$ satisfy the following jump relation
\be
E_+(k)=E_-(k)J^{(E)}(k),
\ee
where the jump matrix is given by
\be
J^{(E)}(k)=\left\{\begin{aligned}
&\mu^{(out)}(\zeta,t;k)J^{(2)}(\zeta,t;k)[\mu^{(out)}(\zeta,t;k)]^{-1},\qquad k\in\Gamma\setminus U_{\pm k_0},\\
&\mu^{(out)}(\zeta,t;k)\mu^{(\pm k_0)}(\zeta,t;k)[\mu^{(out)}(\zeta,t;k)]^{-1},\quad k\in\partial U_{\pm k_0}.
\end{aligned}
\right.
\ee
\item Normalization: $E(k)\rightarrow\mathbb{I}_{4\times4},$ as $k\rightarrow\infty.$
\end{itemize}
\end{rhp}
By Proposition \ref{pro2.5} and \ref{prop3.1}, we have the estimate
\be\label{3.112}
\left|J^{(E)}(k)-\mathbb{I}_{4\times4}\right|\leq c\e^{-\frac{\sqrt{2}t|k\mp k_0|}{4}(k_0^{-2}-|k|^{-2})}\leq c\e^{-\frac{\sqrt{2}t\epsilon}{16k_0}},\ \text{for}\ k\in\Gamma\setminus U_{\pm k_0}.
\ee
For $k\in\partial U_{\pm k_0}$, by \eqref{3.105} and \eqref{3.108}, one can get
\be\label{3.113}
\left|J^{(E)}(k)-\mathbb{I}_{4\times4}\right|=\left|\mu^{(out)}(\zeta,t;k)(\mu^{(\pm k_0)}(\zeta,t;k)-\mathbb{I}_{4\times4})[\mu^{(out)}(\zeta,t;k)]^{-1}\right|\leq ct^{-1/2}.
\ee
The existence and uniqueness of solution to RH problem \ref{rhp3.8} follows from the theory of small-norm RH problems. In fact, let $w^{(E)}=J^{(E)}-\mathbb{I}_{4\times4}$ and $\mathcal{C}_{w^{(E)}}f\doteq\mathcal{C}_-(fw^{(E)})$, where we have chosen, for simplicity, $w^{(E)}_+=w^{(E)}$ and $w^{(E)}_-=\mathbf{0}_{4\times4}$. Then the solution of the RH problem \ref{rhp3.8} can be given by
\be\label{3.116}
E(k)=\mathbb{I}_{4\times4}+
\frac{1}{2\pi\ii}\int_{\Gamma^{(E)}}\frac{(\mu^{(E)}w^{(E)})(s)}{s-k}\dd s,
\ee
where the $4\times4$ matrix-valued function $\mu^{(E)}(x,t;k)$ defined by $\mu^{(E)}=\mathbb{I}_{4\times4}+\mathcal{C}_{w^{(E)}}\mu^{(E)}.$
By \eqref{3.112} and \eqref{3.113}, we find
\be
\|\mathcal{C}_{w^{(E)}}\|_{\mathcal{B}(L^2(\Gamma^{(E)}))}\leq c\|w^{(E)}\|_{L^\infty(\Gamma^{(E)})}\leq ct^{-1/2},
\ee
where $\mathcal{B}(L^2(\Gamma^{(E)}))$ denotes the bounded linear operators $L^2(\Gamma^{(E)})\rightarrow L^2(\Gamma^{(E)})$. Hence, the resolvent operator $(1-\mathcal{C}_{w^{(E)}})^{-1}$ is existent and thus of both $\mu^{(E)}$ and $E$.
Moreover, using the Neumann series, the function $\mu^{(E)}(k)$ satisfies
\be
\|\mu^{(E)}(k)-\mathbb{I}_{4\times4}\|_{L^2(\Gamma^{(E)})}=O(t^{-1/2}),\quad t\rightarrow\infty.
\ee
Now, it can be explained that the definition of $\mu^{(RHP)}$ in \eqref{3.60} is reasonable.

In order to reconstruct the solution of system \eqref{ccSPE}, we need the asymptotic behavior of $E(k)$ as $k\to0$. It follows from \eqref{3.116} that, as $k\to0$
\begin{align}
E(k)=E(0)+E_1k+O(k^2),
\end{align}
where
\begin{align}
E(0)=&\mathbb{I}_{4\times4}+
\frac{1}{2\pi\ii}\int_{\Gamma^{(E)}}\frac{(\mu^{(E)}w^{(E)})(s)}{s}\dd s,\\
E_1=&
\frac{1}{2\pi\ii}\int_{\Gamma^{(E)}}\frac{(\mu^{(E)}w^{(E)})(s)}{s^2}\dd s.
\end{align}
Then, the large time asymptotic behavior of $E(0)$ and $E_1$ can be derived as
\begin{align}\label{3.126}
E(0)-\mathbb{I}_{4\times4}=&
\frac{1}{2\pi\ii}\int_{\Gamma^{(E)}}\frac{(\mu^{(E)}w^{(E)})(s)}{s}\dd s\\
=&\frac{1}{2\pi\ii}\int_{\Gamma^{(E)}}\frac{(\mu^{(E)}(s)-\mathbb{I}_{4\times4})w^{(E)}(s)}{s}\dd s+\frac{1}{2\pi\ii}\int_{\Gamma^{(E)}}\frac{w^{(E)}(s)}{s}\dd s\nn\\
=&\frac{1}{2\pi\ii}\int_{\partial U_{\pm k_0}}\frac{w^{(E)}(s)}{s}\dd s+O(\|\mu^{(E)}-\mathbb{I}_{4\times4}\|_{L^2(\Gamma^{(E)})}\|w^{(E)}\|_{L^\infty(\Gamma^{(E)})})\nn\\
&+O(\|s^{-1}\|_{L^2(\Gamma\setminus U_{\pm k_0})}\|w^{(E)}\|_{L^2(\Gamma\setminus U_{\pm k_0})})\nn\\
=&\frac{1}{2\pi\ii}\int_{\partial U_{\pm k_0}}\frac{\mu^{(out)}(s)(\mu^{(\pm k_0)}(s)-\mathbb{I}_{4\times4})[\mu^{(out)}(s)]^{-1}}{s}\dd s+O(t^{-1})\nn\\
=&t^{-1/2}\mathcal{E}_0+O(t^{-1}\ln t),\nn
\end{align}
and
\begin{align}\label{eq3.127}
E_1=
\frac{1}{2\pi\ii}\int_{\Gamma^{(E)}}\frac{(\mu^{(E)}w^{(E)})(s)}{s^2}\dd s=t^{-1/2}\mathcal{E}_1+O(t^{-1}\ln t).
\end{align}
where
\begin{align}
\mathcal{E}_0=&\sqrt{k_0}\left(\mu^{(out)}(\zeta,t;k_0)\mu_1^{( k_0)}(\zeta,t)[\mu^{(out)}(\zeta,t;k_0)]^{-1}\right.\\
&+\left.\mu^{(out)}(\zeta,t;-k_0)\mu_1^{( -k_0)}(\zeta,t)[\mu^{(out)}(\zeta,t;-k_0)]^{-1}\right),\nn\\
\mathcal{E}_1=&\frac{1}{\sqrt{k_0}}\left(\mu^{(out)}(\zeta,t;k_0)\mu_1^{( k_0)}(\zeta,t)[\mu^{(out)}(\zeta,t;k_0)]^{-1}\right.\label{eq3.128e1}\\
&+\left.\mu^{(out)}(\zeta,t;-k_0)\mu_1^{( -k_0)}(\zeta,t)[\mu^{(out)}(\zeta,t;-k_0)]^{-1}\right).\nn
\end{align}
\subsection{Analysis on the $\bar{\partial}$-problem}\label{sec3.5}
Next, we investigate the existence and long-time asymptotics of \(\mu_{3}(\zeta,t;k)\). The associated pure \(\bar{\partial}\)-problem \ref{dbar1} is equivalent to the corresponding integral equation
\begin{equation}\label{3.119}
	\mu^{(3)}(k) = \mathbb{I}_{4\times4} - \frac{1}{\pi}\iint\limits_{\mathbb{C}}\frac{\mu^{(3)}(s)w^{(3)}(s)}{s-k}\text{d}m(s),
\end{equation}
where \(m(s)\) represents the Lebesgue measure on \(\mathbb{C}\). We define $ C_k$ as the left Cauchy-Green integral operator,
\begin{equation}\label{3.120}
	C_k[f](k) = -\frac{1}{\pi}\iint\limits_{\mathbb{C}}\frac{f(s)w^{(3)}(s)}{s-k}\text{d}m(s).
\end{equation}
Thus, the Equation \eqref{3.119} can be rewritten as
\begin{equation}\label{3.121}
	(1-C_k)\mu^{(3)}(k) = \mathbb{I}_{4\times 4}.
\end{equation}
According to formula \eqref{3.121}, the solution $\mu^{(3)}(k)$ exists if and only if the inverse operator \((1-C_k)^{-1}\) exists. Therefore, our goal is to prove that the operator \((1-C_k)\) is invertible. To achieve this, we first present the following proposition.
\begin{proposition}
As $t\to\infty$, the norm of the integral operator $C_k$ decays to zero, specifically,
	\begin{equation}
	\| C_k\|_{L^{\infty}\to L^{\infty}} = O(t^{-1/6}),
	\end{equation}
which implies \((1-C_k)^{-1}\)exists.
\end{proposition}
\begin{proof}
	 Assume that $f\in L^{\infty}(\Omega_1)$, with $s = u + \ii v$ and $ k= x+\ii y$. Then, according to \eqref{3.120}, we have
	\begin{align}
		|C_k[f](k)| &\leq \|f(k)\|_{L^{\infty}}\frac{1}{\pi}\iint\limits_{\mathbb{C}}\frac{|w^{(3)}(s)|}{|s-k|}\text{d}m(s) \nn \\
		&\leq c\iint\limits_{\mathbb{C}}\frac{|\bar{\partial}R^{(2)}(s)|}{|s-k|}\text{d}m(s).
	\end{align}
Thus, it remains to estimate the above integral. For $\bar{\partial}R^{(2)}(s)$ is a piece-wise function, we focus specifically on the case where the matrix function is supported in region $\Omega_1$, as the other cases can be proved similarly. From \eqref{3.34} and \eqref{3.51}, it follows that
\begin{equation}
	\iint\limits_{\Omega_1}\frac{|\bar{\partial}R^{(2)}(s)|}{|s-k|}\text{d}m(s)\leq I_1 + I_2 + I_3,
\end{equation}
where
\begin{align}
	 I_1 =& \int_{0}^{+\infty}\int_{k_0 + v}^{+\infty}\frac{|\bar{\partial}\chi(s)|\e^{\frac{vt}{2(u^2+v^2)}}}{\sqrt{(u-x)^2+(v-y)^2}}
\text{d}u\e^{2tv\hat{\zeta}}\text{d}v, \\
	  I_2 =& \int_{0}^{+\infty}\int_{k_0 + v}^{+\infty}\frac{|\rho'(u)|\e^{\frac{vt}{2(u^2+v^2)}}}{\sqrt{(u-x)^2+(v-y)^2}}
\text{d}u\e^{2tv\hat{\zeta}}\text{d}v, \\
	\label{3.127}  I_3 =& \int_{0}^{+\infty}\int_{k_0 + v}^{+\infty}\frac{\left((u-k_0)^2+v^2\right)^{-1/4}\e^{\frac{vt}{2(u^2+v^2)}}}
{\sqrt{(u-x)^2+(v-y)^2}}\text{d}u\e^{2tv\hat{\zeta}}\text{d}v.
\end{align}
In the subsequent calculations, we will make use of the inequality
\begin{equation}\label{3.128}
	\left\|\frac{1}{s-k}\right\|^2_{L^2(k_0,+\infty)} = \int_{k_0}^{+\infty}\frac{1}{|v-y|}\left[\left(\frac{u-x}{v-y}\right)^2+1\right]^{-1} \text{d}\left(\frac{u-x}{|v-y|}\right)\leq \frac{\pi}{|v-y|}.
\end{equation}
Given that \(\frac{vt}{2(u^2+v^2)}\) is monotonically decreasing with respect to \(u\), so \(I_1\) admits the following estimates
\begin{align}
	I_1 & \leq \int_0^{+\infty}\left\||s-k|^{-1}\right\|_{L^2\left(k_0,+\infty\right)}\left\|\bar{\partial} \chi(s)\right\|_{L^2\left(k_0,+\infty\right)} \e^{\frac{v t}{2\left(k_0^2+v^2\right)}} \e^{2 t v \hat{\zeta}} \text{d} v \label{3.129}\\
	& \leq c\int_0^{+\infty}|v-y|^{-1 / 2} \exp \left(-\frac{vt}{2}\left(\frac{1}{k_0^2}-\frac{1}{k_0^2+v^2}\right)\right) \text{d} v \nn\\
	& =c\int_0^y(y-v)^{-1 / 2} \exp \left(-\frac{vt}{2}\left(\frac{1}{k_0^2}-\frac{1}{k_0^2+v^2}\right)\right) \text{d} v \nn\\
	& +c\int_{y}^{+\infty}(v-y)^{-1 / 2}\exp \left(-\frac{vt}{2}\left(\frac{1}{k_0^2}-\frac{1}{k_0^2+v^2}\right)\right) \text{d} v.\nn
\end{align}
Since $e^{-z}\leq cz^{-1/6}$ for all $z>0 $, the first integral can then be estimated by
\begin{align}
	 &\int_0^y(y-v)^{-1 / 2} \exp \left(-\frac{vt}{2}\left(\frac{1}{k_0^2}-\frac{1}{k_0^2+v^2}\right)\right) \text{d} v  \label{3.130}\\
	 &\leq c\int_{0}^{y}(y-v)^{-1 / 2}v^{-1/2}t^{-1/6}\text{d}v \leq ct^{-1/6}.\nn
\end{align}
In addition, considering the remaining integral and letting $w = v-y$, we have
\begin{align}
	&\int_{y}^{+\infty}(v-y)^{-1 / 2}\exp \left(-\frac{vt}{2}\left(\frac{1}{k_0^2}-\frac{1}{k_0^2+v^2}\right)\right) \text{d} v  \label{3.131}\\
	&= \int_{y}^{+\infty}(v-y)^{-1 / 2}\exp\left(-\frac{vt}{2}\left(\frac{v^2}{k_0^2(v^2+k_0^2)}\right)\right) \text{d} v \nn\\
	&\leq \int_{y}^{+\infty}(v-y)^{-1 / 2}\exp\left(-\frac{vty^2}{2k_0^2(y^2+k_0^2)}\right) \text{d} v \nn\\
	&\leq \int_{0}^{+\infty}w^{-1/2}\exp\left(-\frac{wty^2}{2k_0^2(y^2+k_0^2)}\right)\text{d}w \exp\left(-\frac{ty^3}{2k_0^2(y^2+k_0^2)}\right)\leq ct^{-1/2}.\nn
\end{align}
Inserting \eqref{3.130} and \eqref{3.131} into \eqref{3.129} yields
\begin{equation}\label{3.132}
	I_1\leq ct^{-1/6}.
\end{equation}
$I_2$ satisfies the same estimate as \eqref{3.132}, and for $I_3$, we first obtain
\begin{align}
	&\left\| \left((u-k_0)^2+v^2\right)^{-1/4}\right\|_{L^{p}(k_0,+\infty)} \label{3.133} \\
	&=\left\{\int_{k_0}^{+\infty}\left[(u-k_0)^2+v^2\right]^{-p/4}\text{d} u\right\}^{1/p}\nn\\
	&= \left\{\int_{k_0}^{+\infty}\left[\left(\frac{u-k_0}{v}\right)^2+1 \right]^{-p/4}\text{d}\left(\frac{u-k_0}{v}\right)\right\}^{1/p}v^{1/p-1/2} \nn\\
	&\leq cv^{1/p-1/2},\nn
\end{align}
and
\begin{align}
	\left\| \frac{1}{|s-k|}\right\|_{L^q(k_0,+\infty)} &= \left\{\int_{k_0}^{+\infty}\left[\left(\frac{u-x}{v-y}\right)^2+1\right]^{-q/2} \text{d}\left(\frac{u-x}{|v-y|}\right)\right\}^{1/q}|v-y|^{1/q-1}\label{3.134} \\
	&\leq c|v-y|^{1/q-1},\nn
\end{align}
where $ p>2$ and $\frac{1}{p}+\frac{1}{q} = 1$. By examining \eqref{3.127}, we arrive that
\begin{align}
	I_3 &\leq\int_{0}^{+\infty} \left\||s-k|^{-1}\right\|_{L^q\left(k_0,+\infty\right)}\left\| \left((u-k_0)^2+v^2\right)^{-1/4}\right\|_{L^{p}(k_0,+\infty)} \e^{\frac{v t}{2\left(k_0^2+v^2\right)}} \e^{2 t v \hat{\zeta}} \text{d} v  \\
	&\leq \int_{0}^{+\infty}v^{1/p-1/2}|v-y|^{1/q-1}\exp \left(-\frac{vt}{2}\left(\frac{1}{k_0^2}-\frac{1}{k_0^2+v^2}\right)\right) \text{d}v \nn \\
	&=\int_{0}^{y}v^{1/p-1/2}(y-v)^{1/q-1}\exp \left(-\frac{vt}{2}\left(\frac{1}{k_0^2}-\frac{1}{k_0^2+v^2}\right)\right) \text{d}v \nn \\
	&+\int_{y}^{+\infty}v^{1/p-1/2}(v-y)^{1/q-1}\exp \left(-\frac{vt}{2}\left(\frac{1}{k_0^2}-\frac{1}{k_0^2+v^2}\right)\right) \text{d}v.\nn
\end{align}
The first integral can then be estimated by
	\begin{align}
		&\int_{0}^{y}v^{1/p-1/2}(y-v)^{1/q-1}\exp \left(-\frac{vt}{2}\left(\frac{1}{k_0^2}-\frac{1}{k_0^2+v^2}\right)\right) \text{d}v \\
		&\leq c\int_{0}^{y}v^{1/p-1}(y-v)^{1/q-1} t^{-1/6} \text{d}v \leq ct^{-1/6}.\nn
	\end{align}
Moreover, considering the remaining integral with $w = v-y$, we have
\begin{align}
	&\int_{y}^{+\infty}v^{1/p-1/2}(v-y)^{1/q-1}\exp \left(-\frac{vt}{2}\left(\frac{1}{k_0^2}-\frac{1}{k_0^2+v^2}\right)\right) \text{d}v  \\
	&\leq \int_{0}^{+\infty}(w+y)^{1/p-1/2}w^{1/q-1}\exp\left(-\frac{wty^2}{2k_0^2(y^2+k_0^2)}\right)\text{d}w \exp\left(-\frac{ty^3}{2k_0^2(y^2+k_0^2)}\right) \nn \\
	&\leq c\int_{0}^{+\infty}w^{-1/2}\exp\left(-\frac{wty^2}{2k_0^2(y^2+k_0^2)}\right)\text{d}w \leq ct^{-1/2}.\nn
\end{align}
Finally, we obtain
\begin{equation}
	I_3\leq ct^{-1/6}.
\end{equation}
Collecting the above results, the proof of the proposition is thus completed.

\end{proof}
Next, to achieve the final goal of reconstructing the potential $(q_1(x,t)\quad q_2(x,t))$ as $t\to\infty$, we need to analyze the asymptotic expansion of $\mu^{(3)}(\zeta,t;k)$ as $k\to 0$. We expand the $\mu^{(3)}(\zeta,t;k)$ as
\begin{equation}
	\mu^{(3)}(\zeta,t;k) = \mathbb{I}_{4\times 4} + \mu^{(3)}_0(\zeta,t)+\mu^{(3)}_1(\zeta,t)k + O(k^2),\quad k\to 0,
\end{equation}
where
\begin{align}
\label{3.140}	\mu^{(3)}_0(\zeta,t) &= -\frac{1}{\pi}\iint\limits_{\mathbb{C}}\frac{\mu^{(3)}(s)w^{(3)}(s)}{s}\text{d}m(s), \\
	\label{3.141}\mu^{(3)}_1(\zeta,t)&=-\frac{1}{\pi}\iint\limits_{\mathbb{C}}\frac{\mu^{(3)}(s)w^{(3)}(s)}{s^2}\text{d}m(s).
\end{align}
Then, $\mu^{(3)}_0(\zeta,t)$ and $\mu^{(3)}_1(\zeta,t)$ satisfy the following properties.
\begin{proposition}\label{pro3.5}
	As $t\to\infty$, the following estimate holds for $\mu^{(3)}(\zeta,t;0)$, that is,
	\begin{equation}
		\|\mu^{(3)}(\zeta,t;0)-\mathbb{I}_{4\times 4}\| = \|\mu^{(3)}_0(\zeta,t)\| \leq ct^{-1}.
	\end{equation}
\end{proposition}
\begin{proof}
	We first consider the case $k\in\Omega_1$, and the proofs for the other cases follow analogously. Applying \eqref{3.34} and \eqref{3.140}, and taking into account the boundedness of $\mu^{(3)}(k)$ and $\mu^{(RHP)}(k)$, we drive
\begin{align}
	\|\mu^{(3)}_0(\zeta,t)\| &\leq \frac{1}{\pi}\iint\limits_{\Omega_1}\frac{|\mu^{(3)}(s)\mu^{(RHP)}(s)\bar{\partial}
		\mathcal{R}^{(2)}(s)[\mu^{(RHP)}(s)]^{-1}|}{|s|}\text{d}m(s) \\
		&\leq c\int_{0}^{+\infty}\int_{k_0+v}^{+\infty}\frac{|\bar{\partial}R_1(s)|
\e^{\frac{vt}{2(u^2+v^2)}}}{(u^2+v^2)^{1/2}}\e^{2tv\hat{\zeta}}\text{d}u\text{d}v \nn\\
		&\leq c(I_4+I_5+I_6),\nn
\end{align}
where $s= u+\ii v$, and
\begin{align}
	I_4 =& \int_{0}^{+\infty}\int_{k_0 + v}^{+\infty}\frac{|\bar{\partial}\chi(s)|\e^{\frac{vt}{2(u^2+v^2)}}}{(u^2+v^2)^{1/2}}
\text{d}u\e^{2tv\hat{\zeta}}\text{d}v, \\
	I_5 =& \int_{0}^{+\infty}\int_{k_0 + v}^{+\infty}\frac{|\rho'(u)|\e^{\frac{vt}{2(u^2+v^2)}}}{(u^2+v^2)^{1/2}}
\text{d}u\e^{2tv\hat{\zeta}}\text{d}v, \\
  I_6 =& \int_{0}^{+\infty}\int_{k_0 + v}^{+\infty}\frac{\left((u-k_0)^2+v^2\right)^{-1/4}
  \e^{\frac{vt}{2(u^2+v^2)}}}{(u^2+v^2)^{1/2}}\text{d}u\e^{2tv\hat{\zeta}}\text{d}v.
\end{align}
To facilitate the estimation, we split $I_4$ into two parts
\begin{equation}
	I_4=\int_{0}^{k_0}\int_{k_0 + v}^{+\infty}\frac{|\bar{\partial}\chi(s)|\e^{\frac{vt}{2(u^2+v^2)}}}{(u^2+v^2)^{1/2}}
\text{d}u\e^{2tv\hat{\zeta}}\text{d}v +\int_{k_0}^{+\infty}\int_{k_0 + v}^{+\infty}\frac{|\bar{\partial}\chi(s)|\e^{\frac{vt}{2(u^2+v^2)}}}
{(u^2+v^2)^{1/2}}\text{d}u\e^{2tv\hat{\zeta}}\text{d}v.
	\end{equation}
	For the first integral, we have
\begin{align}
		&\int_{0}^{k_0}\int_{k_0 + v}^{+\infty}\frac{|\bar{\partial}\chi(s)|\e^{\frac{vt}{2(u^2+v^2)}}}{(u^2+v^2)^{1/2}}
\text{d}u\e^{2tv\hat{\zeta}}\text{d}v  \\
		& \leq c\int_0^{k_0}\left(\left(v+k_0\right)^2+v^2\right)^{-1 / 2}\exp\left(-\frac{vt}{2}\left(\frac{1}{k_0^2}-\frac{1}{(v+k_0)^2+v^2}\right)\right)
\text{d}v \nn\\
		&\leq c\int_0^{k_0}\left(\left(v+k_0\right)^2+v^2\right)^{-1 / 2}\exp\left(\frac{vt}{2((v+k_0)^2+v^2)}\right)\text{d}v \nn\\
		&\leq c\int_{0}^{k_0}\exp\left(\frac{vt}{2k_0^2}\right)\text{d}v
		\leq ct^{-1}.\nn
\end{align}
	For the last integral, by applying \eqref{3.128} with $y=0$, we obtain
\begin{align}
	&\int_{k_0}^{+\infty}\int_{k_0 + v}^{+\infty}\frac{|\bar{\partial}\chi(s)|\e^{\frac{vt}{2(u^2+v^2)}}}{(u^2+v^2)^{1/2}}
\text{d}u\e^{2tv\hat{\zeta}}\text{d}v \\
	&\leq \int_{k_0}^{+\infty}\left\||s|^{-1}\right\|_{L^2}\left\|\bar{\partial} \chi(s)\right\|_{L^2}\exp\left(-\frac{vt}{2}\left(\frac{1}{k_0^2}-\frac{1}{(v+k_0)^2+v^2}\right)\right) \text{d}v \nn\\
	&\leq c\int_{k_0}^{+\infty} v^{-1/2}\exp\left(-\frac{2tv}{5k_0^2}\right)\text{d}v \nn\\
	&\leq c\int_{k_0}^{+\infty}\exp\left(-\frac{2tv}{5k_0^2}\right)\text{d}v \leq ct^{-1}.\nn
\end{align}
	A similar estimate for $I_5$ can be obtained in the same manner.
	As with $I_4$, we split $I_6$ into two parts to facilitate the estimation
	\begin{align}
		I_6 &= \int_{0}^{k_0}\int_{k_0 + v}^{+\infty}\frac{\left((u-k_0)^2+v^2\right)^{-1/4}\e^{\frac{vt}{2(u^2+v^2)}}}
{(u^2+v^2)^{1/2}}\text{d}u\e^{2tv\hat{\zeta}}\text{d}v\\
		&+\int_{k_0}^{+\infty}\int_{k_0 + v}^{+\infty}\frac{\left((u-k_0)^2+v^2\right)^{-1/4}
\e^{\frac{vt}{2(u^2+v^2)}}}{(u^2+v^2)^{1/2}}\text{d}u\e^{2tv\hat{\zeta}}\text{d}v. \nn
	\end{align}
	Given that
	\begin{align}
		&\left\|(u^2+v^2)^{-1/2}\exp\left(\frac{vt}{2(u^2+v^2)}\right)\right\|_{L^4(k_0+v,+\infty)} \\
		&= \left\{\int_{k_0+v}^{+\infty}(u^2+v^2)^{-2}\exp\left(\frac{2vt}{u^2+v^2}\right)
\text{d}u\right\}^{1/4} \nn \\
		&= \left\{\int_{k_0+v}^{+\infty}(-4t)^{-1}v^{-1}u^{-1}
\left[\exp\left(\frac{2vt}{u^2+v^2}\right)\right]'\text{d}u\right\}^{1/4} \nn\\
		&\leq ct^{-1/4}v^{-1/4}\left(\exp\left(\frac{vt}{2((k_0+v)^2+v^2)}\right)+1\right).\nn
	\end{align}
	An estimate for the first integral can be given by combining \eqref{3.133}
	\begin{align}
	&\int_{0}^{k_0}\int_{k_0 + v}^{+\infty}\frac{\left((u-k_0)^2+v^2\right)^{-1/4}\e^{\frac{vt}{2(u^2+v^2)}}}
{(u^2+v^2)^{1/2}}\text{d}u\e^{2tv\hat{\zeta}}\text{d}v	 \\
	&\leq \int_{0}^{k_0}\left\|\left((u-k_0)^2+v^2\right)^{-1/4}\right\|_{L^{4/3}}
\left\|(u^2+v^2)^{-1/2}\exp\left(\frac{vt}{2(u^2+v^2)}\right)\right\|_{L^4}
\e^{2tv\hat{\zeta}}\text{d}v\nn\\
	&\leq ct^{-1/4}\int_{0}^{k_0}v^{3/4-1/2}v^{-1/4}
\left(\exp\left(\frac{vt}{2((k_0+v)^2+v^2)}\right)+1\right)\e^{-\frac{vt}{2k_0^2}}\text{d}v \nn\\
	&\leq ct^{-1/4}\int_{0}^{k_0}\exp\left(-\frac{vt}{2k_0^2}\right)\text{d}v \leq ct^{-5/4}.\nn
	\end{align}
	Let $p>2$ satisfy $\frac{1}{p}+\frac{1}{q}=1$. Then, by applying \eqref{3.133} and \eqref{3.134}, we estimate the second integral as follows:
	\begin{align}
		&\int_{k_0}^{+\infty}\int_{k_0 + v}^{+\infty}\frac{\left((u-k_0)^2+v^2\right)^{-1/4}\e^{\frac{vt}{2(u^2+v^2)}}}
{(u^2+v^2)^{1/2}}\text{d}u\e^{2tv\hat{\zeta}}\text{d}v \\
		&\leq \int_{k_0}^{+\infty} \left\||s|^{-1}\right\|_{L^q}\left\| \left((u-k_0)^2+v^2\right)^{-1/4}\right\|_{L^{p}} \exp\left(-\frac{vt}{2}\left(\frac{1}{k_0^2}-\frac{1}{(v+k_0)^2+v^2}\right)\right) \text{d} v \nn\\
		&\leq c\int_{k_0}^{+\infty} v^{-1/2}\exp\left(-\frac{2tv}{5k_0^2}\right)\text{d}v \nn \\
		&\leq c\int_{k_0}^{+\infty}\exp\left(-\frac{2tv}{5k_0^2}\right)\text{d}v  \leq ct^{-1}.\nn
	\end{align}
	Therefore, we obtain we conclude that
	\begin{equation}
		I_6\leq ct^{-1}.
	\end{equation}
 In summary, based on the above analysis, we complete the proof of the proposition.
\end{proof}
 \begin{proposition}\label{prop3.6}
 	As $t\to\infty$, the following estimate for $\mu^{(3)}_1(\zeta,t)$ holds:
 \begin{equation}
 	\|\mu^{(3)}_1(\zeta,t)\|\leq ct^{-1}.
 \end{equation}
 \end{proposition}
 \begin{proof}
 	We focus on estimating the integral over $\Omega_1$, as the estimates for the other regions are similar. Similar to the previous proposition, we set $s= u+iv$. By applying \eqref{3.34}, \eqref{3.51} and \eqref{3.141}, we derive the following result
 	\begin{align}
 		\|\mu_1^{(3)}(\zeta,t)\| &\leq \frac{1}{\pi}\iint\limits_{\Omega_1}\frac{|\mu^{(3)}(s)\mu^{(RHP)}(s)\bar{\partial}
 			\mathcal{R}^{(2)}(s)[\mu^{(RHP)}(s)]^{-1}|}{|s|^2}\text{d}m(s) \nn \\
 			&\leq c\int_{0}^{+\infty}\int_{k_0+v}^{+\infty}\frac{|\bar{\partial}R_1(s)|\e^{\frac{vt}{2(u^2+v^2)}}}{u^2+v^2}\e^{2tv\zeta}\text{d}u\text{d}v \nn\\
 			&\leq c(I_7+I_8+I_9),
 	\end{align}
 	where
 	\begin{align}
 			I_7 =& \int_{0}^{+\infty}\int_{k_0 + v}^{+\infty}\frac{|\bar{\partial}\chi(s)|\e^{\frac{vt}{2(u^2+v^2)}}}{u^2+v^2}
 \text{d}u\e^{2tv\hat{\zeta}}\text{d}v, \\
 		I_8 =& \int_{0}^{+\infty}\int_{k_0 + v}^{+\infty}\frac{|\rho'(u)|\e^{\frac{vt}{2(u^2+v^2)}}}{u^2+v^2}
 \text{d}u\e^{2tv\hat{\zeta}}\text{d}v, \\
 		I_9 =& \int_{0}^{+\infty}\int_{k_0 + v}^{+\infty}\frac{\left((u-k_0)^2+v^2\right)^{-1/4}
 \e^{\frac{vt}{2(u^2+v^2)}}}{u^2+v^2}\text{d}u\e^{2tv\hat{\zeta}}\text{d}v.
 	\end{align}
 	Observe that for all $s\in\Omega_1$, it holds that
 	\begin{equation}
 		(u^2+v^2)^{-1/2}\leq \frac{1}{k_0}.
 	\end{equation}
 	Consequently, we obtain
 	\begin{equation}
 		I_j\leq \frac{1}{k_0}I_{j-3}, \quad \text{for}\ j= 7,8,9.
 	\end{equation}
 	Hence, the conclusion readily follows from Proposition \ref{pro3.5}.

 \end{proof}

\subsection{Long-time asymptotics for the ccSPE}\label{sec3.6}
We now put together our previous results and formulate the long-time asymptotic formula of $\begin{pmatrix}q_1(x,t)&q_2(x,t)\end{pmatrix}$ in region $\hat{\zeta}<-\varepsilon$. Undoing all transformations carried out previously, we have
\be
\breve{\mu}(\zeta,t;k)=\mu^{(3)}(\zeta,t;k)E(k)\mu^{(out)}(\zeta,t;k)
[\mathcal{R}^{(2)}(k)]^{-1}[T(k)]^{\Sigma_3}[\tilde{\Delta}(k)]^{-1},\quad k\in\bfC\setminus U_{\pm k_0}.
\ee
Moreover, by \eqref{3.6} and \eqref{3.8}, as $k\to0$, we have
\begin{align}
[\tilde{\Delta}(k)]^{-1}=\tilde{\Delta}_0+\tilde{\Delta}_1k+O(k^2),\
T(k)=T(0)(1+T_1k)+O(k^2),\end{align}
where
\begin{align}
\tilde{\Delta}_0=&\begin{pmatrix} \delta_{10} & \mathbf{0}_{2\times2}\\
\mathbf{0}_{2\times2} & \delta_{20}^{-1}\end{pmatrix},\
\tilde{\Delta}_1=\begin{pmatrix} \delta_{11} & \mathbf{0}_{2\times2}\\
\mathbf{0}_{2\times2} & -\delta_{20}^{-1}\delta_{21}\delta_{20}^{-1}\end{pmatrix},\\
T_1=&-4\ii\sum_{\substack{\text{Re}k_n\neq0,\text{Im}k_n>0\\n\in\Delta_{k_0}^-}}
\frac{\text{Im}k_n}{|k_n|^2}
-2\ii\sum_{\substack{\text{Re}k_n=0,\text{Im}k_n>0\\n\in\Delta_{k_0}^-}}
\frac{\text{Im}k_n}{|k_n|^2},
\end{align}
with $\delta_{10}$, $\delta_{11}$, $\delta_{20}$ and $\delta_{21}$ are $2\times2$ constant matrices independent of $k$.
We take $k\to0$ along the imaginary axis such that $\mathcal{R}^{(2)}(k)=\mathbb{I}_{4\times4}$. Expanding $\mu_{*}^{(out)}(\zeta,t;k)=\mu_*^{(out)}(\zeta,t;0)+\mu_{*1}^{(out)}(\zeta,t)k+O(k^2)$, it follows from \eqref{3.126}, \eqref{eq3.127}, Propositions \ref{pro3.5} and \ref{prop3.6} that
\begin{align}
&\lim_{k\to0}\frac{\left[\breve{\mu}^{-1}(\zeta,t;0)\breve{\mu}(\zeta,t;k)\right]_{UR}}{k}\\
=&
\left[\tilde{\Delta}_0^{-1}[T(0)]^{-\Sigma_3}[\mu^{(out)}_*(\zeta,t;0)]^{-1}
\mu^{(out)}_{*1}(\zeta,t)[T(0)]^{\Sigma_3}\tilde{\Delta}_0\right.\nn\\
&+t^{-1/2}\left.\tilde{\Delta}_0^{-1}[T(0)]^{-\Sigma_3}[\mu^{(out)}_*(\zeta,t;0)]^{-1}
\mathcal{E}_1\mu^{(out)}_*(\zeta,t;0)[T(0)]^{\Sigma_3}\tilde{\Delta}_0\right]_{UR}+O(t^{-1}\ln t).\nn
\end{align}
By the first symmetry in \eqref{3.5}, we find $\delta_{j0}^{-1}=\delta_{j0}^\dag$ for $j=1,2$. And hence, we can express
\be\label{3.180}
\delta_{j0}=\begin{pmatrix} a_j & b_j\\ -b_j^*\e^{\ii\phi_j} & a_j^*\e^{\ii\phi_j}\end{pmatrix},
\ee
where $a_j$, $b_j$ are complex constant, $\det[\delta_{j0}]=\e^{\ii\phi_j}$, $\phi_j=\arg\det[\delta_{j0}]$. Thus, by \eqref{eq3.77cb} and \eqref{3.78}, the $(1,3)$ and $(1,4)$-entries of $\tilde{\Delta}_0^{-1}[T(0)]^{-\Sigma_3}\ii[\mu^{(out)}_*(\zeta,t;0)]^{-1}
\mu^{(out)}_{*1}(\zeta,t)[T(0)]^{\Sigma_3}\tilde{\Delta}_0$ can be respectively written as
\begin{align}
\label{3.181}\tilde{q}_{1n}^{(\ell)}(\zeta,t)\doteq\left(\cdot\right)_{13}
=&T^{-2}(0)\left(\left(a_1^*q_{1n}^{(\ell)}(\zeta,t)
+b_1\e^{-\ii\phi_1}[q_{2n}^{(\ell)}(\zeta,t)]^*\right)a_2^*\right.\\
&+\left.\left(a_1^*q_{2n}^{(\ell)}(\zeta,t)
-b_1\e^{-\ii\phi_1}[q_{1n}^{(\ell)}(\zeta,t)]^*\right)b_2^*\right),\nn\\
\label{3.182}\tilde{q}_{2n}^{(\ell)}(\zeta,t)\doteq\left(\cdot\right)_{14}
=&T^{-2}(0)\left(\left(a_1^*q_{1n}^{(\ell)}(\zeta,t)
+b_1\e^{-\ii\phi_1}[q_{2n}^{(\ell)}(\zeta,t)]^*\right)(-b_2\e^{-\ii\phi_2})\right.\\
&+\left.\left(a_1^*q_{2n}^{(\ell)}(\zeta,t)
-b_1\e^{-\ii\phi_1}[q_{1n}^{(\ell)}(\zeta,t)]^*\right)a_2\e^{-\ii\phi_2}\right),\quad \ell=cb, sol.\nn
\end{align}
Finally, together with reconstruction formulae \eqref{2.50}-\eqref{2.51}, we arrive at the asymptotic result \eqref{1.4} described in Theorem \ref{th1.1}.

If $\zeta/t=v$ with $v<0$ but $v\neq v_n$ for all $n=1,\cdots,N$, then the solution $\mu^{(RHP)}$ of the RH problem \ref{rh3.2} should take the following form
\be
\mu^{(RHP)}(\zeta,t;k)=\left\{\begin{aligned}
&E(k),\qquad\qquad\qquad\,\,\, k\in\bfC\setminus U_{\pm k_0},\\
&E(k)\mu^{(k_0)}(\zeta,t;k),\quad\,\,\ k\in U_{k_0},\\
&E(k)\mu^{(-k_0)}(\zeta,t;k),\quad k\in U_{-k_0}.
\end{aligned}
\right.
\ee
Thus, we have
\begin{align}
Q=\lim_{k\to0}\frac{\left[\breve{\mu}^{-1}(\zeta,t;0)\breve{\mu}(\zeta,t;k)\right]_{UR}}{k}
=\frac{T^{-2}(0)}{\sqrt{tk_0}}\left[\tilde{\Delta}_0^{-1}\left(\mu_1^{( k_0)}(\zeta,t)+\mu_1^{( -k_0)}(\zeta,t)\right)\tilde{\Delta}_0\right]_{UR}.
\end{align}
By \eqref{muk0} and \eqref{mu-k0}, we find the conclusion \eqref{1.8} presented in Theorem \ref{th1.1}.
\section{Asymptotics in range $\hat{\zeta}>\varepsilon$}\label{sec4}
We now turn to the study of the asymptotic behavior when $\hat{\zeta}>\varepsilon$. Our starting point is RH problem \ref{rhp2.1}. In this case, there is no stationary point, and we only need the following decomposition for the jump matrix $\breve{J}$ on the $\bfR$:
\be\label{4.1}
\breve{J}=\begin{pmatrix} \mathbb{I}_{2\times2} & \rho^\dag(k)\e^{-2\ii t\theta}\\[4pt]
\mathbf{0}_{2\times2} &\mathbb{I}_{2\times2}\end{pmatrix}\begin{pmatrix} \mathbb{I}_{2\times2} & \mathbf{0}_{2\times2}\\[4pt]
\rho(k)\e^{2\ii t\theta} &\mathbb{I}_{2\times2}\end{pmatrix}.
\ee
Since all the pole conditions have desired decay properties, and hence, following the same argument of Proposition \ref{prop3.2}, we have
\be\label{4.2}
\breve{\mu}(\zeta,t;k)=\left(\mathbb{I}_{4\times4}+O(\e^{-ct})\right)
\mu^{(1)}(\zeta,t;k),
\ee
where $\mu^{(1)}(\zeta,t;k)$ satisfies the following RH problem.
\begin{rhp}\label{rhp4.1}
Find a $4\times4$ matrix-valued function $\mu^{(1)}(\zeta,t;k)$ with the following properties:
\begin{itemize}
\item Analyticity: $\mu^{(1)}(\zeta,t;k)$ is analytic in $\bfC\setminus\bfR$.

\item Jump condition: The continuous boundary values of $\mu^{(1)}$ on $\bfR$ satisfy the jump relation
\be
\mu^{(1)}_+(\zeta,t;k)=\mu^{(1)}_-(\zeta,t;k)J^{(1)}(\zeta,t;k),
\ee
where the jump matrix is given by
\be
J^{(1)}(\zeta,t;k)=\begin{pmatrix} \mathbb{I}_{2\times2} & \rho^\dag(k)\e^{-2\ii t\theta}\\[4pt]
\mathbf{0}_{2\times2} &\mathbb{I}_{2\times2}\end{pmatrix}\begin{pmatrix} \mathbb{I}_{2\times2} & \mathbf{0}_{2\times2}\\[4pt]
\rho(k)\e^{2\ii t\theta} &\mathbb{I}_{2\times2}\end{pmatrix}.
\ee
\item Normalization: $\mu^{(1)}(\zeta,t;k)\rightarrow\mathbb{I}_{4\times4},$ as $k\rightarrow\infty.$
\end{itemize}
\end{rhp}
We open the contours at $k=0$ and define several regions and lines as shown in Figure \ref{fig6}, where
\berr
\Xi=\bigcup_{l=1}^4\Xi_l,\quad \Xi_l=\e^{\frac{(l-1)\ii\pi}{2}+\frac{\ii\pi}{4}}\alpha,\quad l=1,2,3,4,\ 0\leq\alpha<\infty.
\eerr
\begin{figure}[htbp]
  \centering
  \includegraphics[width=3in]{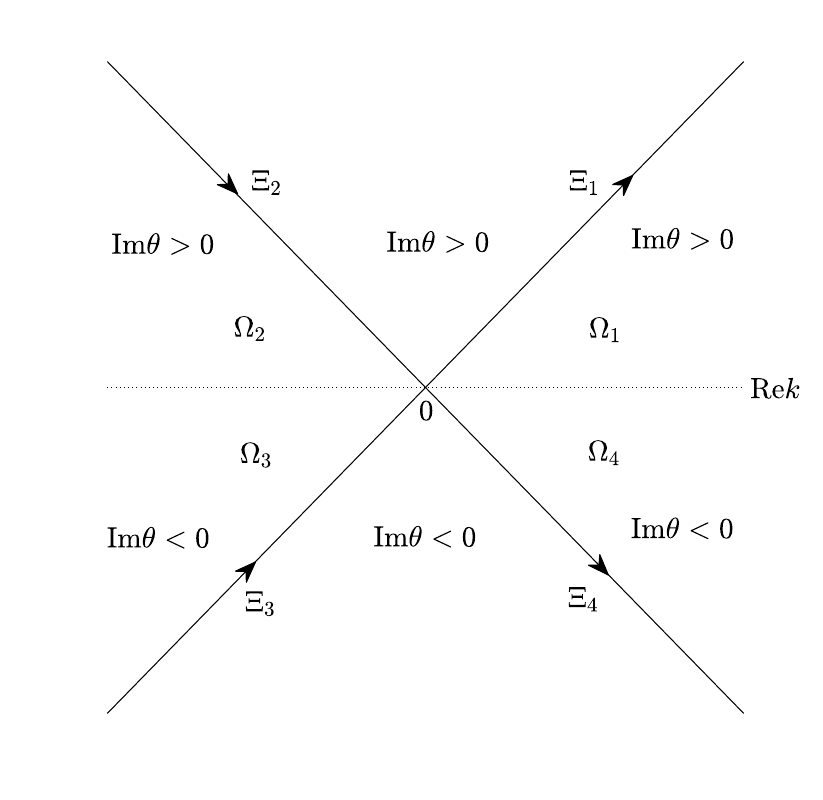}
  \caption{The signature of Im$\theta$ in the case $\hat{\zeta}>\varepsilon$ and the jump lines $\Xi_l$.}\label{fig6}
\end{figure}
Now, we open the jump line at $k=0$ to make a continuous extension, and the first step is to introduce several new functions.
\begin{lemma}\label{lem4.1}
Define functions $R_j:\bar{\Omega}_j\mapsto\bfC$, $j=1,2,3,4$ with boundary values satisfying
\begin{align}
R_1(k)&=\left\{
\begin{aligned}
&-\rho(k),\quad k>0,\\
&-\rho(0),\quad k\in\Xi_1,
\end{aligned}
\right.\quad
R_2(k)=\left\{
\begin{aligned}
&-\rho(k),\quad k<0,\\
&-\rho(0),\quad k\in\Xi_2,
\end{aligned}
\right.\label{4.5}\\
R_3(k)&=\left\{
\begin{aligned}
&\rho^\dag(k),\quad k<0,\\
&\rho^\dag(0),\quad k\in\Xi_3,
\end{aligned}
\right.\quad
R_4(k)=\left\{
\begin{aligned}
&\rho^\dag(k),\quad k>0,\\
&\rho^\dag(0), \quad k\in\Xi_4.
\end{aligned}
\right.
\end{align}
Moreover, $R_j$ admit estimates
\be\label{4.7}
|\bar{\partial}R_j(k)|\leq c_1|\rho'(\text{Re}k)|+c_2|k|^{-1/2},
\ee
for positive constants $c_1$ and $c_2$ depended on $\|\rho\|_{H^1(\bfR)}$.
\end{lemma}
\begin{proof}
We only prove the lemma for $R_1$ on $\bar{\Omega}_1$. Let $k=s\e^{\ii\phi}$. Define the interpolation
\begin{align}
R_1(z)=-\rho(0)+\left[-\rho(\text{Re} k)+\rho(0)\right]\cos(2\phi).
\end{align}
 Then we find
\begin{align}
\bar{\partial}R_4(z)=-\frac{1}{2}\rho'(\text{Re}k)\cos(2\phi)-
\frac{\ii}{2}\e^{\ii\phi}\frac{-\rho(\text{Re}k)+\rho(0)}
{|k|}\sin(2\phi).
\end{align}
A basic estimate shows that \eqref{4.7} holds.
\end{proof}
Now we introduce a $4\times4$ matrix function by the following transformation
\be\label{4.10}
\mu^{(2)}(\zeta,t;k)=\mu^{(1)}(\zeta,t;k)\mathcal{R}^{(2)}(k),
\ee
where
\be
\mathcal{R}^{(2)}(k)=\left\{
\begin{aligned}
&\begin{pmatrix}
\mathbb{I}_{2\times2} & \mathbf{0}_{2\times2}\\
R_j(k)\e^{2\ii t\theta(k)} & \mathbb{I}_{2\times2}
\end{pmatrix},\,\,\  k\in\Omega_j,\,j=1,2,\\
&\begin{pmatrix}
\mathbb{I}_{2\times2} & R_j(k)\e^{-2\ii t\theta(k)}\\
\mathbf{0}_{2\times2} & \mathbb{I}_{2\times2}
\end{pmatrix},\,  k\in\Omega_j,\,j=3,4,\\
&\mathbb{I}_{4\times4},\qquad\qquad\qquad\quad\quad\,\,\,\ \text{elsewhere}.
\end{aligned}
\right.
\ee
Then $\mu^{(2)}(\zeta,t;k)$ is the solution of a mixed $\bar{\partial}$-RH problem as follows:
\begin{dbarrhp}\label{rhp4.1}
Find a $4\times4$ matrix-valued function $\mu^{(2)}(\zeta,t;k)$ with the following properties:
\begin{itemize}
\item Analyticity: $\mu^{(2)}(\zeta,t;k)$ is continuous with sectionally continuous first partial derivatives in $\bfC\setminus\Xi$.

\item Jump condition: The continuous boundary values $\mu^{(2)}(\zeta,t;k)$ satisfy
$\mu^{(2)}_+(\zeta,t;k)=\mu^{(2)}_-(\zeta,t;k)J^{(2)}(\zeta,t;k)$ across $\Xi$ with
\be\label{3.25}
J^{(2)}(\zeta,t;k)=\left\{
\begin{aligned}
&\begin{pmatrix}
\mathbb{I}_{2\times2} & \mathbf{0}_{2\times2}\\
\rho(0)\e^{2\ii t\theta} & \mathbb{I}_{2\times2}
\end{pmatrix},\qquad k\in \Xi_1\cup\Xi_2,\\
&\begin{pmatrix}
\mathbb{I}_{2\times2} & \rho^\dag(0)\e^{-2\ii t\theta}\\
\mathbf{0}_{2\times2} &\mathbb{I}_{2\times2}
\end{pmatrix},\quad k\in\Xi_3\cup\Xi_4.
\end{aligned}
\right.
\ee

\item $\bar{\partial}$-Derivative: For $k\in\bfC\setminus \Xi$, we have
\be
\bar{\partial}\mu^{(2)}(\zeta,t;k)=\mu^{(2)}(\zeta,t;k)\bar{\partial}\mathcal{R}^{(2)}(k).
\ee

\item Normalization: As $k\rightarrow\infty$,  $\mu^{(2)}(\zeta,t;k)\rightarrow\mathbb{I}_{4\times4}.$
\end{itemize}
\end{dbarrhp}
Denote $\mu^{(0)}(\zeta,t;k)$ be the solution of the RH problem by dropping the $\bar{\partial}$ component, namely, letting $\bar{\partial}\mathcal{R}^{(2)}\equiv0$ in $\bar{\partial}$-Riemann--Hilbert problem \ref{rhp4.1}. We then have
\be\label{4.14}
\mu^{(0)}(\zeta,t;k)=\mathbb{I}_{4\times4}+O(\e^{-ct}),\quad t\to\infty.
\ee
Now we introduce
\be\label{4.15}
\mu^{(3)}(\zeta,t;k)=\mu^{(2)}(\zeta,t;k)[\mu^{(0)}(\zeta,t;k)]^{-1},
\ee
then we get the following pure $\bar{\partial}$-problem for $\mu^{(3)}(\zeta,t;k)$.
\begin{dbar}\label{dbar4.1}
Find a $4\times4$ matrix-valued function $\mu^{(3)}(\zeta,t;k)$ with the following properties:
\begin{itemize}
\item Analyticity: $\mu^{(3)}(\zeta,t;k)$ is continuous with sectionally continuous first partial derivatives in $\bfC\setminus\Xi$.

\item $\bar{\partial}$-Derivative: For $z\in\bfC\setminus\Xi$, we have
\be
\bar{\partial}\mu^{(3)}(\zeta,t;k)=\mu^{(3)}(\zeta,t;k)w^{(3)}(\zeta,t;k),
\ee
where
\be
w^{(3)}(\zeta,t;k)=\mu^{(0)}(\zeta,t;k)\bar{\partial}
\mathcal{R}^{(2)}(z)[\mu^{(0)}(\zeta,t;k)]^{-1}.
\ee

\item Normalization: As $k\rightarrow\infty$, $\mu^{(3)}(\zeta,t;k)\rightarrow\mathbb{I}_{4\times4}.$
\end{itemize}
\end{dbar}
We now proceed as in the previous section and study the integral equation related to the $\bar{\partial}$-problem \ref{dbar4.1}
\be
\mu^{(3)}(\zeta,t;k)=\mathbb{I}_{4\times4}
-\frac{1}{\pi}\iint\limits_{\bfC}\frac{(\mu^{(3)}w^{(3)})(\zeta,t;s)}{s-k}\dd m(s).
\ee
Writing $s=u+\ii v$ and $k=\alpha+\ii\beta$, then region $D_1$ corresponds to $u\geq v\geq0$ and
\begin{align}
\text{Re}(2\ii t\theta(s))=-2tv\left(\hat{\zeta}+\frac{1}{4(u^2+v^2)}\right)\leq-\frac{vt}{2(u^2+v^2)}.
\end{align}
We estimate
\be
\iint\limits_{D_1}\frac{|w^{(3)}(s)|}{|s-k|}\dd m(s)\leq c(I_1+I_2),
\ee
where
\begin{align}
I_1=&\int_0^\infty\int_{v}^\infty\frac{1}{|s-k|}
\left|\rho'(\text{Re}s)\right|e^{-\frac{vt}{2(u^2+v^2)}}\dd u\dd v,\nn\\
I_2=&\int_0^\infty\int_{v}^\infty\frac{1}{|s-k|}\frac{1}{|u+\ii v|^{1/2}}e^{-\frac{vt}{2(u^2+v^2)}}\dd u\dd v.\nn
\end{align}
It follows from the estimates in Subsection \ref{sec3.5} that $|I_1|,|I_2|\leq ct^{-1/6}.$
This proves that
\be
\iint\limits_{D_1}\frac{|w^{(3)}(s)|}{|s-z|}\dd m(s)\leq ct^{-1/6},
\ee
which yields the following result.
\begin{proposition}
As $t\to\infty$, the norm of the integral operator $C_k$ given in \eqref{3.120} decays to zero:
\be
\|C_k\|_{L^\infty\rightarrow L^\infty}\leq ct^{-1/6},
\ee
which yields the existence of operator $(1-C_k)^{-1}$, and hence $\mu^{(3)}(\zeta,t;k)$.
\end{proposition}

We now consider the asymptotic expansion of $\mu^{(3)}(\zeta,t;k)$ at $k=0$
\be
\mu^{(3)}(\zeta,t;k)=\mathbb{I}_{4\times4}+\mu_0^{(3)}(\zeta,t)+\mu_1^{(3)}(\zeta,t)k+O(k^2),
\ee
where
\begin{align}
\mu_0^{(3)}(\zeta,t)=&-\frac{1}{\pi}\iint\limits_{\bfC}\frac{(\mu^{(3)}w^{(3)})(s)}{s}\dd m(s),\\
\mu_1^{(3)}(\zeta,t)=&-\frac{1}{\pi}\iint\limits_{\bfC}\frac{(\mu^{(3)}w^{(3)})(s)}{s^2}\dd m(s).
\end{align}
\begin{proposition}
As $t\to\infty$, $\mu_0^{(3)}(\zeta,t)$ and $\mu_1^{(3)}(\zeta,t)$ admit the following estimate
\be
|\mu_0^{(3)}(\zeta,t)|\leq ct^{-1},\quad |\mu_1^{(3)}(\zeta,t)|\leq ct^{-1}.
\ee
\end{proposition}

Inverting the sequence of transformations \eqref{4.2}, \eqref{4.10}, \eqref{4.14} and \eqref{4.15}, using \eqref{2.50} and \eqref{2.51} and taking $k\rightarrow0$ vertically, we have as $t\to\infty$
\begin{align}
q_1(x,t)&=q_1(\zeta(x,t),t)=O(t^{-1}),\\
q_2(x,t)&=q_2(\zeta(x,t),t)=O(t^{-1}),
\end{align}
and
\be
x=\zeta(x,t)+O(t^{-1}).
\ee

{\bf Acknowledgments.}
This work was supported by the National Natural Science Foundation of China under Grant No. 12301311 and Natural Science Foundation of Jiangsu Province under Grant No. BK20220434.

\appendix
\section{Proof of Proposition \ref{prop3.3}}\label{secA}
\begin{proof}
We first prove \eqref{3.92}. For a fixed small number $\beta$ with $0<2\beta<1$, we write
\begin{align*}\label{B.1}
&[\delta_{k_0}^1(z)]^2-z^{-2\ii\nu(k_0)}\e^{\frac{\ii z^2}{2}}\\
=&\e^{\ii\beta z^2}\Bigg\{z^{-2\ii\nu(k_0)}\exp\left\{\frac{\ii(1-2\beta) z^2}{2}\left(1-\frac{z}{(1-2\beta)s^4k_0^{-9/2}t^{1/2}}\right)\right\}
\left(\frac{2k_0}{z/\sqrt{k_0^{-3}t}+2k_0}\right)^{-2\ii\nu(k_0)}\\
&\times\e^{2[\chi_{k_0}([z/\sqrt{k_0^{-3}t}]+k_0)-\chi_{k_0}(k_0)]}-z^{-2\ii\nu(k_0)}\e^{\frac{\ii(1-2\beta) z^2}{2}}\Bigg\}\nn\\
=&\e^{\frac{\ii\beta z^2}{2}}(I+II+III),
\end{align*}
where
\begin{align}
I=&\e^{\frac{\ii\beta z^2}{2}}z^{-2\ii\nu(k_0)}\left[
\exp\left\{\frac{\ii(1-2\beta)z^2}{2}\left(1-\frac{z}{(1-2\beta)s^4k_0^{-9/2}t^{1/2}}\right)
\right\}-\e^{\frac{\ii(1-2\beta) z^2}{2}}\right],\nn\\
II=&\e^{\frac{\ii\beta z^2}{2}}z^{-2\ii\nu(k_0)}
\exp\left\{\frac{\ii(1-2\beta)z^2}{2}\left(1-\frac{z}{(1-2\beta)s^4k_0^{-9/2}t^{1/2}}\right)
\right\}\nn\\
&\times\left[\left(\frac{2k_0}{z/\sqrt{k_0^{-3}t}+2k_0}\right)^{-2\ii\nu(k_0)}-1\right],\nn\\
III=&\e^{\frac{\ii\beta z^2}{2}}z^{-2\ii\nu(k_0)}
\exp\left\{\frac{\ii(1-2\beta)z^2}{2}\left(1-\frac{z}{(1-2\beta)s^4k_0^{-9/2}t^{1/2}}\right)
\right\}\left(\frac{2k_0}{z/\sqrt{k_0^{-3}t}+2k_0}\right)^{-2\ii\nu(k_0)}\nn\\
&\times\left[\e^{2[\chi_{k_0}([z/\sqrt{k_0^{-3}t}]+k_0)-\chi_{k_0}(k_0)]}-1\right].\nn
\end{align}
Obviously, $|z^{-2\ii\nu(k_0)}|=\e^{2\nu(k_0)\arg z}$ and $|[2k_0/(z/\sqrt{k_0^{-3}t}+2k_0)]^{-2\ii\nu(k_0)}|
=\e^{2\nu(k_0)\arg(1+\alpha\e^{\pi\ii/4}/2)},$ thus they are bounded. Moreover,
$\e^{\frac{\ii(1-2\beta)z^2}{2}(1-z/((1-2\beta)s^4k_0^{-9/2}t^{1/2}))}$ is bounded as \berr
\text{Re}[1-z/((1-2\beta)s^4k_0^{-9/2}t^{1/2})]
=1-k_0^4\alpha\frac{\text{Re}s^4+\text{Im}s^4}{(1-2\beta)\sqrt{2}|s|^4}
\geq1-\frac{k_0^4\epsilon}{(1-2\beta)}>0
\eerr
for $\beta$ sufficiently small.
On the other hand, we have
\begin{align}
|I|\leq&c\left|\e^{\frac{\ii\beta z^2}{2}}\right|\left|\exp\left\{\frac{\ii(1-2\beta)z^2}{2}\left(1-\frac{ z}{(1-2\beta)s^4k_0^{-9/2}t^{1/2}}\right)\right\}-\e^{\frac{\ii (1-2\beta)z^2}{2}}\right|\\
\leq&c\left|\e^{\frac{\ii\beta z^2}{2}}\right|\sup_{0\leq \eta\leq1}\left|\frac{\dd}{\dd\eta}
\exp\left\{\frac{\ii z^2}{2}\left(1-\frac{\eta z}{s^4k_0^{-9/2}t^{1/2}}\right)\right\}\right|\leq\left|\e^{\frac{\ii\beta z^2}{2}}z^3\right|\frac{c}{\sqrt{t}}\leq\frac{c}{\sqrt{t}},\nn\\
|II|\leq&c\left|\e^{\frac{\ii\beta z^2}{2}}\right|\left|\left(\frac{2k_0}{z/\sqrt{k_0^{-3}t}+2k_0}\right)^{-2\ii\nu(k_0)}-1\right|
=c\left|\e^{\frac{\ii\beta z^2}{2}}\int_1^{1+\frac{z}{\sqrt{4k_0^{-1}t}}}2\ii\nu(k_0)\zeta^{2\ii\nu(k_0)-1}\dd\zeta\right|\\
\leq&c\left|\e^{\frac{\ii\beta z^2}{2}}z\right|\sup\left\{|\zeta^{2\ii\nu(k_0)-1}|:\zeta=1+\frac{\varsigma z}{\sqrt{4k_0^{-1}t}},0\leq \varsigma\leq1\right\}\leq \frac{c}{\sqrt{t}}.\nn
\end{align}
Next, we consider
\begin{align}
&\left|\e^{\frac{\ii\beta z^2}{2}}\right|\left|\e^{2[\chi_{k_0}([z/\sqrt{k_0^{-3}t}]+k_0)-\chi_{k_0}(k_0)]}-1\right|\\
\leq& c\sup_{0\leq\varsigma\leq1}\left|
\e^{2\varsigma[\chi_{k_0}([z/\sqrt{k_0^{-3}t}]+k_0)-\chi_{k_0}(k_0)]}\right|
\left|2\e^{\frac{\ii\beta z^2}{2}}\left[\chi_{k_0}\left(\frac{z}{\sqrt{k_0^{-3}t}}+k_0\right)
-\chi_{k_0}(k_0)\right]\right|.\nn
\end{align}
It follows from \eqref{3.85} that
\begin{align}
\chi_{k_0}(k)=-\frac{1}{2\pi\ii}\left(\int_{-\infty}^{-k_0}+\int_{k_0}^{\infty}\right)
\ln(k-s)\dd\ln(1+\text{tr}[\rho(s)\rho^\dag(s)]+\det[\rho(s)\rho^\dag(s)]).\nn
\end{align}
Denote $g(s)=\partial_s\ln(1+\text{tr}[\rho(k_0s)\rho^\dag(k_0s)]+\det[\rho(k_0s)\rho^\dag(k_0s)])$, then, we find
\begin{align}
&-2\pi\ii\left(\chi_{k_0}\left(\frac{z}{\sqrt{k_0^{-3}t}}+k_0\right)-\chi_{k_0}(k_0)\right)\\
=&\int_{-\infty}^{-k_0}\ln\left(\frac{\frac{z}{\sqrt{k_0^{-3}t}}+k_0-s}{k_0-s}\right)
\dd\ln(1+\text{tr}[\rho(s)\rho^\dag(s)]+\det[\rho(s)\rho^\dag(s)])\nn\\
&+\int^{\infty}_{k_0}\ln\left(\frac{\frac{z}{\sqrt{k_0^{-3}t}}+k_0-s}{k_0-s}\right)
\dd\ln(1+\text{tr}[\rho(s)\rho^\dag(s)]+\det[\rho(s)\rho^\dag(s)])\nn\\
=&\int_{-\infty}^{-1}\ln\left(1+\frac{\frac{z}{\sqrt{k_0^{-3}t}}}{1-s}\right)
\dd\ln(1+\text{tr}[\rho(k_0s)\rho^\dag(k_0s)]+\det[\rho(k_0s)\rho^\dag(k_0s)])\nn\\
&+\int^{\infty}_{1}\ln\left(1-\frac{\frac{z}{\sqrt{k_0^{-3}t}}}{s-1}\right)
\dd\ln(1+\text{tr}[\rho(k_0s)\rho^\dag(k_0s)]+\det[\rho(k_0s)\rho^\dag(k_0s)])\nn\\
=&\left[\int_{-\infty}^{-1}(g(s)-g(1))\ln\left(1+\frac{\frac{z}{\sqrt{k_0^{-3}t}}}{1-s}\right)
\dd s+\int^{\infty}_{1}(g(s)-g(1))\ln\left(1-\frac{\frac{z}{\sqrt{k_0^{-3}t}}}{s-1}\right)
\dd s\right]\nn\\
&+\left[\int_{-\infty}^{-1}g(1)\ln\left(1+\frac{\frac{z}{\sqrt{k_0^{-3}t}}}{1-s}\right)\dd s+\int^{\infty}_{1}g(1)\ln\left(1-\frac{\frac{z}{\sqrt{k_0^{-3}t}}}{s-1}\right)\dd s\right]=
III_1+III_2.\nn
\end{align}
Using the Lipschitz condition, $|\ln(1+a)|\leq|a|$, we have
\begin{align}
\left|\e^{\frac{\ii\beta z^2}{2}}III_1\right|\leq c\left|\e^{\frac{\ii\beta z^2}{2}}\frac{z}{\sqrt{k_0^{-3}t}}\right|
\left(\int_{-\infty}^{-1}+\int_1^\infty\right)\left|\frac{g(s)-g(1)}{s-1}\right|\dd s\leq\frac{c}{\sqrt{t}},
\end{align}
since $g(s)$ is rapidly decay as $s\rightarrow\infty$.
Moreover, we have
\begin{align}
III_2=&\int_0^2g(1)\ln\left(1-\frac{z}{\sqrt{k_0^{-3}t}s}\right)\dd s+\int_2^\infty g(1)\ln\left(1-\frac{z^2}{k_0^{-3}t s^2}\right)\dd s\nn\\
=&\left(\int_0^1+\int_1^2\right)g(1)\ln\left(1-\frac{z}{\sqrt{k_0^{-3}t}s}\right)\dd s+\int_2^\infty g(1)\ln\left(1-\frac{z^2}{k_0^{-3}t s^2}\right)\dd s\nn\\
=&III_{2,1}+III_{2,2}+III_{2,3}.\nn
\end{align}
Thus, we have
\begin{align}
\left|\e^{\frac{\ii\beta z^2 }{2}}III_{2,2}\right|&\leq c\left|\e^{\frac{\ii\beta z^2}{2}}\frac{z}{\sqrt{k_0^{-3}t}}\right|\leq\frac{c}{\sqrt{t}},\\
\left|\e^{\frac{\ii\beta z^2}{2}}III_{2,3}\right|&\leq c\left|\e^{\frac{\ii\beta z^2}{2}}\frac{z^2}{k_0^{-3}t}\right|\leq\frac{c}{t}.
\end{align}
On the other hand, we have
\begin{align}
III_{2,1}=g(1)\left(1-\frac{z}{\sqrt{k_0^{-3}t}}\right)
\ln\left(1-\frac{z}{\sqrt{k_0^{-3}t}}\right)
+g(1)\frac{z}{\sqrt{k_0^{-3}t}}\ln\left(\frac{-z}{\sqrt{k_0^{-3}t}}\right).\nn
\end{align}
Thus,
\begin{align}
\left|\e^{\frac{\ii\beta z^2 }{2}}III_{2,1}\right|\leq&\frac{c\left|\e^{\frac{\ii\beta z^2 }{2}}z\right|}{\sqrt{k_0^{-3}t}}
\left(1-\frac{|z|}{\sqrt{k_0^{-3}t}}\right)+\frac{c\left|\e^{\frac{\ii\beta z^2}{2}}z\ln z\right|}{\sqrt{k_0^{-3}t}}+c\left|\e^{\frac{\ii\beta z^2}{2}}z\right|\frac{\ln\left(\sqrt{k_0^{-3}t}\right)}
{\sqrt{k_0^{-3}t}}
\leq c\frac{\ln t}{\sqrt{t}}.
\end{align}
Therefore, we prove that
\begin{align}
|III|\leq c\frac{\ln t}{\sqrt{t}}.
\end{align}

We now focus on the proof of case \eqref{3.93}, the case \eqref{3.94} is similar. Denote \be
\tilde{\delta}_1(k)=\left(\delta_{2}(k)-\delta(k) \mathbb{I}_{2\times2}\right)\e^{-2\ii t\theta(k)}.
\ee
It then follows from \eqref{3.4} and \eqref{3.81} that $\tilde{\delta}_1(k)$ satisfies
\begin{align}
\left\{
\begin{aligned}
\tilde{\delta}_{1+}(k)&=\left(1+\text{tr}[\rho(k)\rho^\dag(k)]
+\det[\rho(k)\rho^\dag(k)]\right)\tilde{\delta}_{1-}(k)+f(k)\e^{-2\ii t\theta(k)},\quad  |k|>k_0,\\
&=\tilde{\delta}_{1-}(k),\qquad\qquad\qquad\qquad\quad\quad\quad\,\
\qquad\qquad\qquad\qquad\qquad\qquad |k|<k_0,\\
\tilde{\delta}_1(k)&\to\mathbf{0}_{2\times2},\qquad\qquad\qquad\qquad\
\qquad\qquad\qquad\qquad\quad\qquad\qquad\qquad\quad k\to\infty,
\end{aligned}
\right.
\end{align}
where
\be\label{A.12}
f(k)=\left[\rho(k)\rho^\dag(k)-\left(\text{tr}[\rho(k)\rho^\dag(k)]
+\det[\rho(k)\rho^\dag(k)]\right)\mathbb{I}_{2\times2}\right]\delta_{2-}(k).
\ee
In terms of Plemelj formula, the function $\tilde{\delta}_1(k)$ can be represented as follows:
\begin{equation}
\left\{
\begin{aligned}
\tilde{\delta}_1(k)&=X(k)\int_{\bfR\setminus[-k_0,k_0]}
\frac{\e^{-2\ii t\theta(s)}f(s)}{X_+(s)(s-k)}\dd s,\\
X(k)&=\exp\left\{{\frac{1}{2\pi \ii}\int_{\bfR\setminus[-k_0,k_0]}\frac{\ln\left(1+\text{tr}[\rho(s)\rho^\dag(s)]
+\det[\rho(s)\rho^\dag(s)]\right)}{s-k}\dd s }\right\}.
\end{aligned}
\right.
\end{equation}
Note that $\delta_{2-}(k)$ is only continuous on $\mathbb{R}$. Thus, we denote
\be
\tilde{f}(k)=\rho(k)\rho^\dag(k)-\left(\text{tr}[\rho(k)\rho^\dag(k)]
+\det[\rho(k)\rho^\dag(k)]\right)\mathbb{I}_{2\times2}.
\ee
By the analytic decomposition method for scattering data in \cite{PD}, we conclude that $\tilde{f}(k)$ can be decomposed into $\tilde{f}(k)=\tilde{f}_1(k)+\tilde{f}_2(k)+\tilde{f}_3(k)$, where $\tilde{f}_1(k)$ is only defined on $\bfR$, $\tilde{f}_2(k)$ has an analytic continuation to $L_t$
\begin{align*}
L_t:=&L_{t1}\cup L_{t2}=\left\{k\in\bfC|k=k_0+\frac{1}{t}+k_0\alpha\e^{\frac{\pi\ii}{4}},
0<\alpha<\infty\right\}\\
&\cup\left\{k\in\bfC|k=-k_0-\frac{1}{t}+k_0\alpha\e^{\frac{\pi\ii}{4}},
-\infty<\alpha<0\right\},
\end{align*}
and $\tilde{f}_3(k)$ is rational function. Moreover, we have
\begin{align}
\left|\e^{-2\ii t\theta(k)}\tilde{f}_1(k)\right|\leq& \frac{c}{(1+|k|^2)t^l},\quad k\in\bfR,\\
\left|\e^{-2\ii t\theta(k)}\tilde{f}_2(k)\right|\leq& \frac{c}{(1+|k|^2)t^l},\quad k\in L_t,\\
\left|\tilde{f}_3(k)\right|\leq& \frac{c}{1+|k|^5},
\end{align}
for arbitrary natural number $l$. Thus, for $z\in\{z\in\bfC|z=\sqrt{k^{-1}_0t}\alpha\e^{\frac{\pi\ii}{4}},
-\epsilon\leq\alpha\leq\epsilon\}$, we find
\begin{align}
\left(\mathcal{N}\tilde{\delta}_1\right)(k)\leq&cX\left(\frac{z}{\sqrt{k_0^{-3}t}}+k_0\right)
\left(\int_{-k_0-\frac{1}{t}}^{-k_0}+\int_{k_0}^{k_0+\frac{1}{t}}\right)
\frac{\e^{-2\ii t\theta(s)}\tilde{f}(s)}{X_+(s)\left(s-\frac{z}{\sqrt{k_0^{-3}t}}-k_0\right)}\dd s\\
&+cX\left(\frac{z}{\sqrt{k_0^{-3}t}}+k_0\right)
\left(\int^{-k_0-\frac{1}{t}}_{-\infty}+\int^{\infty}_{k_0+\frac{1}{t}}\right)
\frac{\e^{-2\ii t\theta(s)}\tilde{f}_1(s)}{X_+(s)\left(s-\frac{z}{\sqrt{k_0^{-3}t}}-k_0\right)}\dd s\nn\\
&+cX\left(\frac{z}{\sqrt{k_0^{-3}t}}+k_0\right)
\left(\int^{-k_0-\frac{1}{t}}_{-\infty}+\int^{\infty}_{k_0+\frac{1}{t}}\right)
\frac{\e^{-2\ii t\theta(s)}\tilde{f}_2(s)}{X_+(s)\left(s-\frac{z}{\sqrt{k_0^{-3}t}}-k_0\right)}\dd s\nn\\
&+cX\left(\frac{z}{\sqrt{k_0^{-3}t}}+k_0\right)
\left(\int^{-k_0-\frac{1}{t}}_{-\infty}+\int^{\infty}_{k_0+\frac{1}{t}}\right)
\frac{\e^{-2\ii t\theta(s)}\tilde{f}_3(s)}{X_+(s)\left(s-\frac{z}{\sqrt{k_0^{-3}t}}-k_0\right)}\dd s,\nn\\
:=&B_1+B_2+B_3+B_4.\nn
\end{align}
Then, it follows form \cite{GLL-JDE} that
\be
|B_1|\leq ct^{-1},\ |B_2|\leq ct^{-l+1},\ |B_3|\leq ct^{-l+1},\ |B_4|\leq ct^{-1}.
\ee
\end{proof}

\section{The parabolic cylinder model RH problem}\label{secB}
\begin{figure}[htbp]
  \centering
  \includegraphics[width=3in]{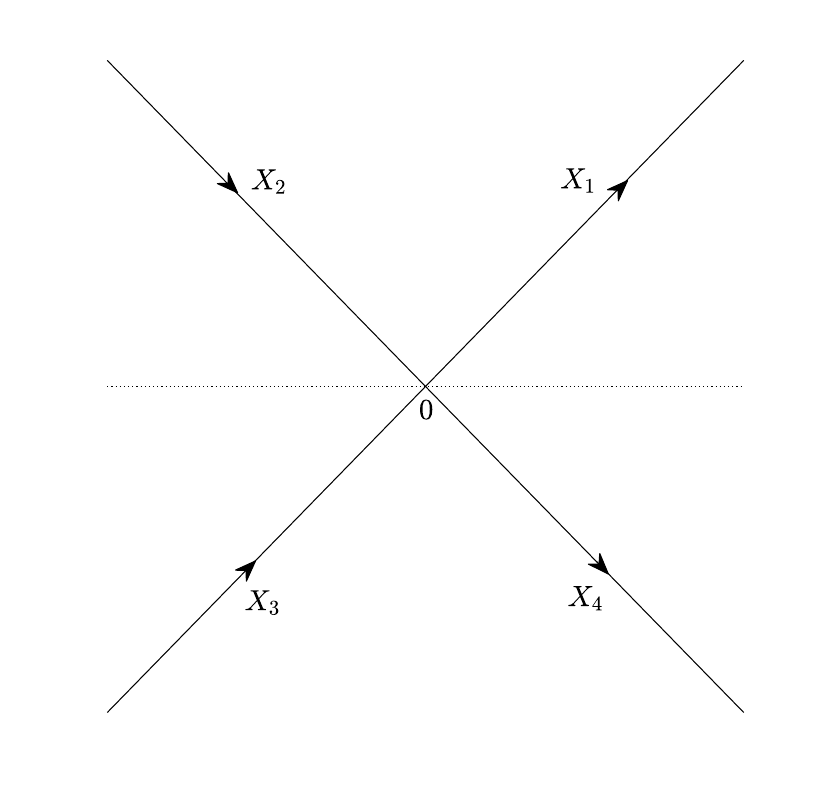}
  \caption{The contour $X=X_1\cup X_2\cup X_3\cup X_4$.}\label{fig7}
\end{figure}

Define the contour $X=X_1\cup X_2\cup X_3\cup X_4\subset\bfC$, where
\be\label{B.1}
\begin{aligned}
X_1&=\{\kappa\e^{\frac{\ii\pi}{4}}|0\leq \kappa<\infty\},\,\,\,\,\ X_2=\{\kappa\e^{\frac{3\ii\pi}{4}}|0\leq\kappa<\infty\},\\
X_3&=\{\kappa\e^{-\frac{3\ii\pi}{4}}|0\leq \kappa<\infty\},\, X_4=\{\kappa\e^{-\frac{\ii\pi}{4}}|0\leq\kappa<\infty\},
\end{aligned}
\ee
and oriented as in Figure \ref{fig7}. For a $2\times2$ complex-valued matrix $q$, define the function $\nu$ by $\nu(q)=-\frac{1}{2\pi}\ln(1+\text{tr}[qq^\dag]+\det[qq^\dag ])$. We consider the following parabolic cylinder model RH problem.
\begin{rhp}\label{rhB.1}
Find a $4\times4$ matrix-valued function $\mu^{(PC)}(q;z)$ on $\bfC\setminus X$ with the following properties:
\begin{itemize}
\item
 Analyticity: $\mu^{(PC)}(q;z)$ is analytic for $z\in\bfC\setminus X$ and extends continuously to $X$.
\item Jump condition:  The jump relation of the continuous boundary values $\mu^{(PC)}_\pm$ on $X$ is
\be\label{B.2}
\mu^{(PC)}_+(q;z)=\mu^{(PC)}_-(q;z)J^{(PC)}(q;z),\quad z\in X,
\ee
where
\be\label{B.3}
J^{(PC)}(q;z)=\left\{
\begin{aligned}
&\begin{pmatrix}
\mathbb{I}_{2\times2} & (\mathbb{I}_{2\times2}+q^\dag q)^{-1}q^\dag\e^{\frac{\ii z^2}{2}}z^{-2\ii\nu(q)}\\
\textbf{0}_{2\times2} & \mathbb{I}_{2\times2}
\end{pmatrix}, \,\ z\in X_1,\\
&\begin{pmatrix}
\mathbb{I}_{2\times2} & \textbf{0}_{2\times2}\\
q\e^{-\frac{\ii z^2}{2}}z^{2\ii\nu(q)} & \mathbb{I}_{2\times2}
\end{pmatrix},\qquad\qquad\qquad\quad\,\ z\in X_2,\\
&\begin{pmatrix}
\mathbb{I}_{2\times2} & q^\dag\e^{\frac{\ii z^2}{2}}z^{-2\ii\nu(q)}\\
\textbf{0}_{2\times2} & \mathbb{I}_{2\times2}
\end{pmatrix},\qquad\qquad\qquad\,\,\,\,\ z\in X_3,\\
&\begin{pmatrix}
\mathbb{I}_{2\times2} & \textbf{0}_{2\times2}\\
q(\mathbb{I}_{2\times2}+q^\dag q)^{-1}\e^{-\frac{\ii z^2}{2}}z^{2\ii\nu(q)} & \mathbb{I}_{2\times2}
\end{pmatrix},\quad z\in X_4.
\end{aligned}
\right.
\ee
\item Normalization: $\mu^{(PC)}(q;z)\rightarrow \mathbb{I}_{4\times4}$, as $z\rightarrow\infty$.
\end{itemize}
\end{rhp}
\begin{theorem}\label{thB.1}
The RH problem \ref{rhB.1} has a unique solution $\mu^{(PC)}(q;z)$ for each $2\times2$ matrix $q$, and this solution satisfies
\be\label{B.4}
\mu^{(PC)}(q;z)=\mathbb{I}_{4\times4}+\frac{\mu^{(PC)}_1(q)}{z}+O\left(z^{-2}\right),\quad z\rightarrow\infty,
\ee
where
\be\label{B.5}
\begin{aligned}
&\mu^{(PC)}_1(q)=\begin{pmatrix}\mathbf{0}_{2\times2}&\ii\beta^{(PC)}\\ -\ii\left(\beta^{(PC)}\right)^\dag&\mathbf{0}_{2\times2}\end{pmatrix},\\
&\beta^{(PC)}=\frac{\sqrt{2\pi}\e^{\frac{3\ii\pi}{4}-\frac{\pi\nu(q)}{2}}}
{\Gamma(\ii\nu(q))\det q}\begin{pmatrix}
q_{22} & -q_{12}\\ -q_{21} &  q_{11}
\end{pmatrix},
\end{aligned}
\ee
where $\Gamma(\cdot)$ denotes the standard Gamma function.
\end{theorem}
\begin{proof}
In this part, we address the model RH problem for $\mu^{(PC)}(q;z)$ and derive an explicit expression for $\mu^{(PC)}_1(q)$ in terms of the standard parabolic cylinder functions. To begin with, we introduce the following transformation
\begin{equation}\label{b.6}
	\Psi(z) = \mu^{(PC)}(z)z^{-\ii\nu\Sigma_3}\e^{\frac{\ii z}{4}\Sigma_3},
\end{equation}
which implies that
\begin{equation}
	\Psi_+(z) = \Psi_-(z)G(q),\quad G(q) = \e^{-\frac{\ii z}{4}\hat{\Sigma}_3}z^{\ii\nu\hat{\Sigma}_3}J^{(PC)}(q;z).
\end{equation}
Since the jump matrix $G(q)$ is independent of $z$ along each ray, it follows that
\begin{equation}
	\frac{\text{d}\Psi_+(z)}{\text{d}z} = \frac{\text{d}\Psi_-(z)}{\text{d}z}G(q).
\end{equation}
In addition, by applying the transformation \eqref{b.6} and extending $\mu^{(PC)}(q;z)$ as given in \eqref{B.4}, we get
\begin{equation}
	\frac{\text{d}\Psi(z)}{\text{d}z}\Psi^{-1}(z) =  \frac{\ii}{2}z\Sigma_3 + \frac{\ii}{2}[\mu_1^{(PC)},\Sigma_3] +O(z^{-1}).
\end{equation}
According to Liouville's theorem, we conclude that
\begin{equation}\label{b.10}
	\frac{\text{d}\Psi(z)}{\text{d}z} -\frac{\ii}{2}z\Sigma_3\Psi(z) = \beta\Psi(z),
\end{equation}
where
\begin{equation}
	\beta = \frac{\ii}{2}[\mu_1^{(PC)},\Sigma_3] = \begin{pmatrix}
		\mathbf{0}_{2\times 2} & \beta_{12} \\
		\beta_{21} &\mathbf{0}_{2\times 2}
	\end{pmatrix}.
\end{equation}
In particular, we get
\begin{equation}
	(\mu_1^{(PC)})_{12} = \ii\beta_{12},\quad (\mu_1^{(PC)})_{21} = -\ii\beta_{21}.
\end{equation}
Meanwhile, notice that $\mu^{(PC)}(q,z)$ satisfies the symmetry relation
\begin{equation}
	\Sigma_3\left[\mu^{(PC)}(z^*)\right]^{\dagger}\Sigma_3 = \left[\mu^{(PC)}(z)\right]^{-1},
\end{equation}
which further yields that
\begin{equation}
	\beta_{12} = \beta_{21}^{\dagger}.
\end{equation}

Now, we rewrite $\Psi(z)$ as a block matrix
\begin{equation*}
\Psi(z) = 	\begin{pmatrix}
		\Psi_{11}(z) & \Psi_{12}(z) \\
		\Psi_{21}(z) & \Psi_{22}(z)
	\end{pmatrix},
\end{equation*}
where $\Psi_{ij}, i, j=1,2$ are all $2\times 2$ matrices. In view of \eqref{b.10}, we arrive at
\begin{gather}
\label{b.15}	\frac{\text{d}^2\Psi_{11}(z)}{\text{d}z^2}+\left(\left(-\frac{\ii}{2}+\frac{z^2}{4}\right)\mathbb{I}_{2\times 2}-\beta_{12}\beta_{21}\right)\Psi_{11} = 0, \\
	\label{b.16}	\beta_{12}\Psi_{21}(z) = \frac{\text{d}\Psi_{11}(z)}{\text{d}z}-\frac{\ii}{2}z\Psi_{11}(z), \\
\label{b.17}	\frac{\text{d}^2\beta_{12}\Psi_{22}(z)}{\text{d}z^2}+\left(\left(\frac{\ii}{2}+\frac{z^2}{4}\right)\mathbb{I}_{2\times 2}-\beta_{12}\beta_{21}\right)\beta_{12}\Psi_{22}(z) = 0, \\
	\label{b.18}	\Psi_{12} = (\beta_{12}\beta_{21})^{-1}\left(\frac{\text{d}\beta_{12}\Psi_{22}(z)}{\text{d}z} +\frac{\ii}{2}z\beta_{12}\Psi_{22}(z)\right).
\end{gather}
Given that $\beta_{12}$ and $\beta_{21}$  are constant $2\times 2$ matrices independent of $z$, we express them as follows:
\begin{equation*}
	\beta_{12} = \begin{pmatrix}
		A & B \\
		C & D
	\end{pmatrix},\quad
	\beta_{12}\beta_{21} = \begin{pmatrix}
		\tilde{A} & \tilde{B} \\
		\tilde{C} & \tilde{D}
	\end{pmatrix}.
\end{equation*}
Set $\Psi_{11} = (\Psi_{11}^{ij})_{2\times 2}, i,j=1,2$. By examining the (1,1) entry and (2,1) entry of \eqref{b.15}, we obtain
\begin{equation}
	\begin{aligned}
		\frac{\text{d}^2\Psi_{11}^{(11)}(z)}{\text{d}z^2}+\left(-\frac{\ii}{2}+\frac{z^2}{4}\right)\Psi_{11}^{(11)}(z) -\tilde{A}\Psi_{11}^{(11)}(z) -\tilde{B}\Psi_{11}^{(21)}(z) = 0, \\
		\frac{\text{d}^2\Psi_{11}^{(21)}(z)}{\text{d}z^2}+\left(-\frac{\ii}{2}+\frac{z^2}{4}\right)\Psi_{21}^{(11)}(z) -\tilde{C}\Psi_{11}^{(11)}(z) -\tilde{D}\Psi_{11}^{(21)}(z) = 0.
	\end{aligned}
\end{equation}
Let $s$ be a constant satisfying $\tilde{B}\tilde{C}= (s-\tilde{D})(s-\tilde{A})$. Then, a straightforward computation yields
\begin{equation}\label{b.20}
	\frac{\text{d}^2\left[\tilde{C}\Psi_{11}^{(11)}(z)+(s-\tilde{A})\Psi_{11}^{(21)}(z)\right]}{\text{d}z^2}+\left(-\frac{\ii}{2}+\frac{z^2}{4}-s\right)\left[\tilde{C}\Psi_{11}^{(11)}(z)+(s-\tilde{A})\Psi_{11}^{(21)}(z)\right]=0.
\end{equation}
 By substituting variables, \eqref{b.20} can be transformed into the well known Weber equation:
 \begin{equation}
 	\frac{\text{d}^2g(\xi)}{\text{d}\xi^2}+\left(a+\frac{1}{2}-\frac{\xi^2}{4}\right)g(\xi) = 0,
 \end{equation}
which has two linear independent solutions $D_a(\xi)$ and $D_a(-\xi)$. Therefore, there exists two constants $c_1$ and $c_2$ such that
\begin{equation}
	g(\xi) = c_1D_a(\xi)+c_2D_a(-\xi),
\end{equation}
where $D_a(\cdot)$ is the standard parabolic cylinder function. Then, letting $a=-\ii s$ and $\xi = \e^{-\frac{\pi\ii}{4}}z$, \eqref{b.20} can be rewritten as
\begin{equation}\label{b.23}
	\tilde{C}\Psi_{11}^{(11)}(z)+(s-\tilde{A})\Psi_{11}^{(21)}(z) = c_1D_a(\e^{-\frac{\pi\ii}{4}}z) +c_2D_a(\e^{\frac{3\pi\ii}{4}}z).
\end{equation}
 Furthermore, by \cite{ET1927}, as $\xi\to\infty$,
 \begin{equation}\label{b.24}
 		D_a(\xi)=\begin{cases}
 		\xi^a\e^{-\frac{\xi^2}{4}}(1+O(\xi^{-2})),&|\arg \xi|<\frac{3\pi}{4},\\
 		\xi^a\e^{-\frac{\xi^2}{4}}(1+O(\xi^{-2}))
 		-\frac{\sqrt{2\pi}}{\Gamma(-a)}\e^{a\pi \ii}\xi^{-a-1}\e^{\frac{\xi^2}{4}}(1+O(\xi^{-2})),&\frac{\pi}{4}<\arg \xi< \frac{5\pi}{4},\\
 		\xi^a\e^{-\frac{\xi^2}{4}}(1+O(\xi^{-2}))
 		-\frac{\sqrt{2\pi}}{\Gamma(-a)}\e^{-a\pi \ii}\xi^{-a-1}\e^{\frac{\xi^2}{4}}(1+O(\xi^{-2})),&-\frac{5\pi}{4} <\arg \xi<-\frac{\pi}{4}.
 	\end{cases}
 \end{equation}
Note that as $z\to\infty$,
\begin{equation}\label{b.25}
	\Psi_{11}(z)\to z^{-\ii\nu}\e^{\frac{\ii z^2}{4}}\mathbb{I}_{2\times 2}, \quad \Psi_{22}(z)\to z^{\ii\nu}\e^{-\frac{\ii z^2}{4}}\mathbb{I}_{2\times 2}.
\end{equation}
Based on the asymptotic expansion given in \eqref{b.25} and the expressions in \eqref{b.23} and \eqref{b.24}, along the line $z = \sigma\e^{\frac{\pi\ii}{4}}$ with $\sigma > 0$, we can deduce that $c_1 = \tilde{C}z^{-\ii\nu-a}\e^{\frac{a\pi\ii}{4}}$ and $c_2 = 0$. Moreover, it is evident that $\Psi_{11}^{(11)}(z)$ and $\Psi_{11}^{(21)}(z)$ both satisfy the asymptotic expansion \eqref{b.25}, indicating that they are linearly independent. Consequently, it follows from \eqref{b.23} that $s$ must be unique. By the definition of
 $s$, we obtain $\tilde{C} = \tilde{B} = 0$. Thus, $\beta_{12}\beta_{21}= \text{diag}\{\tilde{A},\tilde{D}\}$, and \eqref{b.15} simplifies accordingly
 \begin{equation}
 	\frac{\text{d}^2}{\text{d}z^2}\begin{pmatrix}
 		\Psi_{11}^{(11)} & \Psi_{11}^{(12)}\\
 		\Psi_{11}^{(21)} & \Psi_{11}^{(22)}
  		 	\end{pmatrix}
  	+\left(-\frac{\ii}{2}+\frac{z^2}{4}\right)\begin{pmatrix}
  		\Psi_{11}^{(11)} & \Psi_{11}^{(12)}\\
  		\Psi_{11}^{(21)} & \Psi_{11}^{(22)}
  	\end{pmatrix}
  	-\begin{pmatrix}
  			\tilde{A}\Psi_{11}^{(11)} & \tilde{A}\Psi_{11}^{(12)}\\
  		\tilde{D}\Psi_{11}^{(21)} & \tilde{D}\Psi_{11}^{(22)}
  	\end{pmatrix}= 0.
 \end{equation}
Since $\Psi_{11}^{(11)},\Psi_{11}^{(12)}$ and $\Psi_{11}^{(21)},\Psi_{11}^{(22)}$ satisfy the same differential equation, by setting
\begin{equation}
	a_1 = -\ii\tilde{A},\quad a_2 = -\ii\tilde{D},
\end{equation}
similar to \eqref{b.23}, $\Psi_{11}^{(11)}$ and $\Psi_{11}^{(22)}$ can be expressed linearly in terms of
$D_{a_1}(\e^{-\frac{\pi\ii}{4}}z),D_{a_1}(\e^{\frac{3\pi\ii}{4}}z)$ and $D_{a_2}(\e^{-\frac{\pi\ii}{4}}z),D_{a_2}(\e^{\frac{3\pi\ii}{4}}z)$. On the other hand, in view of \eqref{b.25}, we obtain
\begin{equation}
	\Psi_{11}^{(12)}(z)\to 0,\quad \Psi_{11}^{(21)}(z)\to 0,\quad z\to\infty.
\end{equation}
Hence, form \eqref{b.15} and \eqref{b.17}, we have
\begin{equation}
	\begin{aligned}
		\Psi_{11}^{(11)}(z) &= c_1^{(1)}D_{a_1}(\e^{-\frac{\pi\ii}{4}}z)+ c_2^{(1)}D_{a_1}(\e^{\frac{3\pi\ii}{4}}z), \\
		\Psi_{11}^{(22)}(z) &= c_1^{(2)}D_{a_1}(\e^{-\frac{\pi\ii}{4}}z)+ c_2^{(2)}D_{a_1}(\e^{\frac{3\pi\ii}{4}}z), \\
		\tilde{A}\Psi_{22}^{(11)}(z) &= c_1^{(3)}D_{a_1}(\e^{\frac{\pi\ii}{4}}z)+ c_2^{(3)}D_{a_1}(\e^{-\frac{3\pi\ii}{4}}z),\\
		\tilde{A}\Psi_{22}^{(22)}(z) &= c_1^{(4)}D_{a_1}(\e^{\frac{\pi\ii}{4}}z)+ c_2^{(4)}D_{a_1}(\e^{-\frac{3\pi\ii}{4}}z),
	\end{aligned}
\end{equation}
where $c_1^{(j)}$ and $c_2^{(j)}\ (j=1,2,3,4)$ are constants. Then, for $\arg z\in(\frac{3\pi}{4},\frac{5\pi}{4})$, we find that
\begin{equation}
	\begin{aligned}
		\Psi_{11}^{(11)}(z) &= \Psi_{11}^{(22)}(z) = \e^{-\frac{3\pi\nu}{4}}D_{a_1}(\e^{\frac{3\pi\ii}{4}}z),\quad a_1 = a_2 = -\ii\nu,\\
		\Psi_{22}^{(11)}(z) &= \Psi_{22}^{(22)}(z) = \e^{-\frac{3\pi\nu}{4}}D_{-a_1}(\e^{-\frac{3\pi\ii}{4}}z).
	\end{aligned}
\end{equation}
 From \eqref{b.16}, \eqref{b.18} and the properties of $D_a(\cdot)$
 \begin{equation}
 	\frac{\text{d}D_a(\xi)}{\text{d}\xi} + \frac{\xi}{2}D_a(\xi)- aD_{a-1}(\xi) = 0,
 \end{equation}
 we can infer that
\begin{align}
	\beta_{12}\Psi_{21}(z) &= \e^{\frac{3\pi}{4}(\ii-\nu)}a_1D_{a_1-1}(\e^{\frac{3\pi\ii}{4}}z)\mathbb{I}_{2\times 2}, \\
	\Psi_{12}(z) &=\beta_{12}\e^{-\frac{\pi}{4}(\ii+3\nu)}D_{-a_1-1}(\e^{-\frac{3\pi\ii}{4}}z).
\end{align}
Denote
\begin{equation}\label{b.34}
	\tilde{\Psi}(z) =\begin{pmatrix}
		\Psi_{11}(z) & \Psi_{12}(z) \\
		\beta_{12}\Psi_{21}(z) & \Psi_{22}(z)
	\end{pmatrix}.
\end{equation}
Therefore, by a similar computation, it is easy to get
\begin{equation}\label{b.35}
	\tilde{\Psi}(z)= \begin{cases}
		\begin{pmatrix}
			\e^{-\frac{3\pi\nu}{4}}D_{a_1}(\e^{\frac{3\pi\ii}{4}}z)\mathbb{I}_{2\times2} & \beta_{12}\e^{-\frac{\pi}{4}(\ii+3\nu)}D_{-a_1-1}(\e^{-\frac{3\pi\ii}{4}}z) \\
			 \e^{\frac{3\pi}{4}(\ii-\nu)}a_1D_{a_1-1}(\e^{\frac{3\pi\ii}{4}}z)\mathbb{I}_{2\times 2}& \e^{-\frac{3\pi\nu}{4}}D_{-a_1}(\e^{-\frac{3\pi\ii}{4}}z)\mathbb{I}_{2\times 2}
		\end{pmatrix},\quad &\arg z\in(\frac{3\pi}{4},\frac{5\pi}{4}), \\
		\begin{pmatrix}
			\e^{\frac{\pi\nu}{4}}D_{a_1}(\e^{-\frac{\pi\ii}{4}}z)\mathbb{I}_{2\times 2} & \beta_{12}\e^{-\frac{\pi}{4}(\ii+3\nu)}D_{-a_1-1}(\e^{-\frac{3\pi\ii}{4}}z) \\
			\e^{\frac{\pi}{4}(\nu- i)}a_1D_{a_1-1}(\e^{-\frac{\pi\ii}{4}}z)\mathbb{I}_{2\times 2} &\e^{-\frac{3\pi\nu}{4}}D_{-a_1}(\e^{-\frac{3\pi\ii}{4}}z)\mathbb{I}_{2\times 2}
		\end{pmatrix},\quad &\arg z\in (\frac{\pi}{4},\frac{3\pi}{4}),\\
	
		\begin{pmatrix}
			\e^{\frac{\pi\nu}{4}}D_{a_1}(\e^{-\frac{\pi\ii}{4}}z)\mathbb{I}_{2\times 2} & \beta_{12}\e^{\frac{\pi}{4}(3\ii+\nu)}D_{-a_1}(\e^{\frac{\pi\ii}{4}}z) \\
				\e^{\frac{\pi}{4}(\nu- i)}a_1D_{a_1-1}(\e^{-\frac{\pi\ii}{4}}z)\mathbb{I}_{2\times 2} &  \e^{\frac{\pi\nu}{4}}D_{-a_1-1}(\e^{\frac{\pi\ii}{4}}z)\mathbb{I}_{2\times 2}
		\end{pmatrix},\quad &\arg z\in(-\frac{\pi}{4},\frac{\pi}{4}),\\
	
		\begin{pmatrix}
			\e^{-\frac{3\pi\nu}{4}}D_{a_1}(\e^{\frac{3\pi\ii}{4}}z)\mathbb{I}_{2\times2} & \beta_{12}\e^{\frac{\pi}{4}(3\ii+\nu)}D_{-a_1}(\e^{\frac{\pi\ii}{4}}z) \\
			\e^{\frac{3\pi}{4}(\ii-\nu)}a_1D_{a_1-1}(\e^{\frac{3\pi\ii}{4}}z)\mathbb{I}_{2\times 2} &\e^{\frac{\pi\nu}{4}}D_{-a_1-1}(\e^{\frac{\pi\ii}{4}}z)\mathbb{I}_{2\times 2}
		\end{pmatrix},\quad &\arg z\in (-\frac{3\pi}{4},-\frac{\pi}{4}).
	\end{cases}
\end{equation}

 Along the ray $\arg z = \frac{3\pi}{4}$, we know that
 \begin{equation}
 	\Psi_+(z) = \Psi_-(z)\begin{pmatrix}
 		\mathbb{I}_{2\times2} & \textbf{0}_{2\times2}\\
 		q & \mathbb{I}_{2\times2}
 	\end{pmatrix}.
 \end{equation}
From \eqref{b.34} and \eqref{b.35}, we deduce that the (2,1) entry of the above RH problem satisfies
\begin{equation}\label{b.37}
	\e^{\frac{\pi}{4}(\nu- i)}a_1D_{a_1-1}(\e^{-\frac{\pi\ii}{4}}z)\mathbb{I}_{2\times 2} = \e^{\frac{3\pi}{4}(\ii-\nu)}a_1D_{a_1-1}(\e^{\frac{3\pi\ii}{4}}z)\mathbb{I}_{2\times 2} +\beta_{12}\e^{-\frac{3\pi\nu}{4}}D_{-a_1}(\e^{-\frac{3\pi\ii}{4}}z)q
\end{equation}
Meanwhile, from \cite{ET1927}, it follows that
\begin{equation}\label{b.38}
		D_a(z)=\frac{\Gamma(a+1)}{\sqrt{2\pi}}\left[\e^{\frac{a\pi \ii}{2}}D_{-a-1}(\ii z)+\e^{-\frac{a\pi \ii}{2}}D_{-a-1}(-\ii z)\right].
\end{equation}
Therefore,
\begin{equation}\label{b.39}
	D_{-a_1}(\e^{-\frac{3\pi\ii}{4}})=\frac{\Gamma(-a_1+1)}{\sqrt{2\pi}}\left[\e^{-\frac{a_1\pi \ii}{2}}D_{a_1-1}(\e^{-\frac{\pi\ii}{4}}z)+\e^{\frac{a_1\pi \ii}{2}}D_{a_1-1}(\e^{\frac{3\pi\ii}{4}} z)\right]
\end{equation}
By comparing the coefficients in \eqref{b.37} and \eqref{b.39}, we get
\begin{equation}
	\beta_{12} =\frac{\sqrt{2\pi}\e^{\frac{3\pi\ii}{4}-\frac{\pi\nu}{2}}}
	{\Gamma(\ii\nu)\det q}\begin{pmatrix}
		q_{22} & -q_{12}\\ -q_{21} &  q_{11}
	\end{pmatrix},
\end{equation}
which is related to \eqref{B.5} by setting $\beta^{(PC)} = \beta_{12}$.

\end{proof}

\medskip
\small{

}
\end{document}